\documentclass[11pt]{article}

\usepackage[english]{babel}
\usepackage[utf8x]{inputenc}
\usepackage{comment}
\usepackage{booktabs}
\usepackage{tabu}
\usepackage[T1]{fontenc}

\usepackage[letterpaper,top=1in,bottom=1in,left=1in,right=1in]{geometry}

\usepackage{sansmathfonts}
\usepackage[sfdefault]{biolinum}
\usepackage[T1]{fontenc}
\usepackage[dvipsnames]{xcolor}

\usepackage{amsmath}
\usepackage{amsthm}
\usepackage{amssymb}
\usepackage{graphicx}
\usepackage{amsthm}
\usepackage[noadjust]{cite}
\usepackage{mathtools}
\usepackage{graphicx}
\usepackage{array}
\usepackage{xspace}
\usepackage{float}
\usepackage{wrapfig}
\usepackage{thmtools, thm-restate}
\usepackage{enumitem}
\usepackage{subcaption}
\usepackage{afterpage}
\usepackage[boxruled, nofillcomment, linesnumbered, lined]{algorithm2e}

\theoremstyle{plain}
\newtheorem{theorem}{Theorem}[section]
\newtheorem{lemma}[theorem]{Lemma}
\newtheorem{corollary}[theorem]{Corollary}
\newtheorem{claim}[theorem]{Claim}

\theoremstyle{definition}
\newtheorem{definition}[theorem]{Definition}

\newtheorem{fact}[theorem]{Fact}

\theoremstyle{remark}
\newtheorem{remark}[theorem]{Remark}

\newcommand{\hl}[1]{{\em \color{NavyBlue} #1}}

\usepackage{hyperref}
\hypersetup{colorlinks=true, urlcolor=Blue, citecolor=Green, linkcolor=BrickRed, breaklinks, unicode}

\usepackage[capitalise, noabbrev]{cleveref}



\crefname{claim}{claim}{claims}

\newcommand{\CONGEST}{\mathsf{CONGEST}}

\newcommand{\Congest}{\mathsf{CONGEST}}

\newcommand{\EC}{{\color{red}E_C}}
\newcommand{\ER}{{\color{black}\boldsymbol{E_R}}}
\newcommand{\ES}{{\color{cyan}E_S}}

\title{$\tilde{\text{O}}$ptimal Distributed Maximum Flow Approximation\\ in Undirected Planar Graphs}

\author{%
 Yaseen Abd-Elhaleem\thanks{{\em Department of Computer Science, University of Haifa. Supported in part by the Israel Science Foundation grants No. 810/21 and 2829/25.}}\\
 \texttt{\href{mailto:yaseenuniacc@gmail.com}{yaseenuniacc@gmail.com}}\\
 \and Michal Dory\thanks{{\em Department of Computer Science, University of Haifa. Supported in part by the Israel Science Foundation grant No. 2829/25.}}\\
 \texttt{\href{mailto:mdory@ds.haifa.ac.il}{mdory@ds.haifa.ac.il}}\\
 \and Oren Weimann\thanks{{\em Department of Computer Science, University of Haifa. Supported in part by the Israel Science Foundation grant No. 365/26.}}\\
 \texttt{\href{mailto:oren@cs.haifa.ac.il}{oren@cs.haifa.ac.il}}
}

\date{}

\begin{document}
\maketitle

\thispagestyle{empty}

\begin{abstract}
Persistent efforts in recent years have been devoted to devising distributed algorithms for fundamental optimization problems in planar graphs. In particular,  
for Single-Source Shortest-Paths, there is an $\tilde O(D^2)$-rounds\footnote{The $\tilde{O}(\cdot)$ notation hides poly-logarithmic factors in $n$.} exact algorithm [Li, Parter STOC'19] for directed planar graphs, and an $\tilde {O}(D)$-rounds $(1+o(1))$-approximation algorithm [Rozhon, Grunau, Haeupler, Zuzic, Li STOC'22] for undirected planar graphs (where $D$ is the graph's hop-diameter). Recently [Abd-Elhaleem, Dory, Parter, Weimann PODC'25], a matching bound for the exact case was obtained for the Maximum $st$-Flow problem. Namely, an $\tilde O(D^2)$-rounds exact algorithm for directed planar graphs. However, for the approximate case, they give a $D\cdot n^{o(1)}$-rounds $(1-o(1))$-approximation algorithm for undirected planar graphs that works only for the special case where both $s$ and $t$ lie on the same face. 

In this paper, we remove the restriction that both $s$ and $t$ must lie on the same face (we also eliminate the $n^{o(1)}$ factor). Namely, we present the first distributed near-optimal $\tilde{O}(D)$-rounds $(1-o(1))$-approximation algorithm for Maximum $st$-Flow in general undirected planar graphs.
Our main technical contribution is a distributed implementation of the classical Reif's [SICOMP'83] centralized algorithm. This is achieved by a careful recursive incision procedure on the planar dual $G^*$ of the graph $G$. 
It is challenging, because we need to simulate dynamic changes (incisions) over the dual graph $G^*$, while we can only communicate over the input graph $G$.
\end{abstract}

\newpage
{\tableofcontents}
\thispagestyle{empty}

\newpage \setcounter{page}{1}

\section{Introduction}
The maximum $st$-flow problem asks for the maximum amount of flow that can be sent from a source vertex $s$ to a target vertex $t$ in a graph with edge-capacities.   
The problem is one of the most fundamental and well-studied optimization problems in theoretical computer science. 
It was first presented in 1954 by Harris and Ross~\cite{HarrisRoss} who used the graph to model the Soviet railway network (which happens to be a planar graph).  
The first algorithm for the problem was given in 1956 by Ford and Fulkerson~\cite{ff56} in their seminal paper (still taught in every elementary algorithms class). 
This paper is known for introducing the celebrated min-max theorem, stating that the value of the maximum $st$-flow is equivalent to that of the minimum $st$-cut. It is less known, that in the same paper, Ford and Fulkerson also gave the first algorithm for maximum $st$-flow in planar graphs where both $s$ and $t$ lie on the same face (called {\em $st$-planar} graphs). 

In the centralized setting, after decades of research, the maximum $st$-flow problem is by now almost optimally resolved. Namely, in general graphs there are by now {\em almost}-linear time (i.e., optimal up to a subpolynomial factor) algorithms, with the currently fastest being~\cite{ChenKLPGS22}. In planar graphs, there are by now {\em near}-linear time (i.e., optimal up to a logarithmic factor) algorithms, with the currently fastest being $O(n\log n)$ in directed planar graphs~\cite{BorradaileKlein} and $O(n\log \log n)$ in undirected planar graphs~\cite{Italiano}.

\medskip
\noindent
{\bf Distributed planar graph algorithms.}
In the distributed $\Congest$ setting \cite{peleg-book}, we are given an $n$-vertex graph $G$ with hop-diameter $D$. Communication on $G$ occurs in synchronous rounds, where in each round $O(\log n)$ bits can be sent along every edge.
The goal is to minimize the number of rounds required to solve various optimization problems on $G$. 
In this setting, the {\em exact} maximum $st$-flow problem for general graphs is still open with the state-of-the-art being $n^{1/2+o(1)}\cdot (\sqrt{n}+D)$ rounds \cite{mincostflow_Vos23}.\footnote{In fact,~\cite{mincostflow_Vos23} solve the more general problem of {\em minimum-cost $st$-flow}.}
The approximate version, however, is resolved (up to subpolynomial factors)~\cite{ghaffari2015near}. 
Namely, up to an $n^{o(1)}$ factor, there is a tight lower and upper bound of $(D+\sqrt{n})$-rounds for approximate maximum $st$-flow~\cite{SarmaHKKNPPW11,Elkin04_MSTlower_bound,PelegR99_MSTlower_bound}, as well as for many other fundamental problems such as approximate single-source shortest-paths (SSSP)~\cite{RozhonGHZL22_shortestpaths,GHSYZ22} and exact minimum spanning tree (MST)~\cite{GarayKP98, KuttenP98}. 
The graph used for these lower bounds \cite{Elkin04_MSTlower_bound,SarmaHKKNPPW11, PelegR99_MSTlower_bound} is sparse, of arboricity two, but not planar.
This means that planar graphs are arguably one of the most interesting graph families to which the lower bound does not apply and the only lower bound is the trivial $\Omega(D)$. 
Nevertheless, the maximum $st$-flow problem is far from being resolved in planar graphs as we now elaborate. 

Near-optimal $\tilde{O}(D)$-round distributed planar graph algorithms
were given in 2016 by Ghaffari and Haeupler for computing a planar embedding~\cite{ghaffariH16_embedding}, for MST~\cite{GhaffariH16_shortcuts}, and for approximate undirected global min-cut~\cite{GhaffariH16_shortcuts}.  
Other $\tilde{O}(D)$-round algorithms soon followed for approximate SSSP~\cite{GHSYZ22,RozhonGHZL22_shortestpaths}, DFS~\cite{GP17,deterministic_sep}, undirected weighted girth~\cite{planardistributedmaxflow24}, undirected global min-cut~\cite{GZ22}, reachability~\cite{ParterReachability20a}, and planar separators~\cite{GP17,LP19,deterministic_sep,OurSIROCCO}.  
For other problems, $\tilde{O}(D^2)$-round algorithms were developed. These include directed maximum $st$-flow~\cite{planardistributedmaxflow24}, weighted directed SSSP~\cite{LP19}, directed weighted girth~\cite{ParterReachability20a}, and directed global min-cut~\cite{planardistributedmaxflow24}.
Notice that both maximum $st$-flow and SSSP have exact $\tilde{O}(D^2)$-round algorithms. Indeed, as we next explain, these two problems are highly related in planar graphs. However, while SSSP can  be {\em approximately} solved in near optimal $\tilde{O}(D)$-rounds in undirected graphs, no such bound was known for approximating maximum $st$-flow. 
Recently,~\cite{planardistributedmaxflow24} gave an almost optimal $D\cdot n^{o(1)}$-rounds $(1-o(1))$-approximation for maximum $st$-flow, but only in undirected $st$-planar graphs (where $s$ and $t$ must lie on the same face). 

\medskip
\noindent
{\bf Our result.}
In this paper, we present the first near-optimal $\tilde{O}(D)$-rounds $(1-o(1))$-approximation algorithm for maximum $st$-flow in undirected planar graphs.
This improves the algorithm of~\cite{planardistributedmaxflow24} that is restricted to $st$-planar graphs (and also has an additional $n^{o(1)}$ factor). 
Formally, we prove:

\begin{restatable}{theorem}{TheoremDistributedReif}
\label{thm: max_flow}
Let $G$ be an undirected weighted planar network of hop-diameter $D$, and let $s,t$ be two vertices whose ID is known to all vertices of $G$.
Then, in $\tilde{O}(D)$ deterministic rounds 
we can $(1-o(1))$-approximate the value 
of the maximum $st$-flow (equivalently, $(1+o(1))$-approximate the value  of the minimum $st$-cut) s.t. every vertex $v\in G$ knows this value. 
In addition, every $v \in G$ learns on which side of the minimum $st$-cut $v$ is, and $v$'s incident edges that belong to the cut. 
\end{restatable}

Notice that we find the approximate value of the maximum $st$-flow by finding an approximate minimum $st$-cut. We also find the cut itself and not just its value. Finding the corresponding flow  assignment distributively in the same $\tilde{O}(D)$ round complexity remains an interesting open question.

\subsection{Technical Overview} \label{sec:technicaloverview}
Given a planar graph $G$, its {\em dual} graph $G^*$  is a planar graph whose nodes correspond to the faces of $G$. Two nodes are connected in $G^*$ if their corresponding faces in $G$ share an edge (see \cref{fig: dual_graph} and a formal definition in \cref{sec: preliminaries}). 
In the centralized setting, the driving force behind the algorithmic improvements to maximum $st$-flow was the gradual understanding of the connection between maximum $st$-flow in $G$ and SSSP in $G^*$. This connection was already implicitly used by Ford and Fulkerson~\cite{ff56}.  
In 1981, Hassin~\cite{Hassin} made the connection explicit by showing that in undirected $st$-planar graphs $G$, maximum $st$-flow reduces to a single SSSP computation in the dual $G^*$ with positive edge-lengths. The $D\cdot n^{o(1)}$-rounds maximum $st$-flow approximation of~\cite{planardistributedmaxflow24} is a distributed (and approximate) implementation of Hassin's algorithm. 
In 1983, Venkatesan \cite{Venkatesan} showed that in general directed planar graphs (i.e., not just $st$-planar graphs) maximum $st$-flow reduces to $\log \lambda$ SSSP computations in the dual $G^*$ with positive {\em and negative} edge-lengths, where $\lambda$ is the maximum $st$-flow value.
The $\tilde{O}(D^2)$-rounds exact maximum $st$-flow algorithm  of~\cite{planardistributedmaxflow24} is a distributed implementation of Venkatesan's algorithm.
Note that the fastest known distributed algorithm for directed SSSP that allows negative edge-lengths takes $\tilde{O}(D^2)$ rounds~\cite{LP19}. In order to obtain a faster $\tilde{O}(D)$-round algorithm we use a different approach by Reif~\cite{Reif83}, which enables us to use a faster approximate SSSP algorithm that works for undirected graphs with non-negative edge-lengths. 
 Namely, in 1983, Reif~\cite{Reif83} showed that in general undirected planar graphs, maximum $st$-flow reduces to multiple (as many as $\Omega(n)$) SSSP computations on $G^*$ (with positive edge-lengths) that can be executed efficiently in a divide and conquer manner as we next explain. 
 
    We mention that the approach of Hassin, and therefore, the distributed approximation algorithm of~\cite{planardistributedmaxflow24}, is inherently limited to $st$-planar graphs. I.e., there is no clear way of extending it to work for general $s$ and $t$. That is because, Hassin's algorithm starts by augmenting $G$ with an edge $e=(t,s)$, then performs a single SSSP computation on the dual of the augmented graph.
    This is fine when $s,t$ are on the same face, for $e$ can be embedded inside the common face, preserving planarity (i.e., the planar dual can still be defined).
    However, for general $s$ and $t$, adding this edge violates planarity and the planar dual cannot be defined.
    Therefore, another approach is required.
    Here comes our contribution, which is a distributed implementation of Reif's algorithm.
    The technical challenge compared to the approximation algorithm of~\cite{planardistributedmaxflow24} is that, we need to (distributively) support multiple simultaneous SSSP computations on $G^*$ while $G^*$ changes dynamically during the algorithm (undergoes incisions) rather than just one SSSP computation.

\begin{figure}[htb]
\centering
\begin{subfigure}[t]{0.31\textwidth}
\centering
\includegraphics[width=1\linewidth]{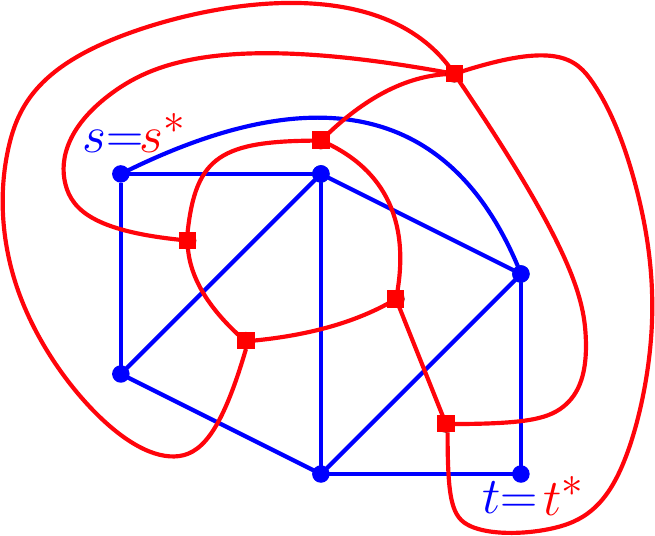} 
\caption{
\label{fig: dual_graph}
}
\end{subfigure}
\hspace{0.4 cm}
\begin{subfigure}[t]{0.3\textwidth}
\centering
\includegraphics[width=.75\linewidth]{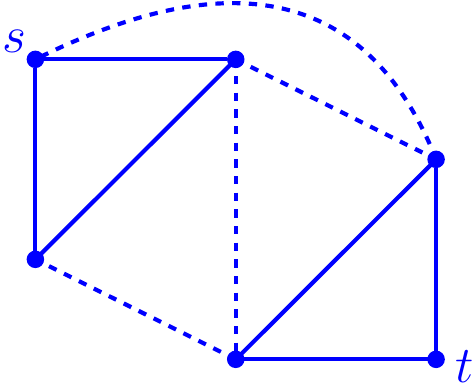}
\caption{ 
}
\end{subfigure}
\hspace{0.4 cm} 
\begin{subfigure}[t]{0.3\textwidth}
\centering
\includegraphics[width=1\linewidth]{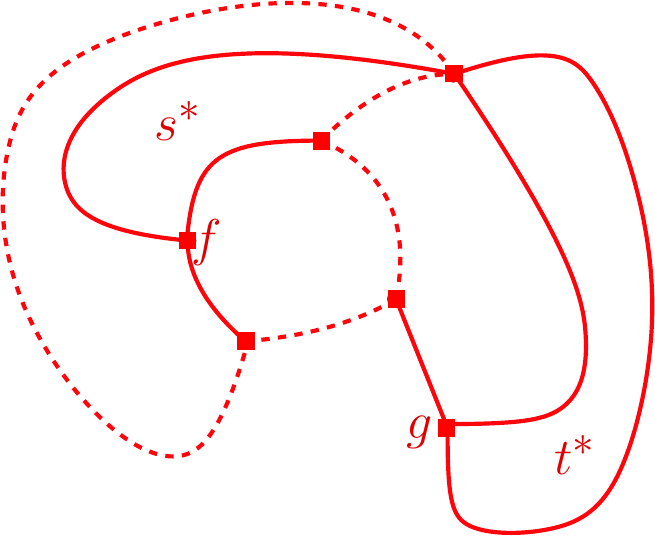}
\caption{
\label{fig: st_cut_dual}
}
\end{subfigure}
\caption{(a) The primal graph {\color{blue} $G$} in {\color{blue} blue} and its dual {\color{red} $G^*$} in {\color{red} red}. (b) Two distinguished vertices $s$ and $t$ 
in $G$, and an $st$-cut (dashed). (c) The faces $s^*$ and $t^*$ of $G^*$, and the dual cycle to the $st$-cut $C$ (dashed). $C$ encloses $s^*$ but not $t^*$. The node $f$ (resp. $g$) of $G^*$ belongs to the face $s^*$ (resp. $t^*$). 
\label{fig: st_cut_overview}
}
\end{figure}

\medskip
\noindent
{\bf Centralized Reif.}
In \cref{sec:Centralized Reif} we describe Reif's algorithm in detail. Reif finds the value of the maximum $st$-flow  by finding the equivalent value of a minimum $st$-cut. In planar graphs, the minimum $st$-cut in $G$ is also equivalent to the shortest cycle $\mathcal C$ in the dual graph $G^*$ that encloses exactly one of the two faces $s^*,t^*$ of $G^*$ corresponding to $s,t$. See \cref{fig: st_cut_overview}.

In order to find the cycle $\mathcal C$, Reif uses the fact that if $P$ is a shortest path in $G^*$ connecting the two faces $s^*$ and $t^*$, then $\mathcal C$ must cross $P$ exactly once. Then, Reif makes an {\em incision} along the path $P$, duplicating $P$'s edges and vertices. This way, the shortest {\em cycle} $\mathcal C$ becomes the shortest {\em path} $P_i$ from a vertex $p_i$ of $P$ to its copy. See \cref{fig: incision_P_overview}.

\begin{figure}[htb]
    \centering
    \begin{minipage}[t]{0.31\textwidth}
        \centering
    \includegraphics[width=0.95\linewidth, height=3cm , keepaspectratio]{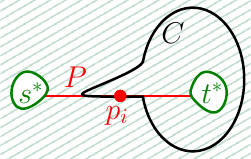}
\end{minipage}
    \begin{minipage}[t]{0.31\textwidth}
        \centering
    \includegraphics[width=0.95\linewidth,keepaspectratio, height=3cm]{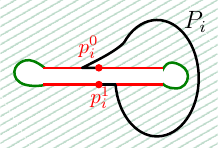}
    \end{minipage}
        \begin{minipage}[t]{0.31\textwidth}
        \centering
    \includegraphics[width=0.95\linewidth,keepaspectratio, height=3cm]{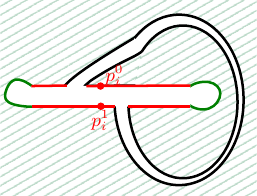}
    \end{minipage}
    \caption{
    Left: The shortest cycle $C$ in $G^*$ that crosses $P$ once and encloses  $t^*$ (not $s^*$).
    Middle: The graph $G^*$ after incising $P$, the $p_i^0$-to-$p_i^1$ path $P_i$ corresponds to the cycle $C$ from the left image.
    Right: The two annuli (recursive subgraphs) of the first recursion level after incising $P_i$. 
   \label{fig: incision_P_overview} 
    }
\end{figure}

Finding all such paths $P_i$ can be done with multiple executions of an SSSP algorithm on $G^*$. Moreover, since the $P_i$ paths do not cross, they can all be found very efficiently (total $O(n\log n)$ centralized time) in a (recursive) divide-and-conquer scheme. Namely, once we find some $P_i$ we make an incision along it, thus, splitting the graph into two subgraphs: One containing every $P_j$ for $j<i$ and the other containing every $P_j$ for $j>i$, so we can recurse on each subgraph separately. I.e., $\mathcal{C}$ is either $P_i$ or some $P_j$ in one of the two recursive subgraphs.\footnote{If $\mathcal{C}$ does not correspond to $P_i$, then we can always assume it is included in one of the recursive subgraphs. Otherwise, it must cross $P_i$, in which case we can replace  segments of $\mathcal{C}$ with segments of $P_i$ in order to ensure that $\mathcal{C}$ is entirely contained in a single recursive subgraph. For more details see \cref{sec:Centralized Reif}.}. See \cref{fig: incision_P_overview}.
We call these recursive subgraphs {\em annuli} since each of them is a subgraph sandwiched between two paths $P_i,P_j$.
See \cref{fig: general_annulus_overview}.

 \begin{figure}[htb]
 \centering
   \includegraphics[width=0.15\linewidth]{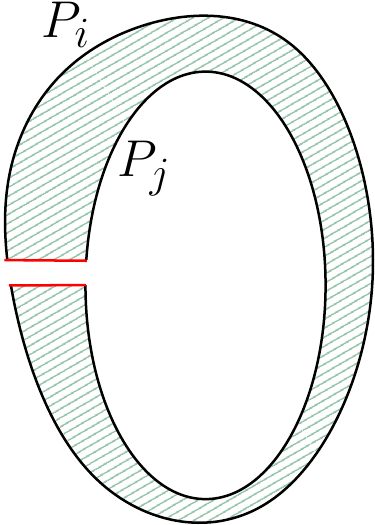}
   \caption{An annulus bounded by paths $P_i,P_j$. \label{fig: general_annulus_overview}}
\end{figure}

\medskip
\noindent
{\bf Distributed Reif.}
In this paper, we show how to implement (an approximate version of) Reif's algorithm distributively.\footnote{The reason our algorithm is approximate is solely because we use an approximate SSSP algorithm as a black-box.}
Our main contribution (\cref{sec: distributed_reif}) is in showing that distributed (approximate) SSSP computations can be executed efficiently on $G^*$ while $G^*$ undergoes incisions. More precisely, we show that the SSSP computations over all annuli in the same recursive level can be performed in $\tilde O(D)$ rounds. 
This is quite challenging since $G^*$ is not the communication network and since the diameter of $G^*$ can be much larger than that of $G$. 
That is, we want to simulate the algorithm using our original communication network $G$ and benefit from the fact that $G$ is a planar graph with diameter $D$.

\begin{figure}[htb]
\centering
\begin{subfigure}[t]{0.46\textwidth}
  \centering
   \includegraphics[width=0.7\linewidth, keepaspectratio, height=4.5cm]{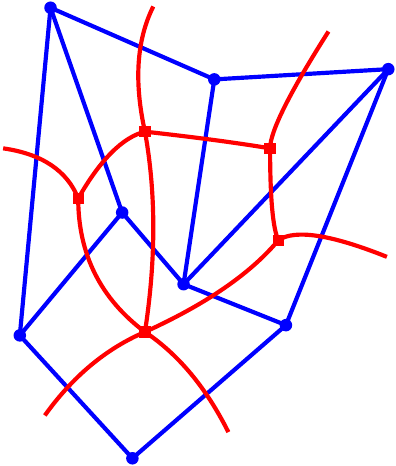}
\end{subfigure}
\hspace{0.3 cm}
\begin{subfigure}[t]{0.46\textwidth}
  \centering
   \includegraphics[width=0.9\linewidth, keepaspectratio, height=4.5cm]{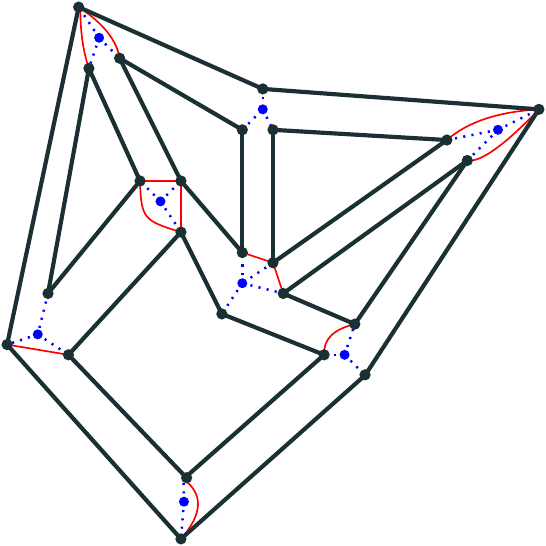}
  \end{subfigure}
    \caption{ 
    Left: The primal graph $G$ ({\color{blue} blue}) and its dual $G^*$ ({\color{red} red}). For clarity, the ({\color{red} red}) node corresponding to the external face is omitted.
   Right: The graph $\hat{G}_0$ of~\cite{GP17}. 
   Note, each face of $G$ corresponds to a (disjoint) {\bf black} cycle of $\hat{G}_0$, and each vertex of $G$ has its   ({\color{blue} blue}) star.
   The {\color{red} red} edges of $\hat{G}_0$ have a 1-1 mapping to edges of $G^*$.
   \label{fig: hat_g_base_case_overview}}
\end{figure}

To do so, we create a new graph $\hat{G}_{\ell}$ that can be simulated on $G$ and captures all recursive (dual) annuli in level $\ell$ of Reif. For the base case $\hat{G}_{0}$, we just want to simulate an algorithm on the dual graph $G^*$. This was already done by~\cite{GP17, planardistributedmaxflow24}, who introduced a graph $\hat{G}$ that allows us to simulate basic algorithms based on aggregate operations (usually simple operations as broadcast, SUM, OR, etc.) on the graph $G^*$ using the communication network $G$. Intuitively, in this graph, for each dual node $f \in G^*$ there is a cycle representing the corresponding face in $G$. Two cycles that correspond to different faces are disjoint, and two cycles $f,g$ are connected by an edge if and only if there is a dual edge $(f,g)$ in $G^*$. See \cref{fig: hat_g_base_case_overview}.
There are additional vertices and edges that make the graph $\hat{G}$ similar to $G$. In particular, each primal vertex $v$ is replaced by a star connecting it to all the cycles that correspond to faces $f$ that include $v$. These extra vertices and edges make sure that the graph has $O(D)$ diameter, and has similar structure to $G$, which allows us to efficiently simulate algorithms on $\hat{G}$ using the communication graph $G$.

Our goal is to extend the graph $\hat{G}_{0}=\hat{G}$ to the graph $\hat{G}_{\ell}$ that allows simulating algorithms simultaneously on all annuli of a recursive level $\ell$ of Reif. Note that a dual node $f$ in level $\ell$ can be either a dual node of $G^*$, or some duplicate of a node in $G^*$ that we got after making incisions. We thus refer to it as a {\em node-part}. Note that when we incise along a path that contains the node $f$, the edges adjacent to $f$ are split into two parts $f_1,f_2$ 
(where the edges of the path go to both sides), so the corresponding node-parts we get after one such incision are defined by a consecutive subset of the edges adjacent to $f$. If we look at these edges in the primal graph $G$, they correspond to a connected subgraph of the corresponding face $f$ in $G$, we call these connected subgraphs {\em face-parts} $p(f_1)$, $p(f_2)$.
Note that face-parts are not disjoint. First, each edge in the primal graph is part of two faces (i.e., incident to two different dual nodes). In addition, the two edges of the path we incised that are incident to $f$ go to both $p(f_1)$ and $p(f_2)$. See \cref{fig: nodepart_pf_overview}.

\begin{figure}[htb]
\centering
\begin{subfigure}[c]{0.44\textwidth}
  \centering
   \includegraphics[width=0.7\linewidth, keepaspectratio, height=3.6cm]{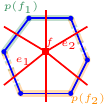}
\end{subfigure}
\hspace{0.3 cm}
\begin{subfigure}[c]{0.44\textwidth}
  \centering
   \includegraphics[width=0.9\linewidth, keepaspectratio, height=4cm]{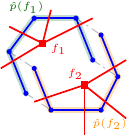}
  \end{subfigure}
    \caption{
    Left: A node $f\in G^*$ (in {\color{red} red}) and its incident edges $e_1,e_2$ on the incision path. After the incision, all edges between $e_1,e_2$ (resp. $e_2,e_1$)  get incident to the first (resp. second) node-part $f_1$ (resp. $f_2$). The {\color{blue}blue} cycle is the face (corresponding to) $f$ in $G$. The face-parts $p(f_1),p(f_2)$ in $G$ are highlighted ({\color{PineGreen} green}, and {\color{YellowOrange} orange}). Note that  $p(f_1),p(f_2)$ overlap in the edges (whose duals are) $e_1,e_2$.
    Right: The node-parts $f_1,f_2$ resulting from $f$ after the incision (each has its own copies of $e_1,e_2$). The transformation that $\hat{G}_\ell$ applies on the face-parts $p(f_1),p(f_2)$ of $G$ results in (highlighted) vertex-disjoint components $\hat{p}(f_1),\hat{p}(f_2)$ in $\hat{G}_\ell$. These components correspond to the node-parts, and each of them has its own copy of the edges $e_1,e_2$.
    \label{fig: nodepart_pf_overview}}
\end{figure}

In order to simulate algorithms efficiently, we want to have a disjoint connected subgraph for each such face-part $p(f_i)$. 
To achieve this, in $\hat{G}_{\ell}$ we duplicate each edge at most four times (once for each time it appears in a face-part $p(f)$ where $f$ is a level-$\ell$ node-part). 
We connect these edges in a way that for each level-$\ell$ node-part $f$: (1)  
there is a connected subgraph $\hat{p}(f)$ in $\hat{G}_{\ell}$, (2) subgraphs for distinct node-parts are disjoint, and (3) two subgraphs $\hat{p}(f), \hat{p}(g)$ are connected by an edge iff there is a dual edge in level $\ell$ between the corresponding node-parts $f,g$. See \cref{fig: nodepart_pf_overview} and \cref{fig: hatGell_local_overview}.
We show that in the graph $\hat{G}_\ell$, there is a connected component for each level-$\ell$ annuli. 
Finally, similarly to the original $\hat{G}$, we add another set of edges and vertices that make the graph similar to the original graph $G$. This is again done by replacing each vertex $v$ with some copies connected by a 
depth-two tree instead of a star in $\hat{G}_0$. Each copy connects $v$ to a subgraph $\hat{p}(f)$ for each node-part $f$ where $v \in p(f)$. 
See \cref{fig: hatGell_local_overview}.

\begin{figure}[htb]
\centering
\begin{subfigure}[c]{0.46\textwidth}
  \centering
     \includegraphics[width=0.8\linewidth, keepaspectratio, height=3cm]{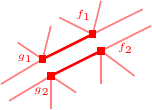}
\end{subfigure}
\hspace{0.3 cm}
\begin{subfigure}[c]{0.46\textwidth}
  \centering
   \includegraphics[width=0.8\linewidth, keepaspectratio, height=4cm]{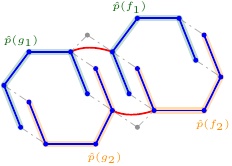}
  \end{subfigure}
    \caption{Left: four node-parts $g_1,g_2$ of $g\in G^*$, and $f_1,f_2$ of $f\in G^*$, resulting from an incision over a path of $G^*$ in which $f,g$ participate. In particular, $f,g$ are neighbors and the edge $(f,g)$ is duplicated as a result of the incision (its duplicates are solid {\bf \color{red} red}).
    Right: the corresponding components in $\hat{G}_\ell$ of the node-parts. Note, there is an edge ({\color{red}red}) connecting $\hat{p}(g_i)$ and $\hat{p}(f_i)$ in $\hat{G}_\ell$ for the edge $(g_i,f_i)$ in the dual after the incision ($i\in \{1,2\}$).
    \label{fig: hatGell_local_overview}}
\end{figure}

We show that the resulting graph $\hat{G}_{\ell}$ is planar, has diameter $O(D)$, and has a structure similar to $G$ which allows to simulate on it algorithms efficiently using the communication network $G$.
Moreover, $\hat{G}_{\ell}$ has a similar structure to the union of level-$\ell$ annuli, which allows us to simulate Reif's algorithm efficiently on them. 
Defining the graph $\hat{G}_\ell$ is challenging. In particular, it aims to capture a structure of $G$, that in a sense is global
(not only faces of $G$, but also  face-parts of $G$ and the way they interact after incisions).
Since this information is not known a priori (our aim is to learn it via $\hat{G}_\ell$), $\hat{G}_\ell$'s definition should rely merely on a more restricted (possibly local) information. E.g. the neighborhood of vertices and whether incident edges are incised in the dual.
We succeed to define the graph $\hat{G}_\ell$ that achieves both. A global structure with a local construction.
In fact, we define one graph $\hat{G}_\ell$ that captures the topology of {\em all} annuli of the same recursion level, simultaneously.

\medskip
\noindent
{\bf Minor-aggregation on annuli.}
More concretely, we can simulate minor-aggregation algorithms (a certain type of algorithms that are based on contractions and aggregate operations, see \cref{sec:Minor-aggregation}), which is enough in order to simulate SSSP as needed for Reif. In particular, to simulate such algorithm we need to solve the \emph{part-wise aggregation} problem, where the goal is to compute an aggregate function over different sets of nodes $S_i$ (where the sets $S_i$ are connected and disjoint from each other). We show that these sets $S_i$ on the dual graph correspond to connected disjoint subgraphs on the graph $\hat{G}_{\ell}$. Since $\hat{G}_{\ell}$ is a planar graph of diameter $O(D)$,  we can simulate such computations on $\hat{G}_{\ell}$ in $\tilde{O}(D)$ rounds using \emph{low-congestion shortcuts}~\cite{GhaffariH16_shortcuts, GhaffariH21_shortcuts, HaeuplerIZ21_shortcuts} for planar graphs (see \cref{sec: IDs_to_face_parts} and \cref{sec:Minor-aggregation} for details). 
Using this approach, we can already obtain a $D\cdot n^{o(1)}$ round complexity for the (approximate) SSSP problems we need to solve. To do so, we can use as a black-box a minor-aggregation distributed SSSP algorithm~\cite{GHSYZ22}. Note that we use approximate SSSP for the lack of almost-$D$-rounds exact distributed SSSP. Nonetheless, we show that this does not impair the correctness of Reif's algorithm, and that we get an approximation to the minimum $st$-cut (or maximum $st$-flow). The only issue with using the SSSP of~\cite{GHSYZ22} is its $D\cdot n^{o(1)}$ round complexity. To obtain the promised $\tilde{O}(D)$-round complexity we need an additional ingredient described next. 

\medskip
\noindent
{\bf Improved complexity via Eulerian orientation.} 
In order to remove the $n^{o(1)}$ factor and achieve our promised $\tilde{O}(D)$ bound, we wish to use another SSSP algorithm~\cite{RozhonGHZL22_shortestpaths}. This algorithm indeed requires  $\tilde{O}(D)$ rounds, but it is not a minor-aggregation algorithm. Specifically, exactly one procedure that it uses is not in the minor-aggregation model: a procedure that  finds an {\em Eulerian Orientation} to a certain (Eulerian) graph $H$ that is (roughly) a subgraph of the input graph to the SSSP algorithm. 
This problem in~\cite{RozhonGHZL22_shortestpaths} asks to find an orientation to the edges of $H$, so that the in-degree equals the out-degree for all nodes. 
In our case, $H$ is a (Eulerian) subgraph of an annulus. In \cref{sec: SSSP} we show how to solve this problem distributively in $\tilde{O}(D)$ rounds for all annuli of a given recursive level.

A classical (centralized) solution~\cite{Euler_par_AtallahV84} is: (1) nodes of $H$ pair their incident edges, locally decomposing the graph into edge-disjoint cycles. Then, (2) orient each cycle consistently (thus, each two paired edges are oriented in opposite directions as needed).
Trying to simulate this on annuli, we face two challenges: 
(1) A node $f$ in an annulus is not a real node (it is a face-part, simulated by many nodes of $G$). I.e., it does not know its neighbors and cannot pair its edges. 
(2) In order to orient the resulting cycles, we need to run a procedure on all of those cycles simultaneously. However, we do not know how to simulate algorithms efficiently on multiple subgraphs of annuli that are not disjoint (note that the cycles share nodes).  

The first challenge can be solved on the face-part $p(f)$ of $f$. I.e., pairing consecutive pairs of edges on $p(f)$ directly implies a pairing on $f$'s edges (here we use planarity). This can be computed by simulating aggregate computations on the (connected and disjoint) subgraphs $\hat{p}(f)$ in $\hat{G}_\ell$. 
To deal with the second challenge, we augment the graph $\hat{G}_\ell$ such that the cycles we need to orient correspond to \emph{disjoint} subgraphs of $\hat{G}_\ell$. This is done by splitting some of the vertices of $\hat{G}_\ell$. Once the (subgraphs that map to) cycles are disjoint, we can solve the orientation problem in $\tilde{O}(D)$ rounds using a minor-aggregation algorithm over the graph $\hat{G}_\ell$.

\medskip
\noindent
{\bf Organization.}
In the remainder of this section, we review some preliminaries. In \cref{sec:Centralized Reif} we discuss Reif's centralized algorithm in detail and introduce some useful notions and extensions for its distributed implementation. In \cref{sec: distributed_reif}, we provide our main contribution: A distributed simulation of incisions on the dual graph $G^*$ using $\hat{G}_{\ell}$. In \cref{sec: SSSP} we solve the aforementioned Eulerian orientation problem. 
Finally, \cref{sec: implemetnation_details} puts everything together and provides a concrete distributed implementation of Reif's algorithm.

\subsection{Preliminaries}
\label{sec: preliminaries}
We denote by $G=(V,E)$ the undirected edge-weighted planar network of communication, and by $D$ the network's (unweighted) hop-diameter. 
We denote by $G[S]$ the subgraph of $G$ induced by a set of vertices or edges $S$.

\medskip
\noindent
{\bf The \boldmath$\CONGEST$ model.}
We work in the standard distributed $\CONGEST$ model~\cite{peleg-book}. Initially, each vertex knows only its unique $O(\log n)$-bit identifier and the identifiers of its neighbors. Communication occurs in synchronous rounds.
In each round, each vertex can send each neighbor a distinct $O(\log n)$-bit message on their common edge.
We assume that the edges of $G$ are weighted and that the weights are polynomially bounded integers.
Thus, the weight of an edge can be transmitted in $O(1)$ rounds. This is a standard assumption in the $\CONGEST$ model.

\medskip
\noindent
{\bf Planar embedding.}
Let $G = (V, E)$ be an undirected planar graph.
The {\em geometric} planar embedding of $G$ is a drawing of $G$ on a plane so that edges intersect only in vertices.
A {\em combinatorial} planar embedding of $G$ provides for each $v\in G$, the local clockwise order of its incident edges, such that, the ordering of edges is consistent with some geometric planar embedding of $G$. 
Throughout, we assume that a combinatorial embedding of $G$ is known locally for each vertex. This is done in $\tilde{O}(D)$ rounds using the deterministic planar embedding algorithm of Ghaffari and Haeupler~\cite{ghaffariH16_embedding}.

 \begin{figure}[htb]
  \centering
   \includegraphics[width=0.25\linewidth]{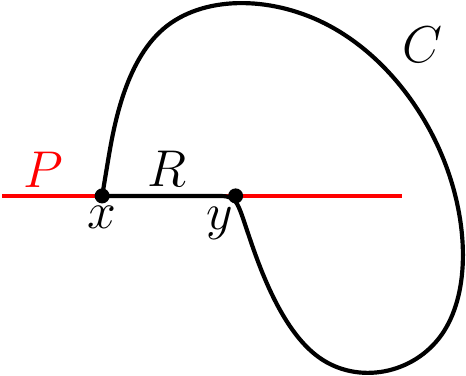}
   \caption{A (black) cycle $C$ crossing the ({\color{red}red}) path $P$ at an $x$-to-$y$ subpath $R$.  \label{fig:crossing}}
\end{figure}

\label{def: crossings}
\medskip 
\noindent
{\bf Crossing Paths.}
Consider a path $P = p_0,p_1,\ldots, p_t$. For an internal vertex $p_i$, all the edges incident to $p_i$ that appear between $(p_{i-1},p_i)$ and $(p_i,p_{i+1})$ in the clockwise order of $p_i$ are said to be on the {\em left side} of $P$. Symmetrically, the edges that appear between $(p_i,p_{i+1})$ and $(p_{i-1},p_i)$ in the clockwise order of $p_i$ are said to be on the {\em right side} of $P$.
Let $P$ and $C$ be undirected paths or cycles. We say that $P$ and $C$ {\em cross at subpath $R$} if (1) $P$ and $C$ share a proper (not a prefix or suffix) subpath $R$, and (2) The
edges of $C$ that follow and precede $R$ are in different sides of $P$. See \cref{fig:crossing}. Notice that it is possible that $R$ contains a single vertex.

\label{def: dual_graph}
\medskip 
\noindent
{\bf The dual graph $G^*$.}
The {\em dual} of the embedded {\em primal} planar graph $G$ is an embedded planar graph, denoted $G^*$ (see \cref{fig: st_cut_overview}).
The nodes of $G^*$ correspond to the faces of $G$.
For every edge $e$ in $G$ there is an edge $e^*$ in $G^*$ that connects the nodes corresponding to the two faces of $G$ that contain $e$ (and a self-loop when the same face appears on both sides of an edge). 
Observe that if two faces of $G$ share several edges then $G^*$ has several parallel edges between these faces. I.e., $G^*$ might be a multi-graph even when $G$ is a simple graph.
We refer to the vertices of $G$ as {\em vertices} and to the vertices of $G^*$ as {\em nodes}. 
We use the well-known duality (see e.g.,~\cite{KleinM_book,Reif83}) between primal cuts and dual cycles (see \cref{fig: st_cut_dual}):

\begin{fact}[Undirected Cycle-Cut Duality]
    \label{fac: cycle_cut_duality}
   A set of edges $C$ is a simple $st$-cut\footnote{A simple $st$-cut is a cut whose removal disconnects the graph into exactly two connected components, one containing $s$ and the other $t$.} in an embedded connected undirected planar graph $G$ if and only if $C$ is a simple cycle in $G^*$ that encloses exactly one of $s^*,t^*$ (the faces of $G^*$ corresponding to the vertices $s,t$ of $G$).
\end{fact}

\section{Centralized Reif}
\label{sec:Centralized Reif}
 In undirected planar graphs (Fact \ref{fac: cycle_cut_duality}), the weight of a minimum $st$-cut (or maximum $st$-flow) is  equivalent to the length of the shortest cycle $C$ in $G^*$ that encloses exactly one of the faces $s^*,t^*$.
Reif's centralized algorithm (\cref{alg: Reif}) from the 80's~\cite{Reif83} finds this cycle $C$ in $O(n\log n)$ time, using the following.

\begin{claim}[Reif~\cite{Reif83}]
\label{claim:Reif}
Let $f$ (respectively $g$) be any node of $G^*$ that belongs to $s^*$ (respectively $t^*$).  
Let $P$ be a shortest $f$-to-$g$ path in $G^*$.
The shortest cycle $C$ in $G^*$ that encloses exactly one of the faces $s^*,t^*$ is  the shortest cycle in $G^*$ that crosses $P$ exactly once. 
\end{claim}

Intuitively, $C$ must cross $P$ at least once, because it encloses exactly one of $s^*,t^*$. Moreover, we can assume that not more than once, because if it crosses $P$ at subpaths $R_1$ and $R_2$ then the subpath of $C$ between $R_1$ and $R_2$ can be replaced with the subpath of $P$ between $R_1$ and $R_2$, which is a shortest path (note this argument only holds in {\em undirected} planar graphs).

\begin{algorithm}[htb]
\caption{Reif's algorithm}
\KwIn{A $D$-diameter positively weighted undirected planar network $G$, and two vertices $s,t\in G$}
\KwOut{A minimum $st$-cut}
\label{alg: Reif}
\SetKwFunction{FindSPAnnulus}{ShortestCycle} 
\SetKwProg{Proc}{Procedure}{}{}

Compute an $f$-to-$g$ shortest path $P$ \tcc*[r]{$f,g\in G^*$ and $f\in s^*,g\in t^*$}

Incise along $P$ (results in a graph $A$) 

\Return \FindSPAnnulus{$A$, $0$, $|P|-1$}

\Proc(\tcc*[f]{$P_j,P_k$ define the annulus $A$}){\FindSPAnnulus{$A$, $j$, $k$}:}{

{\bf if} $k\leq j+1$  {\bf then} {\Return $\text{argmin}\{|P_k|,|P_j|\}$} \tcc*[r]{leaf annulus}

Compute $P_i$ for $i:=\lfloor (j+k)/2\rfloor$

Incise $P_i$ and contract common paths (results in annuli $A_l, A_r$)

\Return $\text{argmin}\{|P_i|$, |\FindSPAnnulus{$A_l$, $j$, $i$}|, |\FindSPAnnulus{$A_r$, $i$, $k$}|$\}$}

\end{algorithm}

To find $C$, Reif uses divide-and-conquer, based on the following. Let $C_i$ be the shortest cycle that crosses $P$ exactly once at a subpath that contains node $p_i$ of $P$. We can assume that $C_i$ does not cross any other $C_j$ because both are shortest cycles (the proof is similar to that of Claim~\ref{claim:Reif}). 
Reif's algorithm computes all cycles $C_i$ and declares the shortest one to be $C$.  First, the cycle $C_i$ is computed for $i =  |P|/2$. The cycle $C_i$ divides the graph $G^*$ into two subgraphs: the one enclosed by $C_i$, and the one not enclosed by $C_i$. 
The algorithm recurses on both subgraphs.

 \begin{figure}[H]
  \centering
   \includegraphics[width=0.3\linewidth]{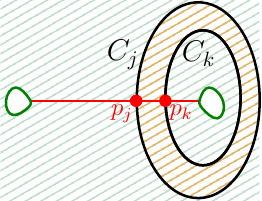}
   \caption{An annulus ({\color{orange}orange}) sandwiched between two cycles $C_j,C_k$.\label{fig: general_annulus}}
\end{figure}

Note that at most half of $P$'s vertices are strictly enclosed by $C_i$ and at most half are not strictly enclosed by $C_i$. Thus, the recursion depth is logarithmic. Observe that if $C_i$ and $P$ share a common $p_{i_1}$-to-$p_{i_2}$ subpath containing $p_i$, then we do not need to recursively compute $C_t$ for any $i_1<t<i_2$ since one of $C_{i_1}$ or $C_{i_2}$ must be as short as $C_t$ or shorter.  A formal proof is provided later in Claim~\ref{claim: Pi_P_intresection}.
Moreover, since $C_i$s do not cross, all subgraphs in the same recursion level are internally disjoint. 
 That is, every recursive subgraph (called {\em annulus}) is the subgraph sandwiched between two cycles $C_j,C_k$. See \cref{fig: general_annulus}.
 The goal of the recursive call in the annulus bounded by $C_j$ and $C_k$ is to compute (inside the annulus) the cycle $C_i$ for $i=\lfloor (j+k)/2\rfloor$.
 Note that two annuli in the same recursion level can intersect only in their bounding cycles. 
 
 Before computing all cycles $C_i$, Reif initially makes an \emph{incision} along $P$ (duplicating $P$'s edges and vertices, see \cref{fig: incision_P_overview}):

\begin{definition}[Incision along $P$]
\label{def: incision_P}
Splits every node $p_i$ of $P$ into two duplicates $p_i^0$ and $p_i^1$ called node-parts. Every edge $(p_i,p_{i+1})$ of $P$ is replaced by two edges $(p_i^0,p_{i+1}^0)$ and $(p_i^1,p_{i+1}^1)$, and each edge $(p_i,f)$ is replaced by an edge $(p_i^0,f)$ (resp. $(p_i^1,f)$) if it emanates left (resp. right) from $P$.  
\end{definition}
There is a subtle point regarding the endpoints of $P$: they are incident to only one edge of $P$ and two are required in order to define left and right. We therefore choose a non-$P$ edge of $s^*$ (resp. $t^*$) that is incident to $p_0$ (resp. $p_{|P|-1}$) to define left and right (and thus the duplicates) of $p_0$ (resp. $p_{|P|-1}$). 
A broader discussion on endpoints of $P$ appears in \cref{remark: incision_P}, \cref{sec: hat_G}.

\begin{figure}[htb]
    \centering

    \begin{subfigure}{0.35\textwidth}
        \centering
        \includegraphics[width=\linewidth]{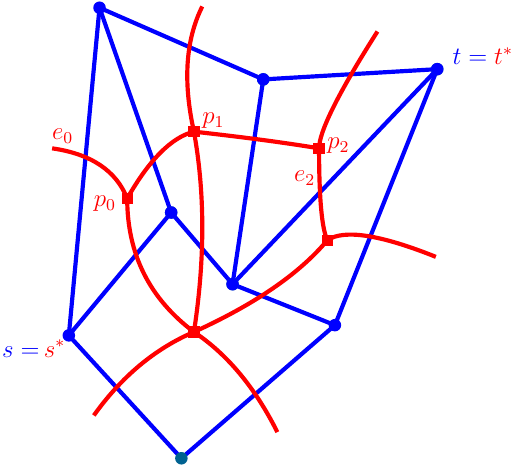}
    \end{subfigure}
    \hspace{0.5cm}
    \begin{subfigure}{0.27\textwidth}
        \centering
        \includegraphics[width=\linewidth]{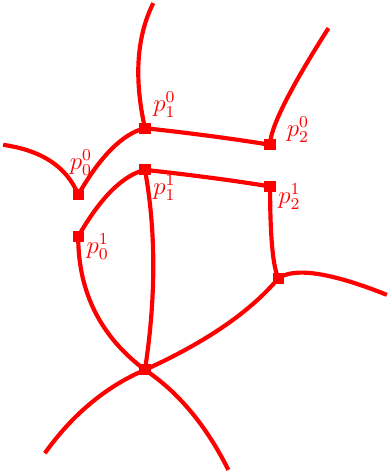}
    \end{subfigure}
           \caption{Left: The primal graph {\color{blue} $G$} in {\color{blue} blue} and its dual {\color{red} $G^*$} in {\color{red} red}. The path $P$ is $p_0,p_1,p_2$.
           The edges $e_0,e_2$ are the two non-$P$ edges that are incident to $P$'s endpoints in order to define the incision for them.
           Right: the dual graph after making an incision along $P$. 
           \label{fig: incision}}   
\end{figure}

The goal of this incision is that $C_i$ before the incision is equal to the shortest $p_i^0$-to-$p_i^1$ path $P_i$   after the incision, and $P_i$ (being a shortest path and not a cycle) can be found by an SSSP computation. Moreover, $P_i$ does not cross $P_j$ for any $j \neq i$ for the same reason that $C_i$ does not cross $C_j$.   
Therefore, after finding $P_i$ we can make an incision along $P_i$.  
Incising $P_i$ is similar to incising $P$ (\cref{def: incision_P}) in the sense that we duplicate nodes and edges of $P_i$. However, it is more delicate as we need to make sure not to duplicate edges of $P_i$ that overlap with a previously incised path such as $P$ or some $P_j$ (which would lead to high congestion). We discuss this and formally define incisions on $P_i$s shortly in \cref{sec: Refis_extensions}.
The incision along $P_i$ then guarantees that all the annuli in a given recursive level $\ell$ are connected and disjoint.
We prove this formally in \cref{sec: incision_decomposition}.  
In the centralized setting, finding $P_i$ in an annulus can be done in time linear in the size of the annulus by using the linear time (exact) SSSP algorithm of \cite{Henzinger}.  
Since the annuli in every recursion level are  disjoint, this sums up to $O(n)$ time per level and therefore to $O(n\log n)$ time overall.

\begin{algorithm}[htb]
\caption{Extended Reif}
\KwIn{A $D$-diameter positively weighted undirected planar network $G$, two vertices $s,t\in G$, and a precision parameter $(\log n)^{-O(1)}\leq \epsilon$.}
\KwOut{A  {\color{red} $(1+\epsilon)$ approximate} minimum $st$-cut} 
\label{alg: extended_Reif}
\SetKwFunction{FindSPAnnulus}{ShortestCycle} 
\SetKwProg{Proc}{Procedure}{}{}

Compute an $f$-to-$g$ shortest path $P$

Incise along $P$ (results in a graph $A$) 

\Return \FindSPAnnulus{$A$, $0$, $|P|-1$}

\Proc
{\FindSPAnnulus{$A$, $j$, $k$}:}{

{\bf if} $k\leq j+1$  {\bf then} {\Return $\text{argmin}\{|P_k|,|P_j|\}$}

{\color{red}\If(\tcc*[f]{do not recurse on overlapping annuli}){$P_j\cap P_k\neq \emptyset$}{
Compute a $(1+\epsilon)$ approximate SSSP tree $T$ from $x\in P_j\cap P_k$\\
\Return the minimal $p_t^0$-to-$p_t^1$ path in $T$ for ${j\leq t\leq k}$ 
}}

Compute a {\color{red} $(1+\epsilon)$ approximate} path $P_i$, for $i:=\lfloor (j+k)/2\rfloor$ 

Let $P_i'$ be the maximal subpath of $P_i$ that is internally-disjoint from (copies of) $P$, and let $p^0_{i_1},p^1_{i_2}$ be the endpoints of the common subpath of $P_i$ and (copies of) $P$.
 \tcc*[f]{$P_i$ intersects $P$ with a subpath}
 
Incise $P_i'$ (results in annuli $A_l, A_r$)

\Return $\text{argmin}\{|P_i|$, |\FindSPAnnulus{$A_l$, $j$, $i_1$}|, |\FindSPAnnulus{$A_r$, $i_2$, $k$}|$\}$}

\end{algorithm}

\subsection{Extensions of Reif}
\label{sec: Refis_extensions}
We introduce here two small changes to Reif's algorithm  (marked in red  in \cref{alg: extended_Reif}) that do not affect his original centralized version but are needed for our  distributed version.

\medskip
\noindent
{\bf Overlapping annuli.} 
Since Reif assumes that the total number of edges over all annuli of the same level is  $O(n)$, whenever we incise  a path $P_i$, we must make sure that we do not duplicate edges too many times. Reif takes care of this as follows.
Consider the second time we want to incise an edge of $G^*$. 
This can only happen when there is an annulus sandwiched between two paths $P_k, P_j$, and the path $P_i$ (for $i=(k+j)/2$) shares a common subpath (between two nodes $x$ and $y$) with $P_k$ or $P_j$ (or both). 
See \cref{fig:contraction}. To cope with this, Reif used the following simple solution: Assume $P_i$ shares a common $x$-to-$y$ subpath with $P_k$ (the case of $P_j$ is similar), then all paths $P_t$ for $k\le t \le i$ must also share this $x$-to-$y$ subpath. Therefore, before making an incision along $P_i$ and recursing into the annulus sandwiched between $P_k$ and $P_i$, Reif contracts the $x$-to-$y$ subpath into a single edge. Similarly, if $P_i$ and $P$ share a common subpath then Reif contracts this subpath before recursing.

\begin{figure}[t]
    \centering
    \begin{minipage}[t]{0.3\textwidth}
\centering   
\includegraphics[width=1\linewidth,height=4.5cm, keepaspectratio]{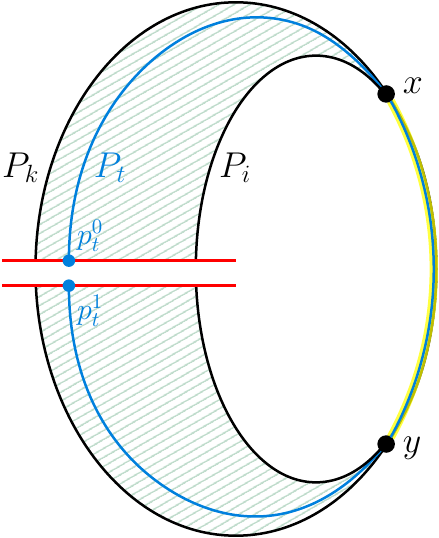}
   \caption{$P_i$ shares the (yellow) $x$-to-$y$ subpath with $P_k$. All paths $P_t$ for $k\le t \le i$ also share this subpath. Red paths are incision duplicates of $P$.\label{fig:contraction}}
    \end{minipage}
    \hfill
    \begin{minipage}[t]{0.67\textwidth}
        \centering
    \includegraphics[
        width=0.5\linewidth,
        height=4.5cm,
        keepaspectratio]{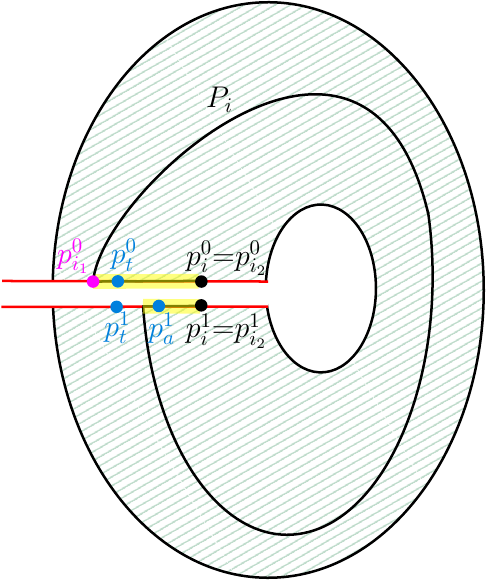}
\hspace{0.7 cm}
    \includegraphics[
        width=0.5\linewidth,
        height=4.5cm,
        keepaspectratio]{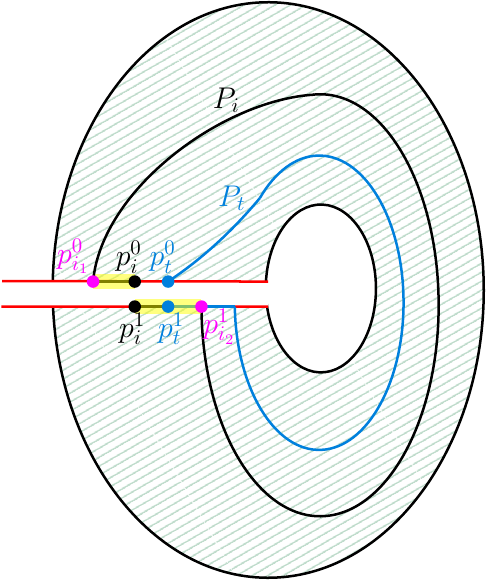}
\caption{The path $P_i$ shares a common (yellow) subpath with $P$. For every node $p_t$ along this subpath we have that $|P_{i_1}|\leq |P_t|$ or $|P_{i_2}|\leq |P_t|$. In the right (left) image, the common subpath contains (does not contain) an edge of $P$ and its copy. \label{fig: Pi_intersects_P}}
\end{minipage}
\end{figure}

Unfortunately, simulating such subpath contractions is infeasible in the distributed setting. This is because the contracted edges of $G$ would have too many corresponding dual edges in a given level of recursion, leading to high congestion. 
Luckily, as we show next, this can be solved easily.

First consider the case that $P_i$ shares a common $x$-to-$y$ subpath with $P_k$.    
Instead of contracting the $x$-to-$y$ subpath,  we apply the following simple solution: since all paths $P_t$ for $k\le t \le i$ must share the $x$-to-$y$ subpath, we compute an SSSP tree from $x$, which provides the distances from $x$ to every $p_t^0$ and $p_t^1$. Then, for $k\le t \le i$, we can declare $P_t$ to be the concatenation of the shortest $p_t^0$-to-$x$ and $x$-to-$p_t^1$ paths. This takes care of the entire  $P_k, P_i$ annulus and (in contrast to Reif's contraction) we are then done with the annulus (i.e., we do not recurse into it).
 \begin{claim}
 \label{claim: Pi_Pk_Pj_intresection}
Let $A$ be an annulus bounded by paths $P_k,P_j$ where $k<j$. Assume that $P_i$ for  $i = \lfloor(j+k)/2 \rfloor$  shares a node $x$ with $P_k$ (resp. $P_j$). Then, every path $P_{t}$ for $k\leq t\leq i$ (resp. $i\leq t\leq j$) is the  concatenation of the $x$-to-$p_{t}^0$ and the $x$-to-$p_{t}^1$ paths in $A$.
 \end{claim}
 \begin{proof}
 	Every path $P_{t}$ for $k\leq t\leq i$ (the  $i\leq t\leq j$ case is similar) starts from the vertex $p_{t}^0$ 
     that lies on the cycle formed by $p_{k}^0, p_{k+1}^0 \ldots p_{i}^0$ and the prefixes of $P_k$ and $P_i$ ending in $x$. The path $P_{t}$ must exit this cycle since $p_{t}^1$ is external to the cycle. 
     The only way to exit the cycle is through $x$ 
     and so $P_{t}$ must touch $x$ (and is therefore the concatenation of the $x$-to-$p_{t}^0$ and the $x$-to-$p_{t}^1$ paths). See \cref{fig:contraction}.
 \end{proof}

Second, consider the case that $P_i$ shares a common subpath with $P$.  
In this case, instead of contracting this common subpath (as Reif does), when incising along $P_i$ we do not incise along the common subpath, and only incise on the rest of $P_i$. 
All nodes $p_t$ that lie (strictly) on the common subpath simply do not compute their $P_t$ recursively. This is possible by the following claim (see also  \cref{fig: Pi_intersects_P}).

\begin{claim}
\label{claim: Pi_P_intresection}
    Assume that $P_i$ shares a common $p_{i_1}$-to-$p_{i_2}$ subpath with $P$. Then, for every $p_t$ where $i_1<t<i_2$ we have that either $P_{i_1}$ or $P_{i_2}$ is not longer than $P_t$ (i.e., we do not need to compute $P_t$).
  \end{claim}
\begin{proof}
Consider first the case where the common subpath does not contain an edge of $P$ and its copy (see \cref{fig: Pi_intersects_P} right). 
If $P_t$ is enclosed by $P_i$ (the other case is symmetric) then $P_t$ must include a proper suffix of the $p_{i_1}$-to-$p_{i_2}$ subpath of $P$. Otherwise, $P_t$ has to cross $P_i$ and we know this is forbidden. 
Having a shared suffix means that, while $P_t$ is a $p_t^0$-to-$p_t^1$ path it can also be seen as a $p_{i_2}^0$-to-$p_{i_2}^1$ path. I.e., $P_{i_2}$ cannot be longer than $P_t$.  

Now consider the case where the common subpath does contain some edges of $P$ and their copies (see \cref{fig: Pi_intersects_P} left), and again assume w.l.o.g. that $P_t$ is enclosed by $P_i$. In this case, $P_t$ must begin with a $p_t^0$-to-$p_a^1$ prefix where $p_a^1$ belongs to $P_i$. This prefix cannot be shorter than the $p_t^0$-to-$p_a^1$ subpath of $P_i$ (for otherwise $P_i$ would use it) that passes through $p_{i_1}^0$. 
This implies that $P_{i_1}$ cannot be longer than $P_t$.
\end{proof}

\begin{figure}[htb]
\centering
\begin{subfigure}[t]{0.45\textwidth}
  \centering
   \includegraphics[width=.55\linewidth]{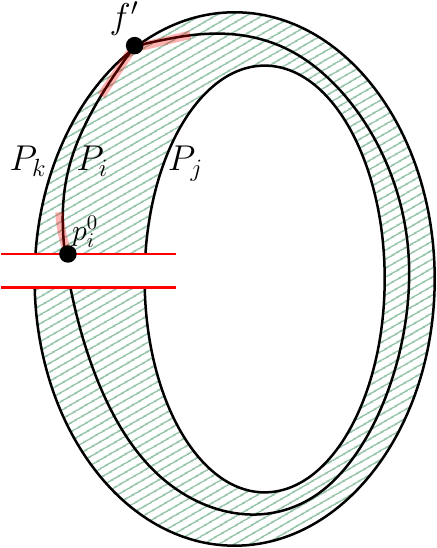}
\end{subfigure}
\hspace{0.7 cm}
\begin{subfigure}[t]{0.45\textwidth}
\centering
\includegraphics[width=.55\linewidth]{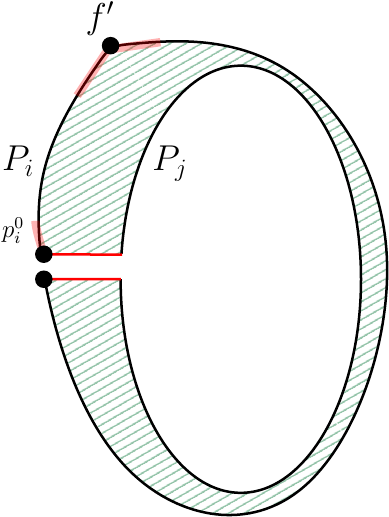}
\end{subfigure}
\caption{Left: An annulus of level $\ell-1$ bounded by paths $P_k,P_j$. $p_i^0\in P_i$ is a node-part of $p_i\in P$. The node-part $f'\in P_i$ is not a node of $P$.
{\color{red}Red} highlighted edges are edges of $P_i$ incident to $f'$ and $p_i^0$.
Right: the annulus defined by $P_i,P_j$ in level $\ell$. 
\label{fig: reursive_nodepart}}
\end{figure}

We conclude with a characterization of incisions on $P_i$s. See \cref{fig: incision_P_overview} and \cref{fig: reursive_nodepart}.
\begin{remark}[Incision along $P_i$]
\label{def: incision_Pi}
Let $A$ be an annulus defined by paths $P_j,P_k$.
Claim~\ref{claim: Pi_P_intresection} and Claim~\ref{claim: Pi_Pk_Pj_intresection} allow us to define incisions on the path $P_i$ ($i=\lfloor (j+k)/2\rfloor$) as follows. We incise (duplicate) only edges of the maximal subpath $P_i'$ of $P_i$ that is internally-disjoint from $P$. If $P_i$ shares a subpath with  $P_j$ or $P_k$, we exclude those edges too from the incision.
\end{remark}

The edges we incise as in the above remark, together with the already incised edges of $P_i$, imply that all edges of $P_i$ are now incised once but not more. Moreover, this incision on $P_i$ defines an internal (resp. external) annulus of $A$ that is sandwiched by $P_i,P_k$ (resp. $P_j,P_i$). 
Such annuli might be discarded in case $P_i$ intersects $P_k$ (or $P_j$) because in this case (overlapping annuli) we terminate the recursion. Henceforth, we obtain the following lemma (a full proof is provided in \cref{appendix: centralized_missing_proofs}).

\begin{restatable}{lemma}{InciseEdgeOnce}
\label{lem: incise_edge_once}
    Every edge of $G^*$ is incised at most once in our implementation of Reif.
\end{restatable}

\medskip
\noindent
{\bf Approximate Reif.}
The second change we do to Reif's algorithm is to turn it from an exact to an approximate algorithm. That is, we wish to compute $P$ as well as every $P_i$ with a $(1+\epsilon)$-approximate (rather than an exact) SSSP algorithm, simply because there is no distributed exact algorithm with an $\tilde O(D)$ round complexity. 
It is not immediately clear that the correctness of Reif's algorithm holds in this case (i.e., finds a $(1+\epsilon)$-approximate minimum $st$-cut). In particular, the crucial  property that $P_i$s do not cross (which relies on the triangle inequality) does not immediately hold.
However, we show that it does hold if the approximate SSSP algorithm obeys the {\em approximate triangle inequality} property (see formal definition in \cref{appendix: centralized_missing_proofs}). 
This property implies that any $u$-to$v$ subpath of a $(1+\epsilon)$-approximate shortest path $P_i$ returned by the SSSP algorithm is also a $(1+\epsilon)$ approximation to the $u$-to-$v$ shortest path.
The approximate shortest paths algorithm of~\cite{RozhonGHZL22_shortestpaths} that we use does not admit this property.  Using a reduction of~\cite{RozhonHMGZ23_triangle_inequality}, we augment~\cite{RozhonGHZL22_shortestpaths}'s SSSP algorithm to have this property (see \cref{section: SSSP}).

\begin{restatable}[Approximate Reif]{lemma}{lemApproximateReif}
\label{lem: approx_reif}
    Using a $(1+\epsilon)$-approximate SSSP algorithm that obeys the approximate triangle inequality, setting $\epsilon=O(\epsilon'/\log n)$   
    ensures that all $P_i$s do not cross and are $(1+\epsilon')$-approximate shortest $p_i^0$-to-$p_i^1$ paths.
\end{restatable}

The correctness of approximate Reif follows from the above lemma. 
The proof is simple and is given in \cref{appendix: centralized_missing_proofs}. 
Intuitively, starting from a $(1+\delta)$ approximate shortest $p_j^0$-to-$p_j^1$ path $P'_j$, we can replace each segment of $P'_j$ that crosses $P_i$ with a subpath of $P_i$. We do the same if $P_j'$ intersects $P$ until we have a path $P_j$ that does not cross any $P_i$ and intersects (each of the two copies of) $P$ with at most one subpath. Doing this, each time we lose at most an additional $(1+\epsilon)$ factor in the approximation.

\subsection{Towards a Distributed Implementation}
\label{sec: incision_decomposition}
The challenge in applying (an approximate version of) Reif's algorithm distributively is performing the incisions on the dual graph $G^*$. 
In particular, a face $f$ of $G$ does not exist as an atomic computational unit, meaning, there is no entity that knows (even basic) information about $f$. E.g. the number of edges in $f$, or the identities of $f$'s neighbors in $G^*$ and their ordering (embedding) around $f$.
In \cref{sec: distributed_reif} (our main technical contribution) we discuss how to overcome this and perform incisions on $G^*$ distributively. Namely, we show that a decomposition of $G^*$ by incisions can be learned and simulated without any vertex of $G$ knowing the entire information related to a face $f$. Moreover, we show how to simulate {\em minor-aggregation} algorithms on $G^*$ while doing incisions (i.e., on the annuli obtained from the incisions). 
The minor-aggregation framework~\cite{GHSYZ22,GZ22} is an interface for recursive edge-contraction based algorithms, that can be implemented in $\CONGEST$ using the well-known tool of {\em low-congestion shortcuts}~\cite{GhaffariH16_shortcuts,GhaffariH21_shortcuts,HaeuplerIZ21_shortcuts}. The (approximate) SSSP algorithm of~\cite{RozhonGHZL22_shortestpaths} that we use, is (mostly) a minor-aggregation algorithm. In particular, it has only one component that is not a minor-aggregation. In \cref{sec: SSSP} we show how to overcome this.

We next define some notions and prove several properties on the structure of annuli (and related graphs in $G$). These will be useful for the distributed implementation. 
We view Reif's recursive incisions as a tree $\mathcal{T}$.

\begin{definition}[Incision decomposition $\mathcal{T}$]
\label{def: incision_decomposition}
The root of $\mathcal{T}$ corresponds to the entire graph $G^*$. The root has a single child corresponding to the graph  obtained from $G^*$ by performing an incision along the path $P$.
Every other node of $\mathcal{T}$ corresponds to an annulus  sandwiched between two paths $P_k,P_j$  as in Reif's algorithm (it is possible that one of $P_k,P_j$ is empty). Such a node has at most two children defined by the annulus sandwiched between $P_k,P_i$ and the annulus sandwiched between $P_i,P_j$ for $i=\lfloor (j+k)/2 \rfloor$ (i.e., the two annuli resulting from making an incision along $P_i$). Recall that if $P_i$ intersects $P_k$, then we do not recurse into the $P_k,P_i$ annulus and so there is only one child (corresponding to the annulus $P_i,P_j$). Similarly for the case where $P_i$ intersects $P_j$. If $P_i$ intersects both $P_j,P_k$ then there are no children (this is a leaf node). Another halting condition (leaf node) is when $j = k+1$, i.e., the annulus contains only two (consecutive) nodes of $P$.
\end{definition}

\noindent
{\bf Nodes and edges of annuli in $\mathcal{T}$.}
Recall that an incision duplicates nodes $f$ in the current level to new duplicates in the next level. 
We can view this process as splitting nodes $f\in G^*$ into two \emph{node-parts},  and splitting $f$'s edges between them. 
Note that each such node-part $f'$ is incident to a consecutive subset of edges of the node $f$; and recursively (in later incisions) $f'$ will be split into node-parts each incident to a smaller consecutive subset of edges of $f'$.
Since we duplicate incised edges, each duplicate 
of an edge incident to a node $f\in G^*$  belongs to only one node-part $f'$ of $f$.  Namely, the two edges in the incision path incident to $f$ have two copies, each in a separate node-part $f'$. We now formalize this. 

\begin{definition}[Node-parts]
\label{def: nodepart}
A node-part $f'$ of a node $f\in G^*$, defined by a pair of edges $e_1,e_2$ incident to $f$ in $G^*$, is incident to all edges between $e_1$ and $e_2$ in $f$'s clockwise order. \footnote{Including $e_1,e_2$. We follow this convention throughout the paper, unless when stated otherwise.}
\end{definition}

The next claim follows from the above.
See \cref{fig: reursive_nodepart}.

\begin{restatable}{claim}{NodesAreNodeparts}
\label{claim: nodes_are_nopdeparts}
In any level $\ell$ of $\mathcal{T}$, all nodes in annuli of level $\ell$ follow the definition of node-parts (\cref{def: nodepart}). Moreover, the consecutive subsets of edges of $f\in G^*$ that define node-parts in level $\ell$, are disjoint, except in their first or last edge. These (first and last) edges are (possibly) duplicated as a result of an incision, such that, each duplicate is incident to a distinct node-part.
\end{restatable}

\begin{proof}
We prove formally by induction on levels of $\mathcal{T}$. 
Obviously, the claim holds for level zero (the graph $G^*$). In addition, from the definition of an incision on $P$ (\cref{def: incision_P}), the claim is immediate for level one ($G^*$ after incising $P$).
Assume the claim holds for a general level $\ell-1\geq 1$, we prove for level $\ell$. 
Recall, given a non-leaf annulus $A$ defined by paths $P_k,P_j$ ($k<j$), we incise on $P_i$ ($i=\lfloor (j+k)/2\rfloor$) to obtain (node-parts of) $A$'s child annuli in level $\ell$.

For non $P_i$ nodes, no change (incision) happens, hence, the claim  holds for such nodes that participate in a level $\ell$ annulus.
Thus, it is enough to prove for nodes of $P_i$. 
A node $f$ of $G^*$ that is on $P_i$ and was not incised earlier, is incised for the first time now, producing two node-parts in different annuli of level $\ell$. The claim follows from the definition of incisions on $P_i$ \cref{def: incision_Pi}. 

We now consider a node-part $f'\in P_i$ of a node $f\in G^*$ that resulted from an incision in an earlier level. See \cref{fig: reursive_nodepart}.
If $f'$ is an internal node of $P_i$ and also on one of $P_k$ or $P_j$ (w.l.o.g. $f'\in P_k$), then it defines (at most) one node-part in level $\ell$ containing all its edges internal to $P_i$ (see \cref{def: incision_Pi}, and \cref{fig: reursive_nodepart}).
Thus, the claim holds.
Otherwise, if $f'\notin P_k,P_j$ and was incised before, then, it lies on a common subpath of $P_i$ and (copies of) $P$. By \cref{def: incision_Pi} of incisions on $P_i$, we do not incise $f'$ (the claim holds), except when $f'$ is an endpoint of the internally-disjoint subpath of $P_i$ from $P$. 
In that case, $f'$ is split to (at most) two node-parts, defined by two consecutive subsets of edges of $f'$ that overlap with exactly one edge
(see \cref{def: incision_Pi}, and the node $f'=p_i^0$ in \cref{fig: reursive_nodepart}). I.e., the claim still holds.
\end{proof}

The next corollary directly follows from the above proof, and is important to the distributed implementation.

\begin{corollary}
    \label{cor: residual_nodeparts}
   Each node $f\in G^*$ has several node-parts\footnote{In fact up to four, this follows from the definitions of incisions on $P$ and $P_i$. However, we do not prove neither use this fact.}  in any given level of recursions. Some of which do not proceed to the next level's annuli. E.g. if their annuli were discarded.
    I.e., in each level of recursion, a node $f\in G^*$ may have some (residual) node-parts that do not participate in any annulus. 
    \end{corollary}

Finally, Since the nodes (node-parts) of an annulus are incident to their own copies of edges that define them, we obtain the following claim.

\begin{restatable}{claim}{DisjointAnnuli}
\label{claim: disjoint_annuli}
    In any given level of $\mathcal{T}$, all annuli are connected and disjoint. 
   \end{restatable}

\begin{proof}
    We prove by induction on levels of $\mathcal{T}$. 
    Obviously, the claim holds for level zero ($G^*$). In addition, from the definition of an incision on $P$ (\cref{def: incision_P}), the claim is immediate for level one ($G^*$ after incising $P$). I.e., there is only one connected annulus in the level. The annulus is connected because the incision over $P$ duplicates $P$, but both its duplicates are connected to the other (via the faces $s^*,t^*$).  Thus, all nodes of the annulus (that were reachable to each other via $P$) are reachable to each other via the cycle composed of the two copies of $P$ and the edges of the faces $s^*$, $t^*$ of $G^*$.

    We assume the claim holds for level $\ell-1\geq 1$ and prove for level $\ell$. Let $A$ be an annulus of level $\ell-1$, then, it is enough to show that the claim holds for the child annuli  of $A$ (if any).
    Assume that $A$ is sandwiched between paths $P_k,P_j$ ($k\leq j$), and assume $k+1<j$ (otherwise $A$ is a leaf of $\mathcal{T}$). Thus, $A$  might have child annuli after incising $P_i$ for $i=\lfloor(k+j)/2\rfloor$.
    Let $P_i'$ be the subpath of $P_i$ that gets incised. I.e., the maximal subpath of $P_i$ which is internally disjoint from copies of $P$ (see \cref{def: incision_Pi}).
    Then, assume $P_i'$ does not intersect $P_k$, thus, one child annulus of $A$ is sandwiched between $P_k,P_i'$ (otherwise, $P_i',P_k$ do not define an annulus, and the claim holds trivially). The case of $P_j,P_i'$ is symmetric. We prove that this child  annulus is connected and is disjoint from the other child of $A$ (if any).
    
    The connectivity property is simple. Denote the endpoints of $P_i'$ by $p_x^0,p_y^1$. Note that $P_i'$ can be completed to a cycle, denoted $C$, if one adds the edge $(p_x^0,p_y^1)$. Adding this edge does not violate planarity, because it can be embedded inside the face that contains the copies of $P$. We stress that we add this edge for the sake of the proof only (i.e., this is not simulated in the algorithm). As known, any cycle in a planar graph is a {\em separating cycle}~\cite{KleinM_book}. I.e., the  removal of the cycle disconnects the graph into two, its interior and its exterior. In particular, nodes of $C$ are exactly the nodes of $P_i'$. Hence, incising $P_i'$ is the same operation as incising $C$ and then removing the duplicates of the artificial edge $(p_x^0,p_y^1)$. This obviously proves the disjointness argument because duplicates of $C$ are node-disjoint, and each duplicate is connected either to edges that are only strictly internal or only strictly external to $C$. Obviously, removing the duplicate of the edge $(p_x^0,p_y^1)$ to get back the child annulus of $A$ does not weaken the disjointness argument.
    
    The connectivity property is proved via a similar argument.
    Note that, the annulus $A$ after adding the artificial edge is connected ($A$ is connected). In particular,nodes of $C$ are connected to all nodes internal to $C$ via the interior of $C$. Otherwise, there are two internal nodes $f,g$ that are connected via an external node $h$ and without any node of $C$. This contradicts $C$ being separating. An analogous argument holds for the exterior of $C$.
    Thus, incising $C$ keeps its interior (exterior) connected via $C$ itself, and therefore, the child annuli of $C$:
    Removing the duplicate of the edge $(p_x^0,p_y^1)$ to get back the child annuli of $A$ does not weaken the argument, because, the endpoints of the edge $(p_x^0,p_y^1)$ remain connected via the copy of the path $P_i'$ that results from the incision. \qedhere
    \end{proof}

\medskip
\noindent
{\bf Correspondence to $G$ (face-parts).}
In terms of the primal graph, a face $f$ of $G$ (corresponding to a node) $f\in G^*$ gets split into smaller and smaller parts whenever $f$ is split into node-parts. Each associated with a connected subgraph of the face $f$ in $G$, called, {\em face-part} (See \cref{fig: pf_connected_in_G}):

\begin{figure}[htb]
    \centering
    \includegraphics[width=0.3\linewidth]{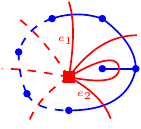}
    \caption{A (dual) node-part $f'$ of $f\in G^*$. Its incident edges (between $e_1$ and $e_2$, clockwise) in solid {\color{red}red}. Dashed {\color{red}red} edges are in $f$ and not in $f'$.  The corresponding (primal) face-part $p(f')$ is solid {\color{blue}blue}. Dashed blue are edges of the face $f$ that are not in $p(f')$.
\label{fig: pf_connected_in_G}} 
\end{figure}

\begin{claim}[Face-parts]
\label{claim:pf}
A node-part $f'$ of a node $f \in G^*$ corresponds to a connected subgraph of the face $f$ in $G$, that we denote as the {\em face-part} $p(f')$. I.e., $p(f')$ contains all edges of $G$ whose duals are incident to $f'$. 
\end{claim}

The proof follows from the formal definitions of faces and duality (see e.g., Chapter 4 of~\cite{KleinM_book}). Intuitively, the face-part $p(f')$ is a simple path (or a simple cycle when $e_1=e_2$) if the face $f$ is a simple cycle in $G$. Generally, $p(f')$ can be a tree since $f$ can be a simple cycle with attached trees in its interior.

Our goal is to perform {\em aggregations} on face-parts, 
which will enable us to simulate minor-aggregation algorithms simultaneously on all annuli in a given level of Reif's recursion. This extends an existing tool for aggregations on faces (as opposed to face-parts) by~\cite{GP17}, which enables minor-aggregation algorithms on the graph $G^*$ (without incisions)~\cite{planardistributedmaxflow24}.\footnote{We remark that~\cite{planardistributedmaxflow24} defined a different notion of face-parts in a different section of their paper and for a different purpose. However, their minor-aggregation simulation only works for the graph $G^*$ with nodes that correspond to faces (and not to face-parts).}
This will be possible because face-parts  overlap in a restricted manner:
\begin{claim}
\label{claim: edge_in_four_p(f)}
An edge $e\in G^*$ has at most four copies in any given level of $\mathcal{T}$. I.e., $e\in G$ participates in at most four face-parts (two for each face of $G$ that contains $e$).
\end{claim}

\begin{proof}
Every (dual) edge $e$ is incident to two dual nodes in $G^*$. Since $e$ is incised at most once (\cref{lem: incise_edge_once}), then $e$ can participate only in two node-parts (of each endpoint). I.e., for each endpoint $f$ of $e$ after the incision, a copy of $e$ is created and gets incident to one node-part of $f$; one whose incident edges start with $e$ and the other whose incident edges end with $e$ (see \cref{def: incision_P}, and \cref{def: incision_Pi} where incisions are discussed). The correspondence to face-parts follows from Claim~\ref{claim:pf}.
\end{proof}

\begin{restatable}{claim}{PartsOfNodeparts}
\label{fact:nodeparts of nodeparts}
In an annulus $A$, for each node-part $g$ of a node $f\in G^*$ there exists exactly one node-part $h$ of $f$ in the parent of $A$ in $\mathcal T$ s.t. $g$ is defined by a consecutive subset of edges of $h$.
I.e.,  $p(g)$ is a subgraph of $p(h)$.
\end{restatable}

\begin{proof}    
Let $h$ be a node-part of $f\in G^*$ in an annulus $A'$. We assume $h$ gets incised (otherwise the claim is trivial). In all cases of an incision on $h$ (see \cref{def: incision_P} and \cref{def: incision_Pi}) the resulting node-parts $g$ in child annuli $A$ of $A'$, are defined by subsets of consecutive edges of $h$ (thus of $f$). 
Note, $g$ cannot be defined by a subset of consecutive edges of any other node-part $h'\neq h$ of $f$ in $A'$, because, the sequences of edges that defines $h$ and $h'$ do not share a sequence of (ordered) edges (Claim~\ref{claim: nodes_are_nopdeparts}). 
The claim regarding face-parts follows from the correspondence to face-parts (Claim~\ref{claim:pf}). 
\end{proof}

\section{Distributed Reif}
\label{sec: distributed_reif}
The main building block we need in order to simulate Reif's algorithm distributively is a method for running an approximate SSSP algorithm  (such as~\cite{RozhonGHZL22_shortestpaths, GHSYZ22}) simultaneously on all annuli of a given level $\ell$ of $\mathcal{T}$.
In this section we show how to achieve this for any {\em Minor-Aggregation} algorithm~\cite{GHSYZ22,GZ22} (i.e., not only approximate SSSP, but any recursive algorithm that is based on edge-contractions).

Notice that even the most basic unit such as a single node-part $f$ in an annulus, is not an actual entity in the communication network $G$. We therefore simulate $f$ by its corresponding face-part  $p(f)$ in $G$. 
This is challenging, because the face-parts $p(f)$: (1) overlap in vertices and edges by Claim \ref{claim: edge_in_four_p(f)},  (2) may have a large diameter, and (3) are connected to other incident face-parts.
To overcome this, we design a certain virtual graph $\hat{G}_\ell$ that is closely related to $G$, and that at the same time captures 
level-$\ell$ annuli in $\mathcal{T}$.
In particular, running a minor-aggregation algorithm simultaneously on annuli translates to running an algorithm on $\hat{G}_\ell$.
In addition, since $\hat{G}_\ell$ is closely related to $G$, we will be able to simulate $\CONGEST$ efficiently on it. This allows to  run minor-aggregation algorithms on annuli (dual subgraphs) via communication on $G$ efficiently.
We note that the SSSP algorithm of~\cite{RozhonGHZL22_shortestpaths} that we use is in fact not a minor-aggregation algorithm.  In \cref{sec: SSSP} we will show how to overcome this with a small change to $\hat{G}_\ell$.

\medskip
\noindent
{\bf Section organization.}
In \cref{sec: hat_G} we define the graph $\hat{G}_\ell$, and we prove its properties in \cref{sec:Ghatproperties}.
Then, in \cref{sec: IDs_to_face_parts} we show how to use $\hat{G}_\ell$ to learn the needed distributed knowledge of a level $\ell$ in $\mathcal{T}$. E.g. IDs of annuli and IDs of node-parts.
Finally, in \cref{sec:Minor-aggregation}, we show that we can simulate any minor-aggregation algorithm simultaneously on all annuli of the same level $\ell$ of $\mathcal{T}$, in a near-optimal round complexity.

\subsection{Definition of the Graph $\hat{G}_\ell$}
\label{sec: hat_G}
For a level $\ell$ of $\mathcal{T}$ we define a graph $\hat{G}_\ell$ that: (1) maps every face-part $p(f)$ to a unique vertex disjoint component so that we can run algorithms on all of them without worrying about them overlapping, (2) captures the topology of level-$\ell$ annuli (represents node-parts and the edges incident to them), and (3) can be simulated efficiently on $G$ itself.
Concretely, in this graph $\hat{G}_\ell$, each $p(f)$ is mapped to a unique (vertex- and edge-disjoint) {\em path} $\hat{p}(f)$ (or to a cycle when $p(f)$ is an entire face of $G$). 

\begin{figure}[htb]
\centering
\begin{subfigure}[t]{0.46\textwidth}
  \centering
   \includegraphics[width=0.7\linewidth]{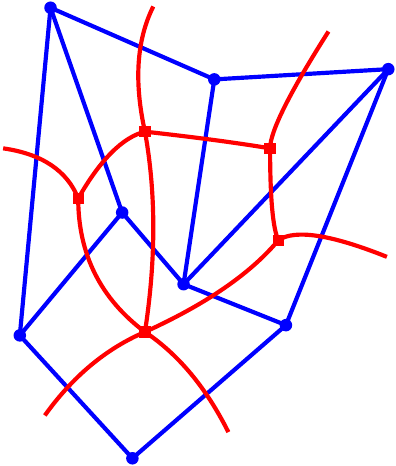}
\end{subfigure}
\hspace{0.3 cm}
\begin{subfigure}[t]{0.46\textwidth}
  \centering
   \includegraphics[width=0.9\linewidth]{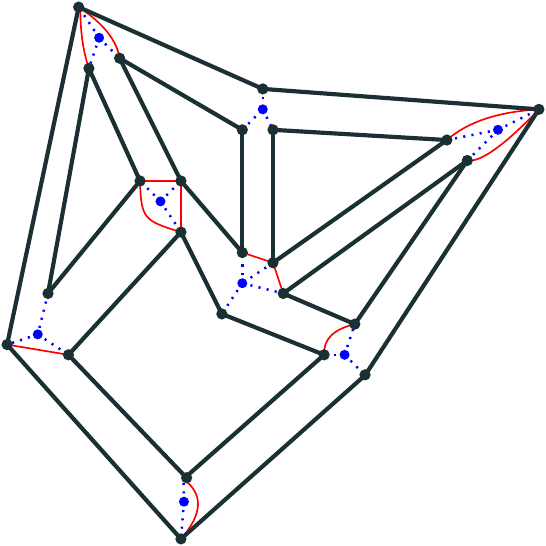}
  \end{subfigure}
    \caption{ 
    Left: The primal graph $G$ ({\color{blue} blue}) and its dual $G^*$ ({\color{red} red}). For clarity, the ({\color{red} red}) node corresponding to the external face is omitted.
   Right: The graph $\hat{G}_0$. {\color{blue} Blue} vertices are star-center copies of $v$. {\bf Black} vertices are copies $v_f$ of $v$. $\ER$, $\ES$, and $\EC$ edges are {\bf black}, dashed {\color{blue} blue}, and {\color{red} red} (respectively).     \label{fig: hat_g_base_case}}
\end{figure}

We begin with the special case of $\ell=0$. In this case, there is only one annulus (the graph $G^*$ itself) and each $p(f)$ is simply an entire face of $G$.
This case follows from~\cite{GP17}, where they define the graph $\hat{G}_0$ (they called it $\hat{G}$) that is defined as follows (see \cref{fig: hat_g_base_case}):  
Every vertex $v\in G$ has a copy vertex $v_f$ in $\hat{G}_0$ for each face $f$ that $v$ participates in. 
The edges of $\hat{G}_0$ are of three types which we call 
$\ER,\ES,\EC$.

\medskip \noindent 
$\ER$: 
For every face $f$ of $G$, for every two vertices $u,v$ that lie consecutively on $f$ there is an edge $(u_f,v_f)$ in $\ER$. 
Notice that every such edge $(u,v)$ in $G$ belongs to two faces, $f$ and $g$. Hence, we also have the edge $(v_g,u_g)$  in $\ER$. 
This way, the edges in $\ER$ corresponding to every face $f$ constitute a (disjoint) cycle $\hat{p}(f)$ that simulates $f$ in $\hat{G}_0$. 

\medskip \noindent
$\ES$: 
With the edges $\ER$, the graph $\hat{G}_0$ is a disconnected graph that is merely a union of disjoint cycles. To make $\hat{G}_0$ connected and of diameter $O(D)$, the edges $\ES$ are added to $\hat{G}_0$. Namely, an additional copy of $v$ (simply denoted by $v$) is added to $\hat{G}_0$. Each copy $v_f$ of $v$ is connected to $v$ by an edge $(v,v_f)$ in $\ES$, creating a star-shaped graph in $\hat{G}_0$ for each $v\in G$. The center of the star is $v$, and the {\em star edges} are $(v,v_f)$.
 The choice of a star (and not a clique for example) keeps $\hat{G}_0$ planar. 

\medskip \noindent
$\EC$: 
The last set of edges $\EC$ (which was implicit in~\cite{GP17} and discussed more explicitly by~\cite{planardistributedmaxflow24}) is  
intended to mimic the topology of $G^*$. Namely, it
connects the cycles $\hat{p}(f),\hat{p}(g)$ with an edge $\hat{e}=(v_g,v_f)$ iff there is an edge $e^*=(f,g)$ in $G^*$.  

In \cite{GP17, planardistributedmaxflow24}, it was shown that we can simulate minor-aggregation algorithms on $G^*$ by running them on $\hat{G}_0$. Specifically, because (1) the topology of $G^*$ (captured by $\EC$ edges) allows translating problems on $G^*$ into problems on $\hat{G}_0$, which can then be solved straightforwardly on $\hat{G}_0$ because (2) $\ER$ edges turn the faces of $G$ into disjoint cycles in $\hat{G}_0$, and efficiently because (3) $\ES$ edges keep $\hat{G}_0$ planar and of small diameter.

We now state the properties of $\hat{G}_\ell$ formally. After which, we will extend the above intuition on $\hat{G}_0$ to $\hat{G}_\ell$,  follow it by the formal definition of $\hat{G}_\ell$, and proofs for its properties.

\begin{theorem}
\label{thm: hat_G}
There exists a connected planar graph $\hat{G}_\ell$ of diameter $O(D)$, satisfying:
    \begin{enumerate}
        \item Each face-part $p(f)$ of a node-part $f$ that appears in an annulus of level $\ell$ maps to a unique path\footnote{To be concise, we refer to $\hat{p}(f)$ as a path even when $p(f)$ is an entire face of $G$ (and then $\hat{p}(f)$ is in fact a cycle).}  $\hat{p}(f)$ in $\hat{G}_\ell$ 
        such that the different $\hat{p}(f)$'s in level $\ell$ are vertex disjoint.
         In particular, every vertex $v\in p(f)$ has a copy on $\hat{p}(f)$.
        
        \item Two paths $\hat{p}(f), \hat{p}(g)$ are connected with a unique edge $\hat{e}$ in $\hat{G}_\ell$ for each edge 
        in an annulus of level $\ell$ in which the node-parts $f$ and $g$ are connected.

        \item Each annulus $A$ of level $\ell$ corresponds to a connected unique subgraph $\hat{A}$ of $\hat{G}_\ell$, s.t. a node-part $f$ is in $A$ iff $\hat{p}(f)$ is a subgraph of $\hat{A}$. 
        
        \item Assuming incisions that define level $\ell$ annuli are distributively stored (defined next), then $\hat{G}_\ell$ is constructed in $O(1)$ communication rounds on $G$.

        \item A $\CONGEST$ round on $\hat{G}_\ell$ can be simulated within $O(1)$ $\CONGEST$ rounds in $G$.

    \end{enumerate}
     
\end{theorem}

\begin{definition}[Level $\ell$ distributively stored incisions]
\label{def: distributed_storage_incisions}
    We say that the set of incision paths which define level $\ell>0$ of $\mathcal{T}$ - i.e., paths $P_i$ computed in previous levels in addition to the path $P$ - is distributively stored, if
    each primal vertex $v\in G$ knows its incident edges whose duals appear on the incision paths.
    In addition, each of $s,t$ knows its incident edge whose dual is an edge added to an endpoint of $P$ (see \cref{def: incision_P}, and \cref{remark: incision_P} later in this section). 
\end{definition}

Compared to $\hat{G}_0$, our $\hat{G}_\ell$ aims to simulate minor-aggregation algorithms simultaneously on all level-$\ell$ annuli (not only $G^*$). 
I.e., we construct one graph $\hat{G}_\ell$ that captures the structure of all level $\ell$ annuli together.   
Thus, to represent node-parts and the edges incident to them, we need to deal with face-parts and (not only) faces of $G$.
This is more complicated than (the already somewhat complicated) $\hat{G}_0$: Not only do faces of $G$ overlap, but now also face-parts of the same face overlap. Therefore, one edge of $G$ participates in potentially four face-parts (see Claim \ref{claim: edge_in_four_p(f)}).   

The generalization of $\hat{G}_0$ to $\hat{G}_\ell$ follows a similar intuition to $\hat{G}_0$. 
Essentially, we want $\hat{G}_\ell[\ER]$ to contain a vertex-disjoint path $\hat{p}(f)$ if $p(f)$ is a face-part of $G$. 
To create these paths, a vertex $v\in G$ has a copy $v$ for each face-part that contains it, and an edge $e\in G$ has up to four $\ER$ copies in $\hat{G}_\ell$ (instead of just two in $\hat{G}_0$) - a duplicate per face-part that contains $e$ (Claim~\ref{claim: edge_in_four_p(f)}).
More concretely, in $\hat{G}_0$ we had a single copy $v_f$ for each face $f$ of $G$ that contains $v$. In $\hat{G}_\ell$ we have two copies $v_f^L,v_f^R$, one for each face-part of $f$ that contains $v$ (when the node $f$ in $G^*$ is split into two node-parts in level $\ell$).

Consequently, we add more edges to the star ($\ES$ edges) of $v$ in order to keep the diameter of $\hat{G}_\ell$ small while preserving its planarity after adding the new copies of $v$. Thus, we extend the star from $\hat{G}_0$ of $v$ to a depth-two tree of $v$ in $\hat{G}_\ell$. We do that by adding the edges of the form $(v_f,v_f^L)$ and $(v_f,v_f^R)$ to $\ES$ in $\hat{G}_\ell$ (in addition to the edges of the form $(v,v_f)$ as in $\hat{G}_0$).
Next, similarly to $\hat{G}_0$, there is a set $\EC$ of edges that mimics the topology of annuli of level $\ell$, in the sense that there is an $\EC$ edge connecting $\hat{p}(f),\hat{p}(g)$ iff there is an edge connecting $(f,g)$ in an annulus of level $\ell$ of $\mathcal{T}$. 
Having those $\EC$ edges and the paths (composed of $\ER$ edges) that correspond to node-parts of level-$\ell$ annuli, means that in $\hat{G}_\ell[\ER \cup \EC]$ we have a disjoint connected component $\hat{A}$ for every level-$\ell$ annulus $A$.

\medskip
\noindent
{\bf Definition of \boldmath$\hat{G}_\ell$.}
\label{par: def_hat_G}
We proceed to define $\hat{G}_\ell$ excluding an adjustment related to the endpoints of $P$ (see \cref{def: incision_P}), which is discussed towards the end of the current subsection (\cref{remark: incision_P}). 
In the following, we assume that the incisions of level $\ell$ are distributively stored as in \cref{def: distributed_storage_incisions}.
We say that an edge $e$ of $G$ \hl{maps}\footnote{We emphasize important sentences in the definition also by color.} to an incised edge of level $\ell$ iff its dual participates in $P$ or in some $P_i$ that bounds a level-$\ell$ annuli.

\medskip
\noindent
{\bf Vertices.}
We start by discussing the vertex set $\hat{V}$ of  $\hat{G}_\ell$. Intuitively, we want a vertex $v_f$ for each vertex $v\in G$ that lies on a face $f$ of $G$. However, $v$ has no way of detecting which faces does it lie on locally (actually, the whole purpose of $\hat{G}_\ell$ is to allow us to do that).
To deal with that, note, from the local view point of $v$, each two consecutive incident edges correspond to a face of $G$ that contains $v$. Hence, what was denoted as $v_f$ is now formally denoted by $v_i$, where, $v_i$ is associated to the $i$th pair of consecutive edges incident to $v$. We denote those edges by $e_i,e_{i+1}$ (\hl{indices indicate edges' ordering in $v$'s embedding and are taken modulo $\deg(v)$}). I.e.,  $v_i$ can be thought of intuitively as $v_f$ if the face $f$ contains the $i$th pair of edges incident to $v$. 
See for example See \cref{fig: local_hatG_no_incision}, where $v_1$ corresponds to the face containing edges $e_1,e_2$.

Since the number of pairs of consecutive edges incident to $v$ is exactly the degree of $v$. 
Then, $v$ locally creates the copies $D(v)= \{v_0,\ldots, v_{\deg(v)-1}\}$ in $\hat{G}_\ell$. This results in a copy of $v$ in $D(v)$ for each pair of two consecutive edges incident to $v$ (for each time  $v$ appears on a face).\footnote{Note that $v$ might have multiple pairs of consecutive edges that are contained in the same face, and that is why we say {\em for each time $v$ appears on a face}.} See \cref{fig: local_hatG_no_incision}.
As mentioned earlier, $v$ has two more types of copies, the copy $v$, and the copies $v_f^L,v_f^R$ (which now we denote by $v_i^L,v_i^R$ following the above reasoning). Let $V_S \coloneqq V(G)$.
The set \hl{$\hat{V}$ is defined to be the union of}
$$
V_S, \bigcup_{v\in V_S}D(v), \text{ and } \bigcup_{v\in V_S}(D(v)\times \{L,R\})
$$
More precisely, \hl{the vertices $V_S$ and $D(v)$ (for all $v\in G$) are always present in $\hat{G}_\ell$.}
However, only a subset of $D(V)\times\{L,R\}$ is  present in $\hat{G}_\ell$.
To explain this, we first give an intuition for the set $D(v)\times \{L,R\}$. 
Since it is possible that $f$ is split into face-parts as a result of incising the dual node $f$ into node-parts (\cref{def: nodepart}, Claim~\ref{claim:pf}), in which case for example, each edge in the pair of edges $e_i,e_{i+1}$ corresponding to $f$ might belong to a different face-part. I.e., if we want to map those face-parts to vertex-disjoint components of $\hat{G}_\ell$, then we must allow a scenario where each of $e_i,e_{i+1}$ is incident to distinct vertices in $\hat{G}_\ell$. Therefore, additional copies of $v$ other than $v_i$ that correspond to those edges are needed.
To allow such a situation (or others related to face-parts of $f$ and this pair of edges), two additional copies of $v$ are created. The copy $v_i^L$ is associated (only) with $e_i$ and the copy $v_i^R$ is associated (only) with $e_{i+1}$, as opposed to $v_i$ that is associated with both. These three vertices allow different scenarios for the $i$th pair of edges of $v$.
Hence, these $D(v)\times \{L,R\}$ vertices are the core of the generalization from $\hat{G}_0$ to $\hat{G}_\ell$, and together with the vertices $v_i$  they map each face or face-part $p(f)$ that $v$ participates in to a unique path or cycle $\hat{p}(f)$ in $\hat{G}_\ell$.

\begin{figure}[htb]
\centering
\begin{subfigure}[t]{0.3\textwidth}
  \centering
   \includegraphics[width=1\linewidth]{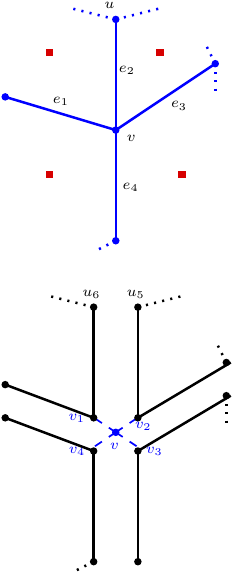}
   \caption{\label{fig: local_hatG_no_incision}}
\end{subfigure}
\hspace{0.2 cm}
\begin{subfigure}[t]{0.3\textwidth}
  \centering
   \includegraphics[width=1\linewidth]{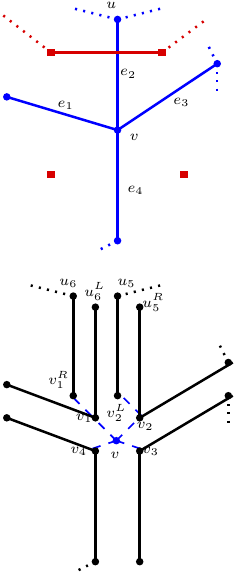}
      \caption{\label{fig: local_hatG_one_edge_incision}}
\end{subfigure}
\hspace{0.2 cm}
\begin{subfigure}[t]{0.3\textwidth}
  \centering
   \includegraphics[width=1\linewidth]{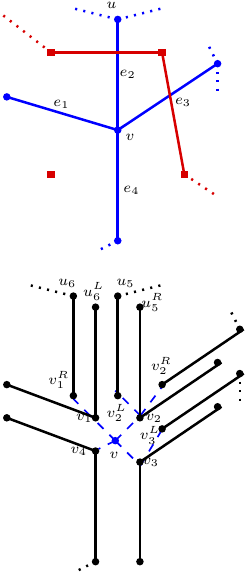}
    \caption{\label{fig: local_hatG_two_edge_incision}}
   \end{subfigure}
\caption{Top row: the primal graph $G$ in {\color{blue} blue} and the nodes of the dual $G^*$ in {\color{red} red}. Only the incision path appears in $G^*$.
Three cases related to the vertex $v$ are demonstrated, left to right:
(a) no incident edge of $v$ is dual to an incision edge, (b) only the edge $e_2$ is dual to an incision edge, and (c) both $e_2,e_3$ are dual to incision edges.
Bottom row: the corresponding graph $\hat{G}_\ell$ (from $v$'s point of view) for each case. {\bf Black} edges are edges of $\ER$. {\color{blue} Blue} edges and vertices are edges of $\ES$ and vertices of $V_S$. $\EC$ edges are omitted.
\label{fig: local_incision}
}
\end{figure}

Now,  
\hl{we define which vertices of $D(v)\times\{L,R\}$ are present.}: 
    If no edge incident to $v$ is dual to an incised edge, then no vertex $v_i^R$ or $v_i^L$ is present. See \cref{fig: local_hatG_no_incision}.
    If an incident edge $e_{i}$ or $e_{i+1}$ is incised, then it participates in two face-parts of the face $f$ corresponding to $e_i,e_{i+1}$ (Claim~\ref{claim: edge_in_four_p(f)}). 
   Therefore, since $v_i$ is responsible for the pair $e_i,e_{i+1}$ (for $f$), then when $e_i$ is incised, we add the copy $v_{i}^L$.
   Hence, $e_i$ would be able to have two copies for $f$, one that is incident to $v_i$, and the other to $v_i^L$. Following the same reasoning, if $e_{i+1}$ is incised, then it participates in two face parts of $f$. Thus, we add the copy $v_{i}^R$, so that $e_{i+1}$ would be able to have two copies for $f$, one that is incident to $v_i$, and the other to $v_{i}^R$.
   (see \cref{fig: local_hatG_one_edge_incision} and \cref{fig: local_hatG_two_edge_incision}).
   Intuitively, from $v$'s local view point, $e_i$ is the left edge in the ordered pair $e_i,e_{i+1}$, hence the superscript $L$ in $v_i^L$. Analogously, $e_{i+1}$ is the right edge in the ordered pair $e_i,e_{i+1}$, hence the superscript $R$ in $v_i^R$.

\begin{figure}[htb]
\centering
    \begin{subfigure}[t]{0.3\textwidth}
        \centering
        \includegraphics[width=\linewidth]{Figures/primal_dual_P.pdf}
        \caption{$G$ ({\color{blue} blue}), its dual $G^*$ ({\color{red}red}), and the path $P= (p_0,p_1,p_2)$.}
        \label{fig: primal_dual_P2}
    \end{subfigure}
    \hspace{0.2cm}
        \begin{subfigure}[t]{0.3\textwidth}
        \centering
        \includegraphics[width=1\linewidth]{Figures/original_hat_G.pdf}
        \caption{The graph $\hat{G}_0$ corresponding to  $G^*$ (with no incisions).}
    \end{subfigure}
    \hspace{0.2cm}
    \begin{subfigure}[t]{0.3\textwidth}
    \centering
    \includegraphics[width=0.8\linewidth]{Figures/dual_incision_P.pdf}
        \caption{The dual graph $G^*$ after the incision along $P$.  \label{fig: dual_incision_P2}}       
    \end{subfigure}

    \vspace{0.5cm}

    \centering
    \begin{subfigure}[t]{0.45\textwidth}
      \centering
       \includegraphics[width=0.7\linewidth]{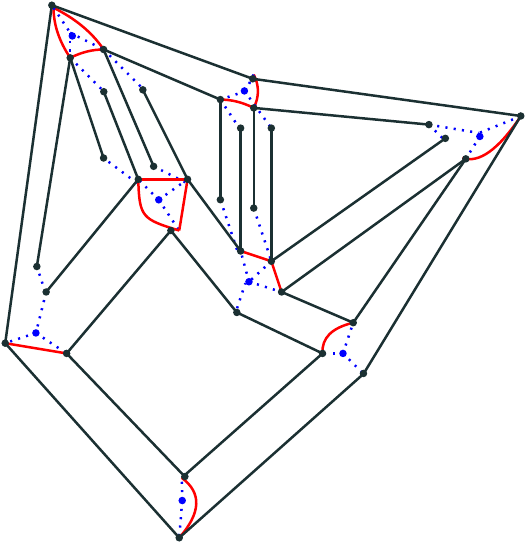}
       \caption{The graph $\hat{G}_1$ representing $G^*$ after the incision along $P$. $\ER$ edges are {\bf black}, $\ES$ edges are {\color{blue} blue}, and $\EC$ edges (mimicking $G^*$ after incising $P$) are {\color{red}red}.  
       E.g. $\EC$ edges connecting $
       \hat{p}(p_0^0),\hat{p}(p_1^0),\hat{p}(p_2^0)$ (yellow in next subfigure) correspond to the path $P=(p_0^0,p_1^0,p_2^0)$.}
    \end{subfigure}
    \hspace{0.5cm}
    \begin{subfigure}[t]{0.45\textwidth}
    \centering
    \includegraphics[width=.7\linewidth]{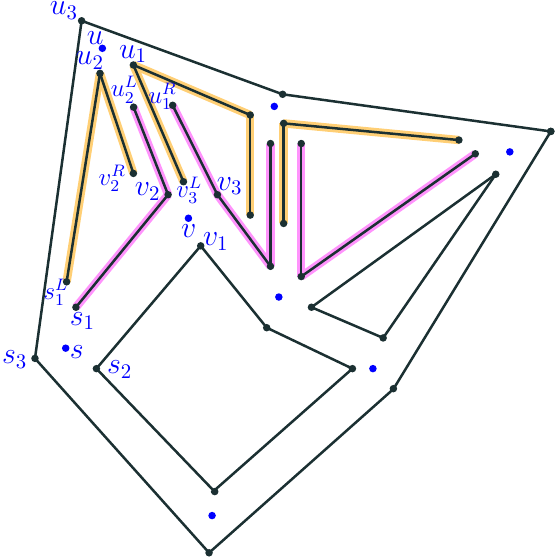}
    \caption{The graph $\hat{G}_1$ only with $\ER$ edges, and names to some vertices. 
    {\color{YellowOrange}Yellow} paths are $\hat{p}(p_0^0),\hat{p}(p_1^0),\hat{p}(p_2^0)$, corresponding to the nodes $p_i^0$ from subfigure (b). 
    Similarly, {\color{magenta}pink} paths are $\hat{p}(p_i^1)$s that correspond to $p_i^1$s. 
    The vertex $s$ is incident to an endpoint-case edge $e_0$ whose copies in $\hat{G}_1$ are $(s_1^L,u_2)$ and $(s_3,u_3)$.
    }
    \end{subfigure}
\caption{\label{fig: hat_G_one} }

\end{figure}

\medskip
\noindent
{\bf Edges.} 
Each edge $e$ of $G$ has up to four duplicates in $\ER$, a duplicate for each face-part $p(f)$ that contains the edge $e$ (Claim \ref{claim: edge_in_four_p(f)}).
The purpose of $\ER$ is to have different $p(f)$'s map to unique paths $\hat{p}(f)$ (see \cref{fig: hat_G_one}).
The purpose of $\ES$ is to keep the diameter of $\hat{G}_\ell$ a constant factor larger than that of $G$, while preserving its planarity. The purpose of $\EC$ is to simulate the dual edges in level-$\ell$ annuli. Formally:

\begin{description}
\item[\boldmath$\ES$:] 
This set always contains the star edges $(v,v_i)$ for all $v\in G$ and all $v_i\in D(v)$. In addition, each of $v_i^L$ and $v_i^R$ (if present) is connected to $v_i$ with an edge $(v_i,v_i^L)$ and $(v_i,v_i^R)$, respectively.

\item[\boldmath$\ER$:] 
Let $e=(u,v)$ be the $i$th edge ($e_i$) and the $j$th edge ($e_j$) in the local embedding of $v$ and $u$ respectively. We examine two cases, namely, if $e$ corresponds to an incised edge of $G^*$ and has two duplicates in level $\ell$ or not. Note, these are the only cases, as an edge of $G^*$ is incised at most once (\cref{lem: incise_edge_once}).

If $e$ does not map to an incised edge, then $e$ is duplicated twice like in the original definition of $\hat{G}_0$ by~\cite{GP17}. 
Intuitively, $e$ is contained in two faces $f,g$ of $G$, and we want to create the edge copies $(v_f,u_f), (v_g,u_g)$ to map those faces to disjoint components $\hat{p}(f), \hat{p}(g)$. However, as mentioned earlier, $u,v$ do not know (shared IDs of) the faces that contain them. Thus, the definition of those edges needs to depend only on local information. In particular, it depends on the edges' order in their endpoints' local embedding.
 
Locally, $f$ is the $i$th face of $v$, and is right to the edge $e_i$ from $v$'s point of view. While $f$ is the $j-1$th face for $u$, and from $u$'s view point $f$ is left to the edge $e_j$. Hence, the edge $(v_f,u_f)$ is formally $(v_{i},u_{j-1})$. This edge is considered to be an edge of $\hat{p}(f)$ (the path corresponding to $f$). Symmetrically, $(v_{i-1},u_{j})$ is the formalization of $(v_g,u_g)$, and participates in $\hat{p}(g)$. 
See \cref{fig: local_hatG_no_incision}. for an illustration from the local point of view of $v$, and \cref{fig: hat_G_one} for a global illustration.

If $e$ does map to an incised edge in level $\ell$, then it participates in possibly four face-parts (Claim~\ref{claim: edge_in_four_p(f)}). Thus, we create four copies of it, one per face-part that contains it. Intuitively, if $e$ is contained in faces $f,g$, then it is contained in two face-parts of $f$, and two face-parts of $g$. Moreover, it is the first edge of a face-part $f^0$ of $f$, the last edge of a face-part $f^1$ of $f$, the first   edge of a face-part $g^1$ of $g$, and finally, the last edge of a face-part $g^0$ of $g$.
Thus, informally, we would have four copies of $e$: $(v_f^L,u_f)$, $(v_f,u_f^R)$, $(v_g,u_g^L)$, $(v_g^R,u_g)$, which correspond receptively to the aforementioned cases.
Following the reasoning above that: $v_{i},u_{j-1}$ correspond to $v_f,u_f$ (respectively), and $v_{i-1},u_{j}$ corresponds to $v_g,u_g$ (respectively). Then formally we add the following edges to $\ER$: $(u_{j-1}^R,v_{i})$, $(u_{j-1},v_{i}^L)$, each corresponding to a distinct face-part of $f$. In addition, we add the edges  $(u_{j}^L,v_{i-1})$ $(u_{j},v_{i-1}^R)$ each corresponding to a distinct face-part of $g$.
Finally, note that since $e$ is incised then by definition of $\hat{V}$, we indeed have the vertex copies $v_{i-1}^R, v_{i}^L \in D(v)\times \{L,r\}$, $ u_{j-1}^R, u_{j}^L \in D(u)\times \{L,r\}$ that the edge $e$ is associated with. Of course, in addition to the (always present) vertex copies $v_{i-1},v_{i}\in D(v)$, $u_{j-1},u_{j} \in D(u)$. 

For a local example from the view point of $v$, see \cref{fig: local_hatG_no_incision} where no incident edge is incised and \cref{fig: local_hatG_one_edge_incision} where the edge $e_2$ is incised, and the face $f$ to its right is split to two face-parts. In this example, $e$ is $e_2$ for $v$ and $e_6$ for $u$, and one face-part of $f$ contains the copy  $(v_2^L,u_5)$ of $e$ as the first edge, while the second contains the copy $(v_2,u_5^R)$ of $e$ as the last edge.
For a more global illustration (of the entire graph $\hat{G}_\ell$), see \cref{fig: hat_G_one}, where $f=p_1$ and $g=p_0$.

\item[\boldmath$\EC$:]
Let $e=(u,v)$ be the $i$th edge and the $j$th edge in the local embedding of $v$ and $u$ respectively. 

If $e=(u,v)$ does not map to an incised edge, and its dual is $e^*=(f,g)$. Then, intuitively we want one edge between the two copies of $v$ (or $u$), one copy that corresponds to $f$, and the other to $g$.  That is, an arbitrary one of the edges that were intuitively denoted $(v_f,v_g),(u_f,u_g)$ are added to simulate $e^*$.
Formally, one of the edges $(u_{j-1},u_{j}),(v_{i}, v_{i-1})$ is added to $\EC$ ($u$ and $v$ decide arbitrarily, say the vertex with the smaller ID adds the edge between its copies). 
This is because if $(u,v)$ belongs to faces $f$ and $g$, then both $u_{j}$ and $v_{i-1}$ are in $\hat{p}(g)$, while both $u_{j-1},v_{i}$ are in $\hat{p}(f)$. Thus, it is enough to add one $\EC$ edge that connects $\hat{p}(g)$ and $\hat{p}(f)$. 

If $e=(u,v)$ maps to an incised edge, then it corresponds to two dual edges $e^*_0=(f^0,g^0)$ and $e^*_1=(f^1,g^1)$, where $f^k$ and $g^k$ (for $k\in \{0,1\}$) are node-parts of $f,g$ (respectively). Thus, we want to add the two edges that we intuitively denote $(u_{f^0},u_{g^0})$ and $(v_{f^1},v_{g^1})$. The first would simulate $e^*_0$, and the second would simulate $e^*_1$.
Formally, we add both edges $(u_{j},u_{j-1}),(v_{i-1}, v_{i})$.
This is because if $(u,v)$ belongs to faces $f$ and $g$ of $G$, then $u_{j-1},u_{j}$ are in $\hat{p}(f^0),\hat{p}(g^0)$ (respectively) s.t. the node-parts $f^0,g^0$ are incident in the annulus, and $v_{i},v_{i-1}$ are in $\hat{p}(f^1),\hat{p}(g^1)$ (respectively) s.t. the node-parts $f^1,g^1$ are incident in the annulus.
For  example, see the $\EC$ edge $(u_1,u_2)$ of $u$ representing $(p_0^0,p_1^0)$, and the $\EC$ edge $v_2,v_3$ of $v$ representing $(p_0^1,p_1^1)$ in \cref{fig: hat_G_one} (in that case $g= p_0$, $f= p_1$).

The correspondence between face-parts $f^1,g^1$ to $v$ and $f^0,g^0$ to $u$ in the definition is w.l.o.g. and merely a matter of notation, in order to draw the distinction that each of $u,v$ handles a different edge of $e^*_0,e^*_1$. That is, there is no known shared identifiers for $f^1,g^1,f^0,g^0$. This fact is important in order for $u,v$ to construct $\hat{G}_\ell$ locally.
At first glance this may seem not symmetric, but this is not the case and the definition is totally symmetric. Simply because $u,v$ add the (syntactically symmetric) edges $(u_{j},u_{j-1}),(v_{i-1}, v_{i})$ obliviously. 
In more detail, the semantic  symmetry is as follows: $v$ is responsible for the face-parts that end and start with its $i-1$th and $i$th edges (respectively), which correspond to face-parts $g^1$ and $f^1$ (respectively). 
Similarly, $u$ is responsible for the face-parts that end and start with its $j$th and $j-1$th edges (respectively), and those correspond to  face-parts $g^0$ and $f^0$ (respectively).
\end{description}

\noindent
{\bf Adjustments related to \boldmath$P$'s endpoints.}
\label{par: def_hatG_endpoint_case}
Recall (\cref{def: incision_P}) that to define the incision on $P$'s endpoints, we add to $P$ an extra edge $e_0$ incident to $p_0$. In $G$, $e_0$ is incident to the vertex $s$.
Analogously for $P$'s last endpoint $p_{|P|-1}$, an edge $e_{|P|-1}$ is added. See for example the endpoints $p_0,p_2$ of $P$ in \cref{fig: primal_dual_P2}. Incident endpoint edges are $e_0$ and $e_2$ (respectively) and their node-parts after the incision as demonstrated in \cref{fig: dual_incision_P2}.

\begin{remark}[Incision on endpoints of $P$]
\label{remark: incision_P}
In order for the incision over $P$ to be defined (see \cref{def: incision_P}), we need each of $P$'s endpoints, to choose an incident non-$P$ edge, called an {\em endpoint-case edge}, so that this edge together with the incident edge in $P$ define how to split the endpoint into node-parts. Recall, each of those edges get incident to only one new node-part.
We denote the chosen edge for the endpoint $p_0$ by $e_0$, and define $e_0$ to belong only to the node-part $p_0^0$ of $p_0$: Edges between $e_0,(p_0,p_1)$ in $p_0$'s embedding get incident to $p_0^0$, while edges between $(p_0,p_1),e_0$, except for $e_0$, get incident to $p_0^1$. 
See \cref{fig: dual_incision_P2}.
It is important (for Reif's correctness) that $e_0\in s^*$. Such edge $e_0$ always exist as  $P$ does not entirely contain edges of the face $s^*$ of $G^*$ (on which $p_0$) lies - this would contradict $P$ being simple. 
The definition is symmetric for $p_{|P|-1}$.
\end{remark}

In order for $\hat{G}_\ell$ to represent  level-$\ell$ annuli correctly, we make minor adjustments to the vertex and edge copies related to the endpoint-case edges in $\hat{G}_\ell$.

Note, in the definition of $\hat{G_\ell}$ we define $\ER$ edges that correspond to an edge $e$ by examining two cases, either that $e$ is not incised and participates in two faces, or it is incised and participates in four face-parts. However, for $e_0$ (and $e_{|P|-1}$), it is not incised, and participates  in exactly one face and one face-part.
Thus, to handle this, we simply combine both cases from the definition of $\hat{G}_\ell$. Namely. let $f\neq p_0$ be the other (dual endpoint) of $e_0$. Then $e_0$ is treated like any non incised edge of $f$ and is duplicated once for $f$. What needs to be handled differently is the copy of $e_0$ for $p_0$. Namely, we need to treat it as if it is the first edge of the node-part $p_0^0$ of $p_0$, without having another copy of it for  $p_0^1$.
Analogously, for $e_{|P|-1}$, it is treated as any non-incised edge in the face $f\neq p_{|P|-1}$ that contains it, but is treated  as if it is the last edge of the node-part $p_{|P|-1}^1$ of $p_{|P|-1}$, without having another copy of it for  $p_{|P|-1}^0$. In fact, in the definition of $\ER$ edges above, this terminology of {\em first edge on face-part} is used for intuition and then formalized.
We still describe the changes we do formally.

\begin{figure}[htb]
\centering
\begin{subfigure}[t]{0.45\textwidth}
    \centering
        \includegraphics[keepaspectratio, height=3.5cm]{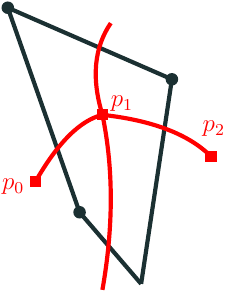}
    \hspace{0.2cm}
        \includegraphics[keepaspectratio, height=3.5cm]{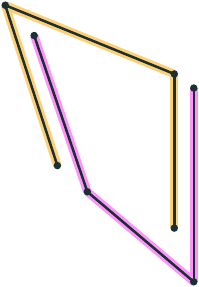}
    \caption{Left: the (non-endpoint) node $p_1$ ({\color{red} red}) and $p(p_1)$ in $G$ ({\bf black}). $p_1$'s incident edges on $P$ are $(p_0,p_1),(p_1,p_2)$. Right: the paths of $\hat{G}_\ell$ corresponding to node parts $p_1^0,p_1^1$. The {\color{YellowOrange}orange} path is $\hat{p}(p_1^0)$, and has a copy of the edges between $(p_0,p_1),(p_1,p_2)$. The {\color{magenta}pink} path is $\hat{p}(p_1^1)$, and has a copy of the edges between $(p_1,p_2),(p_0,p_1)$. Note, \hl{both} edges $(p_0,p_1),(p_1,p_2)$ appear in each of $\hat{p}(p_1^0)$ and $\hat{p}(p_1^1)$. 
      \label{fig: no_endpointcase_edge}}
  
\end{subfigure}
\hspace{0.4cm}
\begin{subfigure}[t]{0.45\textwidth}
    \centering
        \includegraphics[keepaspectratio, height=3.5cm]{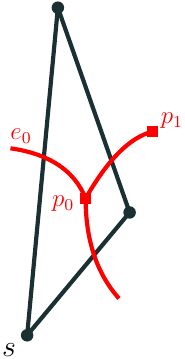}
    \hspace{0.2cm}
        \includegraphics[keepaspectratio, height=3.5cm]{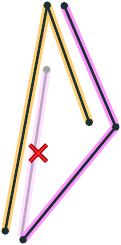}
    \caption{Left: the endpoint node $p_0$ ({\color{red} red}) and $p(p_0)$ in $G$ ({\bf black}). $p_0$ has only one incident edge $(p_0,p_1)$ on $P$. The edge $e_0$ is the added endpoint-case edge.
    Right: the paths of $\hat{G}_\ell$ corresponding to node parts $p_0^0,p_0^1$. The {\color{YellowOrange}orange} path is $\hat{p}(p_0^0)$, and has a copy of edges between $e_0,(p_0,p_1)$. The {\color{magenta}pink} path is $\hat{p}(p_0^1)$, and has a copy of edges between $(p_0,p_1),e_0$, \hl{excluding} $e_0$. Hence, \hl{only} the edge $(p_0,p_1)$ appears in both $\hat{p}(p_0^0)$ and $\hat{p}(p_0^1)$. 
     \label{fig: endpointcase_edge} }
  
\end{subfigure}
\caption{Two cases of a node, its corresponding face in $G$, and its corresponding face-parts after incision. 
(a) a non endpoint of $P$. 
(b) an endpoint of $P$. \label{fig: endpointcase_edge_cases}}
\end{figure}

Let $v\in \{s,t\}$, and $e=e_0$ if $v=s$. Otherwise, $e=e_{|P|-1}$ (when $v=t$).
In addition, let $e=(u,v)$ be the $i$th edge and the $j$th edge in the local embedding of $v$ and $u$ (respectively).
If $e$ is incised then it is handled just like any other incised edge. Otherwise, when $e$ is not incised we apply the following changes: 

\medskip
\noindent
\underline{\em Vertices.}
We first define the copy vertices of $u,v$ related to $e$ in $\hat{G}_\ell$. 
For $u$, there is no change on the definition of its copies. I.e., $u_{j-1},u_j$ are present and potentially other copies related to them, such as $u_{j-1}^L$, but not  $u_{j}^L$ and $u_{j-1}^R$. Because, those are added only if $e$ is incised, and it is not.
Similarly for $v$, we already have the copies $v_i$, $v_{i-1}$ and potentially others, but not the copies $v_{i-1}^R$ and $v_{i}^L$ for the same reasoning mentioned earlier.
we make no change besides adding the copy $v_i^L$ (see copies of $s$ in \cref{fig: hat_G_one}, subfigure (d)).
This vertex copy is responsible for the first edge of the face-part of the face right to $e$. I.e., for the copy of $e$ that corresponds to $p_0^0$ (respectively, $p_{|P|-1}^1$) if $v$ is $s$ (respectively, $t$). That is,  $v_i$ participates in $\hat{p}(p_k^{1-b})$, $v_i^L$ in $\hat{p}(p_k^{b})$, and $v_{i-1}$ in the other face  $f\neq p_k$ of $G$ in which $e$  participates, where $k=0,b=0$ iff $v=s$, and $k=|P|-1,b=1$ in case $v=t$.

\medskip
\noindent
\underline{\em Edges.} For the set $\ES$ of edges, we have no change in its definition. I.e., adding the vertex $u_{j}^L$ leads to adding the $\ES$ edge $(u_j,u_{j}^L)$.
See edges incident to copies of $s$ in \cref{fig: hat_G_one} (subfigures (c) and (d)).
In regard of $\ER$ edges, we add the edge $(v_{i}^L,u_{j-1})$, instead of $(v_{i}, u_{j-1})$ when $e$ is not an endpoint-case edge. Meaning $(v_{i},u_{j-1})$ is not present. I.e., this means that the copy of $e$ in $\hat{p}(p_0)$ (respectively, $\hat{p}(p_{|P|-1})$) now goes to $\hat{p}(p_0^0)$ (respectively, $\hat{p}(p_{|P|-1}^1)$).
Other $\ER$ edges related to copies of $u,v$ are added without any changes.
Finally, for $\EC$ edges, we add an arbitrary edge from $(u_{j-1},u_j)$ and $(v_{i-1},v_i^L)$, instead of $(u_{j-1},u_j)$ or $(v_{i-1},v_i)$.
Which connects $\hat{p}(p_0^0)$ ($\hat{p}(p_{|P|-1}^1)$) with a copy of $e$ to the path $\hat{p}(f)$ where $f$ is the other dual endpoint of $e$.
For simplicity, we always add the edge $(u_{j-1},u_j)$ and the edge $(v_{i-1},v_i^L)$ is never present. Again, see \cref{fig: hat_G_one} (for example the $\EC$ edge corresponding for $e_0$, is $(u_2,u_3)$).
Other $\EC$ edges remain without any change.

\medskip
\noindent This completes the description of $\hat{G}_\ell$.

\subsection{Properties of $\hat{G}_\ell$}\label{sec:Ghatproperties}
We now prove the properties of $\hat{G}_\ell$ in order to complete the proof of \cref{thm: hat_G}.

\subsubsection{Planarity and Diameter}

\begin{lemma}
\label{lemma: diameter_hat_G}
    The diameter of $\hat{G}_\ell$ is at most $4D$.
\end{lemma}
\begin{proof}
This follows from the fact that every edge $e=(u,v)$ of $G$ turns into a path of length 4 in $\hat{G}_\ell[\ER\cup \ES]$. If $e$ is not incised then the path is $v-v_i-u_j-u$. Otherwise, if $e$ is incised or is an endpoint-case edge, then the path is $v-v_i-v_i^L-u_j-u$ or $v-v_i-v_i^R-u_j-u$. 
Adding the edges $\EC$ can only make the diameter smaller. \qedhere
\end{proof}

\begin{lemma}
    $\hat{G}_\ell$ is planar.
\end{lemma}
\begin{proof}

For $\ell=0$, we have that $\hat{G}_0$ is planar by~\cite{GP17, planardistributedmaxflow24}. 
For $\ell>0$, we show that $\hat{G}_\ell$ is planar by showing how to transform $\hat{G}_0$ into $\hat{G}_\ell$ using four steps, each preserving planarity.
Intuitively, for each edge $e$ that maps to an incised edge: (1) we complete the set of $\EC$ to match its definition in $\hat{G}_\ell$ by adding both $\EC$ edges that correspond to $e$ (as opposed to only an arbitrarily chosen one in $\hat{G}_0$), (2) we  turn  each of the two $\ER$ edges corresponding to $e$ into two parallel edges (this creates four copies that participate in four face-parts), (3) we subdivide each such edge into one part that belongs to $\ER$ and another part that belongs to $\ES$, and (4) we handle the endpoint-case edges in a similar manner. Formally: 

\medskip
\noindent
{\bf The four steps transformation.}
Let $e=(u,v)\in G$ be the $i$th and $j$th edge in the local embedding of $v$ and $u$, respectively. 
Recall from \cref{sec: hat_G} that in $\hat{G}_0$ no edge
 is an incised edge, and there are four vertices related to $e$: $v_{i-1},v_i,u_{j-1},u_j$.  In $\hat{G}_\ell$ however, if $e$ is an incised edge, then there are the additional vertices $v_{i-1}^R,v_i^L,u_{j-1}^R,u_j^L$ related to $e$.
The transformation from $\hat{G}_0$ to $\hat{G}_\ell$ is then:

\medskip
\noindent \underline{1.  \em Adding missing $\EC$ edges.}
    In $\hat{G}_0$, one of the edges $(v_{i-1},v_i),(u_{j-1},u_j)$ is in $\EC$. This is in contrast to $\hat{G}_\ell$, where for an incised edge $e\in G$ both edges $(v_{i-1},v_i),(u_{j-1},u_j)$ are in $\EC$ (see \cref{fig: hatG_trans_EC}).
    This step, produces a graph $\hat{G}_0^1$, whose set $\EC$ of edges is identical to the the set $\EC$ of $\hat{G}_\ell$, that is, we add the missing $\EC$ edges to $\hat{G}_\ell$.
    (see definition of $\hat{G}_\ell$ in \cref{par: def_hat_G}).

\begin{figure}[htb]
\centering

\begin{subfigure}[t]{0.3\textwidth}
  \centering
\includegraphics[width=\linewidth]{Figures/primal_dual_P.pdf}
\end{subfigure}
\hspace{0.2 cm}
\begin{subfigure}[t]{0.3\textwidth}
  \centering
\includegraphics[width=\linewidth]{Figures/original_hat_G.pdf}
\end{subfigure}
\hspace{0.2 cm}
\begin{subfigure}[t]{0.3\textwidth}
  \centering
\includegraphics[width=\linewidth]{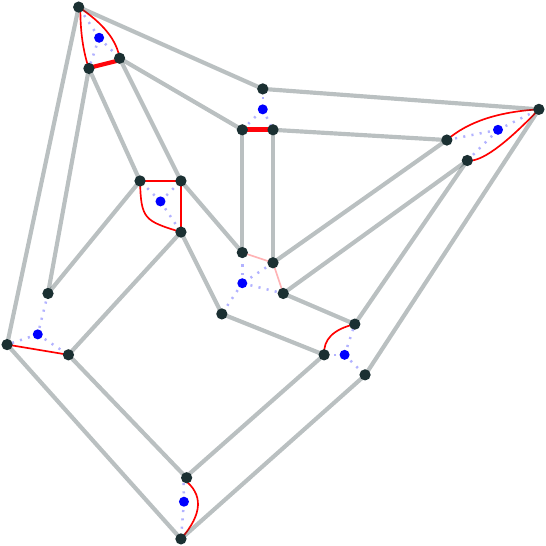}
  \end{subfigure}

\caption{Left: $G$ ({\color{blue} blue}), its dual $G^*$ ({\color{red}red}). The incision path is $P= (p_0,p_1,p_2)$, and its endpoint-case edges are $e_0,e_2$.
Middle: the graph $\hat{G}_0$. $\ER$ edges are {\bf black}, $\ES$ edges are {\color{blue} blue}, and $\EC$ edges are {\color{red} red}. 
Right: the graph $\hat{G}_0^1$  resulting from $\hat{G}_0$ after step (1) where two missing $\EC$ edges (bold {\color{red} red}) are added. \label{fig: hatG_trans_EC}}
\end{figure}

Adding those edges does not violate planarity (follows from~\cite{GP17, planardistributedmaxflow24}). In particular, 
even if all edges $(v_1,v_2), \ldots, (v_{\deg(v)},v_1)$ were added, then they create a cycle on the leaves of $v$'s star. 
One can verify easily that replacing the star of a vertex with a star on a cycle for all vertices $v$ preserves planarity. Hence, adding a subset of those edges also does not.

\medskip
\noindent   
\underline{2. \em Duplicating $\ER$ edges.}
    For every edge $e$ of $G$ there are two $\ER$ copies $(v_{i-1},u_j)$ and $(v_i,u_{j-1})$ in $\hat{G}_0^1$ (the resulting graph from the previous step).
    However, when $e$ is incised, then in $\hat{G}_\ell$ it has four $\ER$ copies. 
   To fix this, we first duplicate each of the edges $(v_{i-1},u_j),(v_i,u_{j-1})$. This results in two parallel edges $(v_{i-1},u_j)$, and another two parallel edges $(v_i,u_{j-1})$. Obviously the resulting graph, denoted $\hat{G}_0^2$ is planar. See \cref{fig: hatG_trans_ER_ES} (left).

\begin{figure}[htb]
\centering
\begin{subfigure}[t]{0.31\textwidth}
  \centering
   \includegraphics[width=\linewidth]{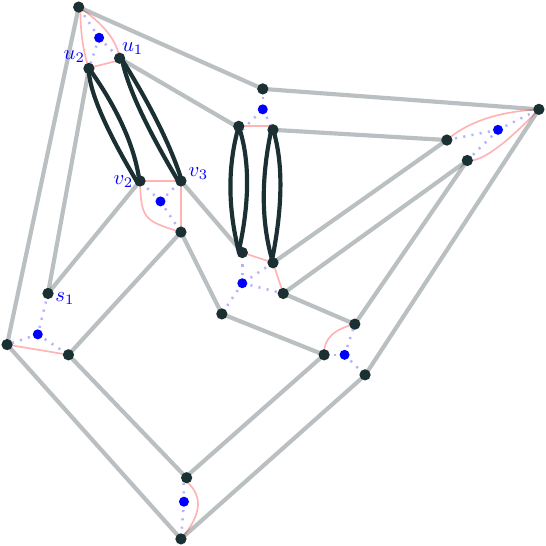}
\end{subfigure}
\hspace{0.1 cm}
\begin{subfigure}[t]{0.31\textwidth}
  \centering
   \includegraphics[width=\linewidth]{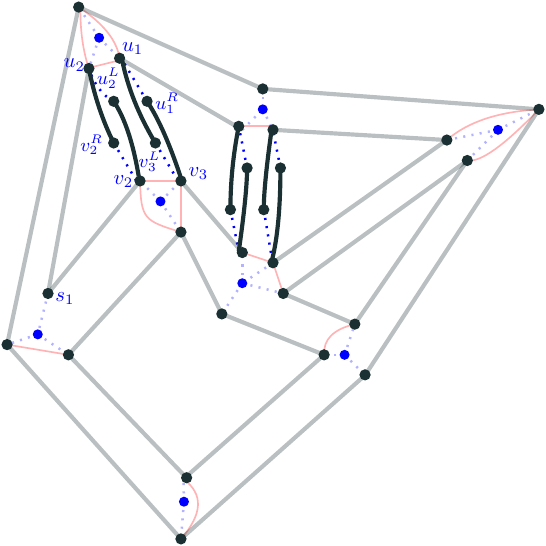}
  \end{subfigure}
  \hspace{0.1 cm}
  \begin{subfigure}[t]{0.31\textwidth}
  \centering
   \includegraphics[width=\linewidth]{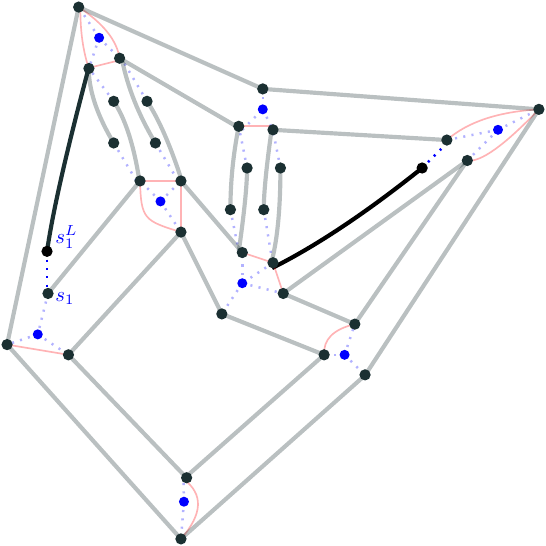}
\end{subfigure}
    \caption{Left: the graph $\hat{G}_0^2$ resulting from $\hat{G}_0^1$ after duplicating $\ER$ edges that map to incised edges in step (2). Middle: the graph $\hat{G}_0^3$ resulting from $\hat{G}_0^2$ after subdividing these parallel edges. In both sides, edges that existed in $\hat{G}_0^{i-1}$ are faded, while new edges in $\hat{G}_0^{i}$ (for $i=2,3$) are not. Right: the graph $\hat{G}_0^4=\hat{G}_\ell$ resulting from step (4) after subdividing  one $\ER$ copy of each endpoint-case edge. 
     Edges that existed in $\hat{G}_0^{3}$ are faded, while new edges are not.
\label{fig: hatG_trans_ER_ES} }
\end{figure}

\medskip
\noindent  
\underline{3. \em Subdividing $\ER$ edges into $\ER$ and $\ES$ edges.}
    We proceed with the changes on $\ER$ edges of $\hat{G}_0^2$ so they match the set of $\ER$ edges of $\hat{G}_\ell$. First, recall that for an incised edge $e$, we have the following four $\ER$ edges in $\hat{G}_\ell$: $(v_{i-1}^R,u_j)$, $(v_{i-1},u_j^L)$, $(v_i^L,u_{j-1})$, $(v_i, u_{j-1}^R)$. 
    We therefore need to add the vertices $v_{i-1}^R,v_i^L,u_{j-1}^R,u_j^L$ and connect them properly with the existing  vertices $v_{i-1},v_i,u_{j-1},u_j$ in $\hat{G}_0^2$. This also includes $\ES$ edges. The missing $\ES$ edges in this case connect a vertex $a\in \{v_{i-1},v_{i},u_{j},u_{j-1}\}$ with its copies that we add from $a\times\{L,R\}$ (e.g. the $\ES$ edge $(v_{i-1}, v_{i-1}^R)$ is missing). See \cref{fig: hatG_trans_ER_ES} (middle) for an illustration. 
    
    These missing vertices, and missing $\ER$ and $\ES$ edges are handled with the following change to $\hat{G}_0^2$:
    Subdivide each one of the parallel edges (from the previous step)  $(v_{i-1},u_j)$, and subdivide each one of the parallel edges  $(v_i,u_{j-1})$. 
    This results in a graph $\hat{G}_0^3$. 
   When the first edge  $(v_{i-1},u_j)$ is subdivided into two edges $(v_{i-1},w),(w,u_j)$, we set the vertex $w$ to be $v_{i-1}^R$. Hence, in $\hat{G}_0^3$, we have the edges $(v_{i-1},v_{i-1}^R),(v_{i-1}^R,u_j)$. The edge $(v_{i-1},v_{i-1}^R)$ is an $\ES$ edge and $(v_{i-1}^R,u_j)$ is an $\ER$ edge. 
    In a similar manner, when subdividing the second edge $(v_{i-1},u_j)$, we  set $w$ to be $u_j^L$ (resulting with the $\ER$ edge $(v_{i-1},u_j^L)$ and the $\ES$ edge $(u_j^L,u_j)$).
    We apply the same logic for the other two subdivided parallel edges $(u_{j-1},v_i)$, resulting with the edges $(u_{j-1},v_i^L)\in \ER$, $(v_i^L,v_i)\in \ES$ and $(v_i, u_{j-1}^R)\in \ER$, $(u_{j-1}^R,u_{j-1})\in \ES$.
    For an incised edge $e\in G$, the vertex and edge copies related to it in $\hat{G}_0^3$ now match the definition (\cref{par: def_hat_G}) of $\hat{G}_\ell$.

\medskip
\noindent  
 \underline{4. \em Endpoint-case edge $e$.}
    An endpoint-case edge $e$ has two $\ER$ copies in $\hat{G}_0$ and in $\hat{G}_0^3$:
     One  $\ER$ copy participates in the cycle $\hat{p}(p_k)$ corresponding to $p_k\in\{p_0,p_{|P|-1}\}$ is an endpoint of the path $P$, and the other $\ER$  copy in the cycle corresponding to $e$'s other endpoint.

     More concretely, let $v$ be $s$ if $p_k=p_0$ or $t$ if $p_k=p_{|P|-1}$. Then, in $\hat{G}_0^3$ we subdivide the copy $(u_{j-1},v_i)$ of $e$ into the two edges: $(v_i,v_i^L)\in \ES$ and $(u_{j-1},v_i^L)\in \ER$.
     Obviously, the graph remains planar after this.
     Finally, we do the following minor change: remove the $\EC$ edge $(v_i,v_{i-1})$ if exists and add the $\EC$ edge $(u_{j-1},u_j)$.
     This matches the definition of $\hat{G}_\ell$ for the an endpoint-case edge (\cref{par: def_hatG_endpoint_case}) and does not violate planarity (as we established in the first step where we added the edges of type $(u_{j-1},u_j)$).
     For an illustration see \cref{fig: hatG_trans_ER_ES} (right).

The above transformation preserves planarity, takes care of all edges (incised/not-incised/endpoint-case), and  matches the definition of $\hat{G}_\ell$ as defined in \cref{par: def_hat_G}.
Hence, $\hat{G}_0^4=\hat{G}_\ell$ and  so $\hat{G}_\ell$ is planar.  \qedhere
\end{proof}

\subsubsection{Vertex-disjoint Face-parts}
We prove the first item of \cref{thm: hat_G} regarding $\hat{G}_\ell$.
\begin{lemma}
\label{lemma: faceparts_hat_G}
    For each level-$\ell$ node-part $f$, $p(f)$ 
 maps to a unique (vertex-disjoint) connected component $\hat{p}(f)$ of $\hat{G}_\ell[\ER]$. Moreover, a consecutive subset of $f$'s edges maps to a unique subpath of $\hat{p}(f)$.
\end{lemma}

    In $\hat{G}_0$ each face $f$ of $G$ maps to a cycle $\hat{p}(f)$ of $\ER$ edges~\cite{GP17}.
    The cycle $\hat{p}(f)$ corresponds to an Euler tour on $f$. 
    Therefore, $\hat{p}(f)$ is connected. 
     $\hat{G}_\ell$ can be seen as a collection of cycles and paths. The cycles (paths) correspond to
      Euler 
    tours on 
    faces (face-parts). 

\begin{figure}[htb]
\centering
\begin{subfigure}[t]{0.45\textwidth}
  \centering
   \includegraphics[width=.7\linewidth]{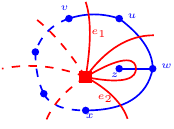}
\end{subfigure}
\hspace{0.5 cm}
\begin{subfigure}[t]{0.45\textwidth}
\centering
\includegraphics[width=.4\linewidth]{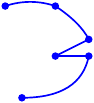}
\end{subfigure}
\caption{Left: A node-part $f$ ({\color{red} red} square) of $h\in G^*$ and its incident edges in solid {\color{red} red} (edges between $e_1$ and $e_2$, clockwise). Dashed {\color{red} red} edges are $h$'s other edges. The corresponding (primal) face-part $p(f)$ is solid {\color{blue} blue}. 
Right: the path $\hat{p}(f)$ in $\hat{G}_\ell$ corresponding to the Euler tour  $v,u,w,z,w,x$ on $p(f)$.
\label{fig: hatp(f')_connected}}
\end{figure} 

    The intuition behind the proof is as follows. By \cref{def: nodepart}, a node-part $f$, is defined by all (dual) edges between two edges $e_{1}, e_{2}$ incident to some node $h$ of $G^*$ ($f$ is a part of $h$). 
            We show that the Euler tour on $p(f)$ corresponds to $\hat{p}(f)$, by starting at the (primal) edge $e_1$ and traversing the face-part $p(f)$. 
If $p(f)$ is a simple path or a cycle then each of its edges is traversed once. Otherwise, we may traverse the same edge twice. This happens if the edge is a self loop in the dual (see e.g., the edge $(w,z)$ in \cref{fig: hatp(f')_connected}).    
    Using duality, we show that for each two consecutive edges in the local embedding of $f$, there is a shared primal vertex to the edges' counterparts in $G$. From here it follows that $p(f)$ is connected. Then, by relating the aforementioned vertices to their copies in $\hat{G}_\ell$ and relating the $f$ edges to their copies in $\hat{G}_\ell$, we get that $\hat{p}(f)$ is connected.
    To prove that for $g\neq f$, $\hat{p}(f)$ is vertex-disjoint from $\hat{p}(g)$, we show that the vertex and edge copies that form each of them are disjoint.

\begin{proof}[Proof of \cref{lemma: faceparts_hat_G}.]
First we eliminate the case of $p(f)$ being an entire face of $G$. 

    \begin{claim}
    \label{claim_0_hatp(f)_connected}
    If $p(f)$ is an entire face of $G$, then  $\hat{p}(f)$ is a cycle and a maximal connected component in $\hat{G}_\ell[\ER]$.
    \end{claim}
        The  claim follows from~\cite{planardistributedmaxflow24, GP17}, since in that case our definition of $\hat{G}_\ell$ matches the definition of $\hat{G}_0$ from~\cite{planardistributedmaxflow24, GP17}. This is because our definition differs from $\hat{G}_0$ only for edges that are (incident to) nodes on an incision path, and if $f$ is not such a node, then $p(f)$ is an entire face of $G$.

        Thus, from now on, it suffices to discuss face-parts $p(f)$ (not an entire face of $G$).
    In addition, we split the proof into two parts, one that discusses connectivity, and the other discusses disjointness of $\hat{p}(f)$.
    
    \medskip
    \noindent
    {\bf Connectivity.}
    In here, we formally describe the tour on $p(f)$ that corresponds to $\hat{p}(f)$  and relate it to $\hat{G}_\ell$ by pointing out the copies of $p(f)$ vertices and edges that are vertices of $\hat{p}(f)$ in $\hat{G}_\ell$ proving:
    \begin{lemma}\label{lem: hatpf_conn}
    For each level-$\ell$ node-part $f$, $p(f)$ maps to a connected component $\hat{p}(f)$ of $\hat{G}_\ell[\ER]$. In particular, a consecutive subset of $f$'s edges maps to a unique subpath of $\hat{p}(f)$.
    \end{lemma}

To prove the lemma, we first provide a claim and a definition that will be important to its proof.

    \begin{claim}
    \label{claim: shared_endpoint_dual_edges}
    Let $f$ be a dual node and $e,e'$ be two incident consecutive edges in the local ordering of $f$. Then, $e,e'$ have a shared endpoint $u$ in $p(f)$ such that $e'$ precedes $e$ in $u$'s local ordering of its incident edges.
    \end{claim}
    
This claim follows from basic arguments on duality~\cite{KleinM_book} and can be seen as special case of  Claim~\ref{claim:pf}.
We next define the first and last vertices of $p(f)$. These are the endpoints of $p(f)$ when $p(f)$ is a path. But $p(f)$ can be a tree (e.g. \cref{fig: hatp(f')_connected}), so we have a more precise definition:

    \begin{definition}[first/last vertex of $p(f)$]
    \label{def: first_last_vertex_pf}
    Let $f$ be a node-part of $h\in G^*$, defined by edges $e_1,e_2$. Then,  
    the shared primal endpoint of $e_1$ and its predecessor edge in $h$ is the {\em first} vertex of $p(f)$, and the shared endpoint of $e_2$ and its successor edge in $h$ is the {\em last} vertex of $p(f)$.
    \end{definition}

We next describe the tour on $p(f)$ that corresponds to $\hat{p}(f)$.
Which combined with Claim~\ref{claim_0_hatp(f)_connected} constitutes the proof of  \cref{lem: hatpf_conn}.

\begin{proof}[Proof of \cref{lem: hatpf_conn}.]
When $f$ is a not split into parts, then $p(f)$ is an entire face of $G$ and the claim holds by Claim~\ref{claim_0_hatp(f)_connected}. 
Now, we prove the lemma for face-parts $p(f)$, showing: Traversing all edges of $p(f)$ one at a time can be mapped to $\hat{p}(f)$, s.t., each traversal on an edge of $p(f)$ corresponds to a copy of the same edge in $\hat{p}(f)$, and two copies of consecutive edges share an endpoint (are connected). See \cref{fig: hatpf_conn}.

\medskip
\noindent
\underline{\em First edge (the edge $e_1$).}
Consider the first edge $e_1$ that is incident to $f$. In $G$, $e_1$ has two primal endpoints $u,v$. 
Assume that $v$ is the first vertex of $p(f)$.
Since $u$ is not the first vertex of $p(f)$, and by Claim~\ref{claim: shared_endpoint_dual_edges}, we get that $u$ is an endpoint in $G$ of the successor edge of $e_1$ in the local ordering of the node $f$. Denote that edge in $G$ by $(u,w)$.
By the same claim, $(u,w)$ is the predecessor of $e_1$ in $u$'s local embedding.
Thus, by definition of $\hat{G}_\ell$ (\cref{par: def_hat_G}) $u$ has one vertex copy that is incident to (a copy of) $e_1$ and to $(u,w)$ in $\hat{G}_\ell$.
This establishes the first two edges in the tour on $p(f)$ which correspond to the first two edges in $\hat{p}(f)$:  the edges $(v,u)$ and $(u,w)$ in $p(f)$.
Next, we specify the corresponding edges of $\hat{p}(f)$. Let $e_1$ be the $i$'th and $j$'th edge of $v$ and $u$ (resp.). Then, in $\hat{G}_\ell$, the copy $u_j$ of $u$ is incident to a copy of $e_1$ whose other endpoint is the copy $v_i^L$ of $v$. I.e., the edge $(u_j,v_i^L)\in \ER$ in $\hat{p}(f)$ corresponds to traversing $e_1$ in $p(f)$ from $v$ to $u$. 
In addition, if $w$ is not the last vertex in $p(f)$, then by definition of $\hat{G}_\ell$ there is a copy $w_k$ of $w$ that is incident to $u_j$ with an $\ER$ copy of the edge $(w,u)$. I.e., $(u_j,w_k)\in \ER$ is the copy for traversing $(u,w)\in p(f)$ in $\hat{p}(f)$.

\begin{figure}[htb]
\centering
    \begin{subfigure}[t]{0.3\textwidth}
        \centering
        \includegraphics[width=\linewidth]{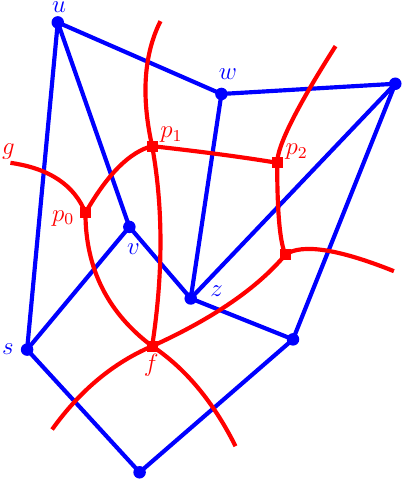}
        \caption{$G$ ({\color{blue} blue}), its dual $G^*$ ({\color{red}red}), and the path $P= (p_0,p_1,p_2)$. \label{fig: hatf_conn_primal}}
    \end{subfigure}
    \hspace{0.15cm}
    \begin{subfigure}[t]{0.3\textwidth}
    \centering
    \includegraphics[width=0.9\linewidth]{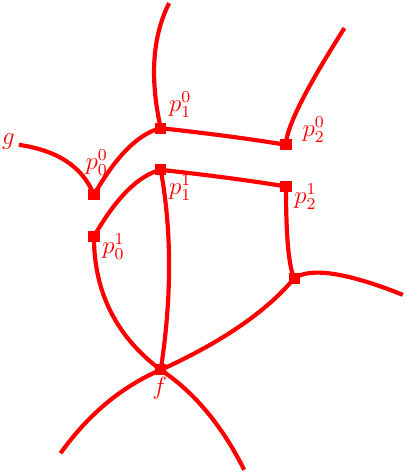}
        \caption{The dual graph $G^*$ after the incision along $P$.}
    \end{subfigure}
    \hspace{0.15cm}
        \begin{subfigure}[t]{0.3\textwidth}
        \centering
        \includegraphics[width=1.2\linewidth]{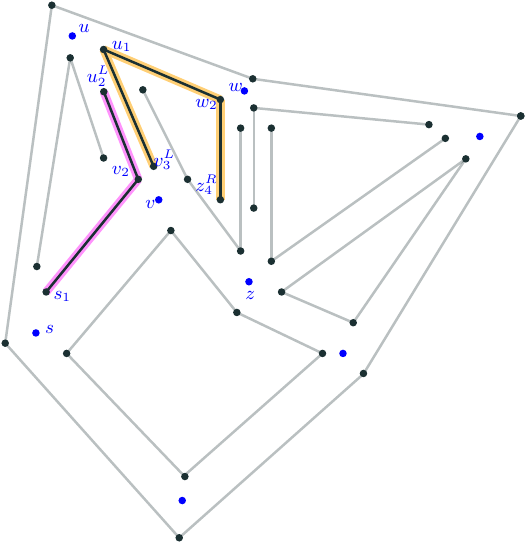}
        \caption{The graph $\hat{G}_1[\ER]$ corresponding to  $G^*$ after incising $P$. \label{fig: hatf_conn_hatG}}
       
        \end{subfigure}
    \caption{The node-part $p_1^0$ in image (b) is defined by edges $e_1=(p_0,p_1), e_2=(p_1,p_2)$ that are dual to the primal edges $(v,u)$ and $(w,z)$, respectively in image (a). The face-part $p(p_1^0)$ in $G$ is the path $(v,u,w,z)$ whose first and last vertices are $v$ and $z$. The corresponding path $\hat{p}(p_1^0)$ is the path $(v_3^L,u_1,w_2,z_4^R)$ in $\hat{G}_1$ (the yellow path in image (c)). 
    The node-part $p_0^1$ in image (b) is defined by the edges $e_1=(p_1,p_0),e_2=(p_0,g)$ that are dual to the primal edges $(u,v)$ and $(s,u)$ respectively in image (a).
     $p^1_0$ is a node part of the endpoint node $p_0\in P$, and the edge (dual to) $(s,u)$ is an endpoint edge. This is the reason that $(p_0,g)$ is not incident to $p_0^1$, and instead $p_0^1$'s last incident edge is $e_2'=(p_0,f)$ that is dual to $(v,s)$.  
    The corresponding face-part $p(p_0^1)$ is the path $(u,v,s)$ whose first and last vertices are $u$ and $s$. In $\hat{G}_1$ the path $\hat{p}(p_1^0)$ is $(u_2^L,v_2,s_1)$ (the pink path in image (c)).
      \label{fig: hatpf_conn} }
\end{figure}

\medskip
\noindent
\underline{\em Edges (strictly) in between $e_1$ and $e_2$.}
We continue in a similar manner as above until we reach the last vertex of $p(f)$, each time traversing the next edge on $p(f)$, where next is defined by the local embedding of $f$. More formally:
\begin{enumerate}[label=(\alph*)]
    \item \label{hatpf_connected_internal_edge_step1}
    The current edge to be traversed is $(w,z)$ such that neither of $w,z$ is the last vertex of $p(f)$. 
    In addition, in the previous step we traversed some edge $e:=(a,w)$ in $p(f)$ and arrived to $w$ (w.l.o.g.) where $e$ is the predecessor of $(w,z)$ in $f$'s local embedding.

\item 
Assume the edge $(w,z)$ is the $k$'th edge of $w$ and the $t$'th edge of $z$.
Since $(w,z)$ is the successor of $(a,w)$ in the embedding of $f$, then vertex $w_k$ is incident to a copy  $z_{t-1}$ of $z$ with an $\ER$ edge (definition of $\hat{G}_\ell$, \cref{par: def_hat_G}). This edge corresponds to the edge being traversed.
Moreover, we arrived to $w$ by traversing $(a,w)$. This is the $k+1$'th edge of $w$ (by Claim~\ref{claim: shared_endpoint_dual_edges} because it precedes $(z,w)$ in $f$'s embedding) and $x$'th edge of $a$. That is, $w_k$ is incident to the edge $(a_x,w_k)\in \hat{p}(f)$. In other words, the two consecutive edges in $p(f)$ corresponding to $(a,w)$ and $(w,z)$ correspond to connected edges in $\hat{p}(f)$ (edges that share an endpoint, $w_k$).

\item
Let the next edge in $f$'s embedding be $(z,y)\neq e_2$ (neither of $z,y$ is the last vertex of $p(f)$). The tour proceeds in a similar manner starting from \ref{hatpf_connected_internal_edge_step1} 
after setting $a:=w$, $w:=z$ and $z:=y$.   
\end{enumerate}

\medskip
\noindent
\underline{\em Last edge (the edge $e_2$).}
Eventually, the edge $e_2$ is traversed. 
The case of $e_2$ is symmetric to $e_1$. I.e., $e_2$ has two endpoints $u',v'$, where $v'$ is the last vertex of $p(f')$. By definition of $\hat{G}_\ell$ (\cref{par: def_hat_G}) $v_{i'}'^R$ is the last vertex of the path $\hat{p}(f)$. In addition,  $v_{i'}'^R$  is connected with an $\ER$ edge to $u'_{j'}$, and $u'_{j'}$ is connected to $w'_{k'}$ with an $\ER$ edge, where the dual of $(u',w')$  proceeds $e_2$ in the embedding of $f$.

\medskip
\noindent
\underline{\em The endpoint-case edge.}
If $e_2$ is an endpoint-case edge, then 
let $e'_2$ be the predecessor of $e_2$ in the local embedding of $f$. 
Let the edge of $G$ corresponding to $e'_2$ be $(u',v')$. 
Then, w.l.o.g. $v'$ is the last vertex of $p(f)$, and $u'$ is incident to a vertex $w'$ in $p(f)$ s.t. the edge $(u',w')$ is dual to the edge that precedes $e'_2$ in the embedding of $f$.
In this case, the last endpoint of the path $\hat{p}(f)$ is simply $v'_{i'}$, which is connected to $u'_{j'}$ with an $\ER$ edge, and $u'_{j'}$ in turn is connected to $w'_{k'}$ with an $\ER$ edge. See \cref{fig: hatpf_conn}.
\end{proof}

\medskip
\noindent
{\bf Uniqueness.}
Here we show the following lemma, which combined with \cref{lem: hatpf_conn} (on connectivity of $\hat{p}(f)$) completes the proof of \cref{lemma: faceparts_hat_G} on connectivity and disjointness of $\hat{p}(f)$s.
\begin{lemma}
\label{lem: hatpf_disjoint}
$\hat{p}(f)$ is vertex-disjoint from $\hat{p}(g)$ for any two level-$\ell$ node-parts $g\neq f$.
\end{lemma}

For that, assume the following claim holds (we provide a proof after proving the above lemma).

\begin{claim}\label{claim: pair_edges_unique_facepart}
Each pair of (ordered) consecutive edges incident to $v\in G$ is unique and belongs only to one face or face-part of an annulus of level $\ell$.
\end{claim}

\begin{proof}[Proof of \cref{lem: hatpf_disjoint}.]
Note that internal vertices (non-endpoints) of $\hat{p}(f)$ are of the form $v_i\in D(v)$ for a vertex $v\in p(f)$. Moreover, by definition of $\hat{G}_\ell$ (\cref{par: def_hat_G}), $v_i$ is incident to two $\ER$ edges only, which belong to $\hat{p}(f)$. In fact, $v_i$ is defined to only be incident to those edges as they are copies of the $i$'th pair of consecutive edges incident to $v$ (for each such pair there is a copy of $v$ in $D(v)$ by definition of $\hat{G}_\ell$). 
Thus, by definition of  $\hat{G}_\ell$,  $v_i$ cannot be shared by $\hat{p}(f)$ and some $\hat{p}(g)$ for $g\neq f$ with an $\ER$ edge (this would contradict the above Claim~\ref{claim: pair_edges_unique_facepart}).
If $p(f)$ has no incised edge ($p(f)$ is an entire face of $G$), then $\hat{p}(f)$ is a cycle formed of $D(v)$-type vertices. Then, the aforementioned reasoning is enough to derive the claim on uniqueness in this case. 

We proceed to the case of endpoints of $\hat{p}(f)$ when it is a path ($p(f)$ contains incised edges and is a face-part of $G$).
It remains to show that no endpoint of $\hat{p}(f)$ is shared by any  $\hat{p}(g)$ where $g\neq f$.
For the special case where $v_i\in D(v)$ is an endpoint of $\hat{p}(f)$, when $f$ is a node-part that is an endpoint of $P$ (see \cref{def: incision_P}, \cref{remark: incision_P}), we get that the reasoning given earlier still holds.
This is because it can be thought of as if $v_i$ had two incident edge copies in $\hat{G}_\ell$ and one of them (the one corresponding to the endpoint-case edge) got deleted (see \cref{par: def_hatG_endpoint_case} for the definition of $\hat{G}_\ell$ in that case). 

In the general case where endpoints of $\hat{p}(f)$ are of the form $v_i^L$ or $v_i^R$, then by definition of $\hat{G}_\ell$ (\cref{par: def_hat_G}), such a vertex has only one $\ER$ edge incident to it.
The primal corresponding edge $e=(u,v)$ to that edge has four copies in $\hat{G}_\ell$, each of them has its unique endpoints. 
In particular, each such $\ER$ edge copy is of the form $(v_i^L,u_j)$. That is a vertex of $D(u)\times\{L,R\}$ and a vertex of $D(v)$ (or symmetrically, a vertex of $D(u)\times\{L,R\}$ and a vertex of $D(v)$).  
Thus, the first and last two vertices of $\hat{p}(f)$  cannot be shared to any $\hat{p}(g)$ where $g\neq f$, because each vertex of $D(u)$ (as proved earlier) belongs only to one $\hat{p}(f)$, and its incident $D(v)\times\{L,R\}$ copy is not incident to any other vertex with an $\ER$ edge, thus, it can only belong to the same $\hat{p}(f)$.
\end{proof}

\begin{proof}[Proof of Claim~\ref{claim: pair_edges_unique_facepart}.]
An (ordered) pair of of consecutive edges incident to $v$ is contained in one face $h$ of $G$, and by Claim~\ref{claim: shared_endpoint_dual_edges} corresponds to a pair of consecutive edges incident to the node $h$ in $G^*$. Let $e,e'$ be such a pair of edges, in that order in $h$'s local ordering in $G^*$.
The pair $e,e'$ of edges can participate together only in one node-part of $h$ (or only the entire node $h$ itself if it is not split to parts), because a node-part is a consecutive set of edges of $f$ (\cref{def: nodepart}), and by definition two node parts can share only their first and last edges with another node-part of the same node. I.e., they cannot share a pair of (ordered) edges. See Claim~\ref{claim: nodes_are_nopdeparts}.
Note, if $e,e'$ are the only edges incident to $v$, then they can also participate in a face-part of the other face $g$ of $G$ that contains them, but in that case their clockwise ordering in (the node) $g\in G^*$ is $e',e$, which is a different (ordered) pair of edges from $e,e'$. \qedhere 
\end{proof}

Note, \cref{lem: hatpf_disjoint} together with \cref{lem: hatpf_conn} concludes the proof of \cref{lemma: faceparts_hat_G} that $\hat{p}(f),\hat{p}(g)$ are connected and vertex disjoint for nodes $f\neq g$ of level $\ell$. \qedhere
\end{proof}

\subsubsection{Correspondence to Annuli}
We prove the second and third items of \cref{thm: hat_G} regarding $\hat{G}_\ell$.

\begin{lemma}
\label{lemma: dual_edges_hat_G}
There is an injection from edges in level $\ell$ of $\mathcal{T}$ to the set of $\EC$ edges in $\hat{G}_\ell$.
For each edge $(f,g)$ of level $\ell$, there is a unique edge $\hat{e}\in \EC$ in $\hat{G}_\ell$ that connects 
   $\hat{p}(f)$ and $\hat{p}(g)$, and there is no $\EC$ edge that connects $\hat{p}(f)$ and $\hat{p}(g)$ if $(f,g)$ is not an edge of level $\ell$.
\end{lemma}

\begin{proof}
    We have established in \cref{lemma: faceparts_hat_G} that each face-part $f$ of $G$ is represented by a path $\hat{p}(f)$ in $\hat{G}_\ell$, such that each vertex $v$ that is contained in $p(f)$ has a copy in $\hat{G}_\ell$ in $\hat{p}(f)$.

    \medskip
    \noindent
    {\bf A level-$\ell$ edge has at least one $\EC$ edge.}
    Let $e^*$ be a level-$\ell$ edge that connects node-parts $f,g$, we show that there is an $\EC$ edge that maps to it by examining two cases. Note, the following covers are all cases, because an edge is incised at most once, by \cref{lem: incise_edge_once}.

    \medskip
    \noindent
    \underline{\em $e^*$ is not incised.} In this case, $e$ has only one dual edge in level $\ell$, connecting node-parts $f,g$. Let $e=(u,v)$ in $G$ be the dual of $e^*$ and the $i$'th edge in $v$'s local embedding, such that the face of $G$ that contains $p(g)$ (resp. $p(f)$) is left (resp. right) to $e$. Hence, the face-part $p(g)$ (resp. $p(f)$) corresponds to the $i-1$'th ($i$'th) pair of consecutive edges incident to $v$. 
    By definition of $\hat{G}_\ell$ (\cref{par: def_hat_G}), since $e^*$ is not incised, then $e$ has exactly two $\ER$ copies in $\hat{G}_\ell$. One is incident to $v_{i-1}$ and the other to $ v_{i}$. In addition, $v_{i-1}$ and $v_i$ correspond to the $i-1$'th and $i$'th pair of consecutive edges incident to $v$ respectively. I.e., $v_{i-1}$ and $v_i$ correspond to the face-parts $p(g)$ and $p(f)$ respectively,
     and $v_{i-1}, v_{i}$ belong to $\hat{p}(f), \hat{p}(g)$ respectively in $\hat{G}_\ell$. Moreover, by definition of $\hat{G}_\ell$, w.l.o.g. there is an edge $(v_{i-1},v_{i})\in \EC$ 
     that represents the dual edge $(f,g)$.
     Otherwise, if $(v_{i-1},v_{i})$ does not exist, then the edge $(u_{j-1},u_{j})$ exists where $e$ is the $j$'th edge in $u$'s local embedding (by definition of $\hat{G}_\ell$ one of those two edges exist when $e^*$ is not incised). In this case, the proof is symmetric (just replace $u$ with $v$, $i$ with $j$, and left with right).
     
    If $e^*$ is an endpoint-case edge of $p_0$ (\cref{def: incision_P}, \cref{remark: incision_P}) then $s$ is an endpoint of $e$ and the same reasoning of a non-incised edge applies, with the difference that the edge $(u_{j-1},u_{j})$ always exists and the edge $(v_{i-1},v_{i})$ never exists iff $v=s$. To verify, see the definition of $\hat{G}_\ell$ in this case (\cref{par: def_hatG_endpoint_case}).
    The other case of $e^*$ being the endpoint-edge of $p_{|P|-1}$ is symmetric by replacing $s$ with $t$.
    
    \medskip
    \noindent
    \underline{\em $e^*$ is incised.} In this case, $e$ has two dual edges in level $\ell$. 
    By Claim~\ref{claim: edge_in_four_p(f)}, in level $\ell$  there are two pairs of node-parts that are connected by a dual of $e$. We consider one such pair $f,g\in A$ of them (the proof for the other pair is identical). Assume w.l.o.g. that $e$ is the last edge in the sequence of  edges that defines the node-part $f$, and $e$ is the first edge in the sequence of edges that defines the node-part $g$. I.e., in $G$ we say that $e$ is the last edge on $p(f)$ and is the first edge on $p(g)$.
    In addition, assume w.l.o.g that $v$ is the last vertex on $p(f)$ and the first vertex of $p(g)$ (intuitively, an endpoint of $e$ that is also an endpoint of $p(f)$ and $p(g)$, formally, as in \cref{def: first_last_vertex_pf}).
    Then, $e$'s other endpoint $u$ is also a vertex in both $p(f)$ and $ p(g)$. Assuming the edge $e$ is the $j$'th edge of $u$, then $u_{j-1},u_{j}$ are copies of $u$ on $\hat{p}(f)$ and $\hat{p}(g)$ respectively. Assume w.lo.g. that the face of $G$ containing $p(g)$ (respectively $p(f)$) is the face that corresponds to the $j-1$'th  (respectively $j$'th) pair of consecutive incident edges to $u$.
    Then, from the definition of $\hat{G}_\ell$ (\cref{par: def_hat_G}) the edge $(u_{j-1},u_{j})\in \EC$ 
    represents the dual edge $(f,g)\in A$.
    Note, in this case the edge $(v_{i-1},v_{i})$ does not map to $e$, since the copies $v_{i-1},v_{i}$ are not in $\hat{p}(f),\hat{p}(g)$ but in $\hat{p}(f'),\hat{p}(g')$ where $f',g'$ are node-parts that overlap with $f,g$  with the edge $e^*$. I.e., the edge $(v_{i-1},v_{i})$ corresponds to the edge $(g',f')$ which is the second duplicate of $e^*$ in    level-$\ell$.
    In fact, the copies of $v$ that participate in $\hat{p}(f),\hat{p}(g)$ are $v_{i-1}^R,v_{i}^L$ (resp.) and by definition of $\hat{G}_\ell$ there is no edge of any type that connects them.

    \medskip
    \noindent
    {\bf A level-$\ell$ edge has at most one $\EC$ edge.}
    Note that by the definition of $\hat{G}_\ell$, for an edge $e\in G$ whose dual in $G^*$ is not incised, there is only one $\EC$ edge in $\hat{G}_\ell$.
    Otherwise, by definition of $\hat{G}_\ell$ there are exactly two $\EC$ edges in $\hat{G}_\ell$ added for $e$, and $e^*\in G^*$ is incised (exactly once, by \cref{lem: incise_edge_once}); Thus, $e$ has at most two duals  in level $\ell$.
    One of these edges might be incident to a node-part $h$ that is not contained in any annulus \cref{cor: residual_nodeparts}. 
    Thus, for each edge of level $\ell$ there is exactly one $\EC$ edge in $\hat{G}_\ell$, and each $\EC$ edge represents exactly one edge (whether it is in an annulus or not).
    I.e., there is an injection from edges of level $\ell$ of $\mathcal{T}$ to the set of $\EC$ edges. 
    \qedhere

\end{proof}

\begin{lemma}
\label{lemma: annuli_components_hatG}
Each annulus $A$ of level $\ell$ corresponds to a unique connected component $\hat{A}$ of $\hat{G}_\ell[\ER\cup \EC]$, s.t. node-part $f\in A$ iff $\hat{p}(f)$ is a subgraph of $\hat{A}$. Moreover $e=(f,g)\in A$ iff $\hat{e}\in \hat{A}$ connects $\hat{p}(f),\hat{p}(g)$.
\end{lemma}

\begin{proof}
By \cref{lemma: dual_edges_hat_G}, any two node-parts $f,g$ are adjacent in an annulus of level $\ell$ iff there is an $\EC$ edge connecting their corresponding paths $\hat{p}(f)$, $\hat{p}(g)$ (which exist and are disjoint for $f\neq g$ by \cref{lemma: faceparts_hat_G}). 

Consider any two node-parts $f,g$ in an annulus $A$. Since an annulus is a connected graph, then $f,g$ are connected by a path $A(f,g)=(f,h_1,h_2,\ldots,h_k,g)$ in $A$. In $\hat{G}_\ell[\ER\cup \EC]$, each of $\hat{p}(f)$, $\hat{p}(h_i)$, and $\hat{p}(g)$ is a path by \cref{lemma: faceparts_hat_G}. In addition, by \cref{lemma: dual_edges_hat_G} there exists an $\EC$ edge connecting the corresponding paths of each two consecutive nodes of $A(f,g)$.
Thus, $\hat{p}(f)$ and $\hat{p}(g)$ belong to the same connected component in $\hat{G}_\ell[\ER\cup \EC]$ as they are connected via the subgraph which is the union of the following: the path $\hat{p}(f)$, the $\EC$ edge connecting $\hat{p}(f)$ and $\hat{p}(h_1)$, the path $\hat{p}(h_1)$, the $\EC$ edge connecting $\hat{p}(h_1)$ and $\hat{p}(h_2)$, the path $\hat{p}(h_2)$ and so on until the $\EC$ edge connecting $\hat{p}(h_k)$ and $\hat{p}(g)$, and lastly the path $\hat{p}(g)$.
Denote the subgraph we just defined by $\hat{A}(f,g)$. Then, the union of all $\hat{A}(f,g)$ for node-parts $f,g\in A$ is defined to be $\hat{A}$.

Now we prove that $\hat{A}$ is in fact a (maximal) connected component of $\hat{G}_\ell[\ER\cup \EC]$, and in particular that it does not contain any (vertex in) $\hat{p}(h)$ for $h\notin A$.
Consider any node-part $f\in A$. $f$ is not connected to any node-part in an annulus $B\neq A$, since by definition $A$ and $B$ are different annuli, which by Claim~\ref{claim: disjoint_annuli} are disjoint.
Thus, by \cref{lemma: dual_edges_hat_G} $f$ has no $\EC$ edge that connects $\hat{p}(f)$ to a $\hat{p}(g)$ for $g\in B$. This with the fact that $\hat{p}(f)$ and $\hat{p}(g)$ are vertex disjoint means that $\hat{p}(f)$ and $\hat{p}(g)$ are disconnected in $\hat{G}_\ell[\ER\cup \EC]$. Thus, no vertex of $\hat{p}(g)$ for $g \notin A$ exists in $\hat{A}$.
\end{proof}

\subsubsection{Simulation of $\hat{G}_\ell$}
We prove the fourth and fifth items of \cref{thm: hat_G} regarding $\hat{G}_\ell$, concluding the proof of the theorem.

\begin{lemma}
\label{lemma: construct_hat_G}
    For any level $\ell$ of $\mathcal{T}$, if incisions to the level are distributively stored (as in \cref{def: distributed_storage_incisions}), then 
     $\hat{G}_\ell$ can be constructed in $O(1)$ rounds. After which, each of its vertices has a unique $\tilde{O}(1)$-bit ID, and every vertex in $G$ knows the information of all its copies in $\hat{G}_\ell$ and their adjacent edges.

\end{lemma}
\begin{proof}
    A vertex $v\in G$ needs to know for each incident edge $e=(u,v)$ the relative order of $e$ in $u$'s local embedding, which can be exchanged in one round.
    The information related to incisions on edges incident to $v$ is given by assumption.
    Finally, IDs of vertices in $\hat{G}_\ell$ can be assigned locally: Each vertex $v$ assigns its (at most) $3\deg(v)+1$ copies an ID consisting of $v$'s ID concatenated with a unique number from the range~$[0,3\deg(v)]$.
\end{proof}

\begin{lemma}
\label{lemma: congest_hat_G}
    Any $\CONGEST$ round on $\hat{G}_\ell$ can be simulated within four rounds in $G$.
\end{lemma}
\begin{proof}
    Note that the only edges that require simulation are edges of $\ER$, as other types of edges are edges between copies of the same vertex $v\in G$, meaning that $v$ simulates them locally. The claim follows because any edge in $G$ has at most four $\ER$ copies in $\hat{G}_\ell$. 
\end{proof}

\subsection{Distributed Knowledge}
\label{sec: IDs_to_face_parts}
In this subsection we show how to use \cref{thm: hat_G} (the graph $\hat{G}_\ell$) in order to prove the following lemma. 

\begin{lemma}\label{lem: distributed_knowledge}
Assuming the incisions of level $\ell$  are distributively stored as in \cref{def: distributed_storage_incisions},  the following can be computed in $\tilde{O}(D)$ rounds. Each vertex $v\in G$, for all level-$\ell$ node-parts $f$  where $v\in p(f)$, knows:

 \begin{enumerate}
    \item A unique $\tilde{O}(1)$-bit ID of $f$ and a unique leader vertex $u\in p(f)$.
     
    \item A unique $\tilde{O}(1)$-bit ID of the level-$\ell$ annulus $A$ containing $f$, and the ID of   $A$'s level-$\ell-1$ parent annulus.

    \item 
    All of $v$'s incident edges that are in $p(f)$. 

    \item
    The IDs of all node-parts $h$ from lower levels whose $p(h)$ contains $p(f)$ (see Claim \ref{fact:nodeparts of nodeparts}). 
    
\end{enumerate}
\end{lemma}

The proof is inductive, and shows how, given the information as stated in the lemma, we can extend it to the next level. To do so we rely on properties of $\hat{G}_\ell$ (e.g. to identify face-parts, we run a connectivity algorithm on a subgraph of it). Then, local computations regarding knowledge of vertices to the previous level are performed, deriving the correspondence to the current level. Before that, we first briefly define a couple of distributed primitives that we will use in the proof and in \cref{sec:Minor-aggregation}.
\label{sec: dual_incisions_prel}

\medskip
\noindent
{\bf Low-congestion shortcuts and part-wise aggregation.}
\label{par: shortcuts_aggregations}
The notion of {\em low-congestion shortcuts} was introduced by Ghaffari and Haeupler~\cite{GhaffariH16_shortcuts} first for planar graphs, and then extended for other graph families~\cite{GhaffariH21_shortcuts, HaeuplerIZ21_shortcuts} as a way to perform certain type of computation (known as {\em aggregation}) on multiple disjoint subgraphs of $G$ simultaneously, within a round complexity proportional to $D$ rather than the (possibly) large diameter of the input subgraphs.
Concretely, consider a partition $\{S_1,\ldots,S_N\}$ of $V$ where every $G[S_i]$ is connected:
\begin{definition}
    [Low-congestion shortcuts]
    \label{def: shortcuts}
    An $\alpha$-shortcut of $G$, is a set of subgraphs $\{H_1,\ldots,H_N\}$ of $G$ such that the diameter of each $G[S_i]\cup H_i$ is at most $O(\alpha)$ and each edge of $G$ participates in at most $O(\alpha)$ subgraphs $H_i$.
\end{definition} 
    
\begin{definition}
    [Aggregate operators]
      Given a set of $b$-bit strings $\{x_1,\ldots,x_m\}$ and an aggregation operator $\oplus$ (e.g. AND /OR /SUM etc.), their aggregate $\bigoplus_i x_i$ is defined as the result of repeatedly replacing any two strings $x_i, x_j$ with the ($b$-bit) string  $x_i \oplus x_j$, until a single string remains.
\end{definition} 

\begin{definition}
    [Part-wise aggregation]
    \label{def: PA}
    Assume each $v\in S_i$ has an $\tilde{O}(1)$-bit string $x_v$. The PA problem asks that each vertex of $S_i$ knows the aggregate function $\bigoplus_{v\in S_i} x_v$.
\end{definition}

\noindent
The following lemma captures the relation between the PA problem and shortcuts.

\begin{lemma}
    [Part-wise aggregation via low-congestion shortcuts~\cite{GhaffariH16_shortcuts, GhaffariH21_shortcuts, HaeuplerIZ21_shortcuts}]
    \label{lem: PA_shortcuts}
    If $G$ admits an $\alpha$-shortcut that can be constructed in $r$ rounds, then the PA problem can be solved on $G$ in $\tilde{O}(r+\alpha)$ rounds.
\end{lemma}

\noindent Since a planar graph $G$ of hop-diameter $D$ admits low-congestion shortcuts with $\alpha=\tilde{O}(D)$ that can be constructed in $\tilde{O}(D)$ rounds~\cite{GhaffariH16_shortcuts, GhaffariH21_shortcuts, HaeuplerIZ21_shortcuts}, we get:

\begin{corollary}
\label{cor: planar_PA}
The PA problem can be solved on a planar graph $G$ of hop-diameter $D$ in $\tilde{O}(D)$ rounds.  
\end{corollary}

\noindent
From the properties of $\hat{G}_\ell$ (\cref{lemma: congest_hat_G} and \cref{lemma: diameter_hat_G}), in addition to \cref{cor: planar_PA}, we get the following:
\begin{corollary}
\label{cor: PA_g_hat}
     For any level $\ell$ of $\mathcal{T}$, if the incisions are distributively stored as in \cref{def: distributed_storage_incisions}, then
    the PA problem can be solved on $\hat{G}_\ell$ in $\tilde{O}(D)$ rounds.
\end{corollary}

Now, we are ready to prove \cref{lem: distributed_knowledge}.
\begin{proof}[Proof of \cref{lem: distributed_knowledge}.]
Intuitively, we use \cref{lemma: faceparts_hat_G} and \cref{lemma: annuli_components_hatG} that relate level-$\ell$ face-parts and annuli to $\hat{G}_\ell$, combined with the fact that the PA problem on $\hat{G}_\ell$ can be solved in $\tilde{O}(D)$ rounds (\cref{cor: PA_g_hat}). In particular, by \cref{lemma: faceparts_hat_G}, each face-part $p(f)$ of a level-$\ell$ node-part $f$ corresponds to a connected component $\hat{p}(f)$ in $\hat{G}_\ell[\ER]$. Thus, identifying $p(f)$s boils down to identifying the connected components of $\hat{G}_\ell[\ER]$. This can be done using any connectivity algorithm that terminates within $\tilde{O}(1)$ PAs~\cite{GhaffariH16_shortcuts,GP17}.
Similarly, using \cref{lemma: annuli_components_hatG} applied to $\hat{G}_\ell[\ER\cup \EC]$ we can identify the annuli of level $\ell$.
Finally, learning the correspondences between face-parts of different levels and between face-parts and annuli, is done locally by examining different cases related to the structure of $\hat{G}_\ell$ and~$\hat{G}_{\ell-1}$.

    In the case of $\ell=0$, the only annulus in the level is the entire graph $G^*$. In this case, it was shown in earlier work~\cite{GP17} how to compute the required information using $\hat{G}$.
    For a general level $\ell>0$ we prove by induction, assuming that the lemma holds for level $\ell-1$ and that the incisions of level $\ell$ are distributively stored. 
    First, we construct $\hat{G}_\ell$ in $O(1)$ rounds, where each vertex of $G$ knows all the information of its copies in $\hat{G}_\ell$ and their IDs (\cref{lemma: construct_hat_G}).

    We work in two phases, the first phase is a global phase that identifies the level-$\ell$ annuli, and $p(f)$s of level-$\ell$ node-parts $f$. The second phase is a local phase that computes the following correspondences: (1) face-parts of level $\ell$ and annuli of level $ \ell$, (2) face-parts of level $\ell$ and face-parts of level $\ell-1$ (see Claim~\ref{fact:nodeparts of nodeparts}), and (3) annuli of level $\ell$  and annuli of level $\ell-1$.

    \medskip
    \noindent
    {\bf Global phase.}
     We run a classic connectivity algorithm~\cite{GP17,GhaffariH16_shortcuts} on $\hat{G}_\ell$ that works in $\tilde{O}(1)$ instances of the PA problem. By \cref{cor: PA_g_hat}, this translates to $\tilde{O}(D)$ rounds on $G$.
    The input to the connectivity algorithm is $\hat{G}_\ell[\ER\cup \EC]$, i.e., $\hat{G}_\ell$ after removing the $V_S$ vertices and $\ES$ edges (they are still used for communication). The identified connected components correspond to the annuli of level $\ell$ (\cref{lemma: annuli_components_hatG}).
    We mention that this connectivity algorithm assigns each component it identifies an ID, which is the maximal ID of a vertex in the component. We use these IDs as the IDs of the level-$\ell$ annuli. They are unique since IDs of vertices of $\hat{G}_\ell$ are unique as shown in \cref{lemma: construct_hat_G}.
    Similarly, in $\tilde{O}(D)$ rounds we identify components of $\hat{G}_\ell[\ER]$. I.e.,  $\hat{G}_\ell$ after removing $V_S$ vertices and $\ES\cup \EC$ edges. Connected components of this graph correspond to 
    $p(f)$ for level-$\ell$ node-parts $f$ (\cref{lemma: faceparts_hat_G}).
    Since the ID of a path $\hat{p}(f)$ for a node-part $f$ is an ID of a vertex that lies on it (and this vertex ID is known to all other vertices of $\hat{p}(f)$) then we consider this vertex as the leader of $\hat{p}(f)$.

    \medskip
    \noindent
    {\bf Local phase.}
    In this phase, a vertex $v\in G$ aims to learn the following correspondences:
    \begin{enumerate}
        \item
        Between incident edges of $v$ and face-parts $p(f)$  where $v\in p(f)$ and $f$ is a level-$\ell$ node-part. In addition, between each such $p(f)$ and the specific level-$ \ell$ annulus that contains $f$.

        \item
        Between face-parts $p(f)$ of level $\ell-1$ where $v\in p(f)$ and the face-parts $p(h)$ of level $\ell$ where $v\in p(h)$, and $h$ is a part of $f$ (see Claim~\ref{fact:nodeparts of nodeparts}).
        \item 
        Between annuli $B$ of level $\ell-1$ for which there exists a node-part $f\in B$ where $v\in p(f)$ and child annuli $A$ of $B$ for which there exists a node-part $f'\in A$ s.t. $v\in p(f')$. 
    \end{enumerate}

    \smallskip
    \noindent
    \underline{1. \em Face-parts $p(f)$ and annuli of level $\ell$.} 
   A connected component in $\hat{G}$ of each vertex and edge of $\hat{p}(f)$ was identified in both calls to the connectivity algorithm in phase one. Thus, vertices of $\hat{G}_\ell$ know the ID of their annulus (from the first call) and the ID of $\hat{p}(f)$ (from the second call).
    Note, a vertex $v\in G$ knows the information about all copies of its vertices and edges in $\hat{G}_\ell$ (\cref{lemma: construct_hat_G}). Thus, for each incident edge $e$, $v$ knows to which annulus and which face-parts of level $\ell$ does $e$ belong.

    \smallskip
    \noindent
    \underline{2. \em Face-parts of level $\ell$ and face-parts of level $\ell-1$.}
    Recall, every face-part $p(f)$ in level-$\ell$ annuli is a subgraph of a unique level $\ell-1$ face-part (see Claim~\ref{fact:nodeparts of nodeparts}).  
Consider two node-parts $g,h$ in levels $\ell-1$ and $\ell$ (respectively) of the same node $f\in G^*$. Let $g$ (respectively $h$) be defined by the sequence of edges between $e_1$ and $e_2$ (respectively $e_1'$ and $e_2'$) in $f$'s embedding.
  Then by Claim~\ref{fact:nodeparts of nodeparts}, either the edges of $h$ are contained in the edges of $g$ (i.e., $e_1\leq e_1'\leq e_2'\leq e_2$ in $f$'s embedding), or the edges of $g$ and $h$ are disjoint, except perhaps in their last or first edge (i.e., $e_1\leq e_2 \leq e_1'\leq e_2'$). This happens in case $h$ is a node-part of another node-part $q'\neq g$ of $f$ (see Claim~\ref{claim: nodes_are_nopdeparts}). 
    In the former case, an edge $e$ of $h$ must satisfy one of the following cases: 
    
    \begin{enumerate}[label=(\alph*)]
    \item 
    If $e$ is not incised in any level ($\ell$ or lower), then $e$  participates in exactly two node-parts ($e$'s endpoints) in each level. I.e., in level $\ell-1$ we have $e=(g,\cdot)$ and in level $\ell$ we have $e=(h,\cdot)$, where $h$ is a part of $g$.
    Then, every endpoint $v$ of $e$ learns the correspondence between face-parts $p(g),p(h)$.
    Concretely, if $e$ is the $i$'th edge in $v$'s local embedding, then in $\hat{G}_\ell$ and $\hat{G}_{\ell-1}$, $e$ has two $\ER$ copies. The correspondence between $\hat{p}(g),\hat{p}(h)$ that contain these copies in the two successive levels is immediate:
    One copy of $e$ is incident to the vertex $v_i$ and the other is incident to the vertex $v_{i-1}$ (in each of $\hat{G}_\ell, \hat{G}_{\ell-1}$). Assuming $f$ is the face left to $e$ in $v$'s local embedding (formally, the face that contains the $(i-1)$'th and the $i$'th edge incident to $v$. Then, $v_{i-1}$ is the copy of $v$ that represents the parts $g,h$ of $f$ in $\hat{G}_{\ell-1},\hat{G}_\ell$ (respectively).
    Hence, $v$ knows that $v_i\in \hat{p}(g)$ in $\hat{G}_{\ell-1}$, and that $v_{i-1}\in \hat{p}(h)$. Finally, because $v$ knows the IDs of $\hat{p}(g),\hat{p}(h)$ it locally learns the correspondence between them, and equivalently between $g$ and $h$ and between $p(g)$ and $p(h)$.
   If $e$ is an endpoint-case edge, then the same proof applies except that $v_i$ is replaced by $v_i^L$.

    \item
    If $e$ was already incised in a lower level than $\ell$, we handle it similarly to the above. The only difference is that $e$ has four copies in each of $\hat{G}_\ell$ and $\hat{G}_{\ell-1}$, incident to the copies $v_{i-1}^R,v_{i-1},v_{i}^L,v_{i}$ of $v$. Then, for each of these copies, the correspondence between the components of $\hat{G}_{\ell-1}$ and $\hat{G}_\ell$ that contain it is locally learned by $v$.

   \item If $e$ is incised in level $\ell$ (e.g.  $e=e_1'$ when $e_1<e_1'$). 
  Consider first the case where $e$ is not an endpoint-case edge. Then, $e$ has two $\ER$ copies in $\hat{G}_{\ell-1}$, one incident to $v_{i}$ and the other incident to $v_{i-1}$. In $\hat{G}_{\ell}$, $e$ has four copies, one copy for each of $v_{i-1}^R,v_{i-1},v_{i}^L,v_{i}$. Thus, the paths in $\hat{G}_\ell$ (corresponding to face-parts $p(f)$) that contain the copies $v_{i-1}^R,v_{i-1}$ of $v$, both correspond to the path of $\hat{G}_{\ell-1}$ that contains the copy $v_{i-1}$ of $v$. Similarly for $v_{i}^L,v_{i}$ in $\hat{G}_\ell$ and $v_{i}$ in $\hat{G}_{\ell-1}$.
    
    In the case that $e$ is an endpoint-case edge,  the edge still has two $\ER$ copies in $\hat{G}_{\ell-1}$, however this time they are connected to $v_{i-1},v_{i}^L$. 
    In this case, the paths of $\hat{G}_\ell$ (that represent face-parts $p(f)$) that contain the copies $v_{i-1}^R,v_{i-1}$ of $v$, both correspond to the path of $\hat{G}_{\ell-1}$ that contains the copy $v_{i-1}$ of $v$, similar to the previous case. However, the paths of $\hat{G}_\ell$ that contain $v_{i}^L,v_{i}$ in $\hat{G}_\ell$ correspond to the path that contains $v_{i}^L$ in $\hat{G}_{\ell-1}$. I.e., the same proof as of the previous just replacing $v_i$ in $\hat{G}_{\ell-1}$ with $v_i^L$.
    \end{enumerate}
    
    \smallskip
    \noindent
    \underline{3. \em Annuli of level $\ell$  and annuli of level $\ell-1$.}
    Given the correspondence between face-parts of successive levels $\ell-1\geq0$ and $\ell$, and the correspondence between face-parts of a level $k\in\{\ell-1,\ell\}$ and the level $k$ annulus that contains them. Then, the correspondence between annuli of levels $\ell-1$ and $\ell$ is immediate: Let annulus $A$ be a child of annulus $B$. All edges of $A$ are such that they (or copies of them) are contained in $B$. Let the edge $e=(u,v)$ in $G$ be an edge (whose dual is) in $A$. Then, $v$ knows the ID of a face-part $p(g)$ that contains $e$ where $g\in B$, and the ID of a face-part $p(h)$ where $h\in A$, and $h$ is a part of $g$ (see Claim~\ref{fact:nodeparts of nodeparts}). Hence, $v$ knows the correspondence between $p(h),p(g)$ and the correspondence between each of them and the annulus that contains it. Then, it learns that $B$ is a child of $A$.
\end{proof}

    \noindent
    {\bf Node-parts that are not in annuli.}
    By \cref{cor: residual_nodeparts}, every node $f$ of $G^*$ has some node-parts that participate in level-$\ell$ annuli and other (complementing) node-parts that do not participate in any level-$\ell$ annulus.  
    Note however, the definition of $\hat{G}_\ell$ does not make the distinction between node-parts in annuli and node-parts that are not in any annuli. 
    Thus, if we consider all graphs that result from the incisions that lead to level $\ell$, and not only annuli, everything we proved also applies to them.  
    In particular, such complementing node-parts $h$ are still assigned IDs by \cref{lem: distributed_knowledge}. Moreover, their connected components $p(h)$ are identified by their corresponding $\hat{p}(h)$ in $\hat{G}_\ell$.
    In \cref{sec: implemetnation_details}, where we put everything together, we discuss this in more detail and show how a vertex $v$ learns about such node-parts $h$ where $v\in p(h)$. We show that $v$ learns that $h$ should be discarded from participating in annuli of level $\ell$ (other information on such a node-part $h$ is then treated as garbage).

\subsection{Simulation of Minor-aggregation} \label{sec:Minor-aggregation}
We define the minor-aggregation model as first defined by~\cite{GHSYZ22}. The goal of this model is to eliminate low-level implementation details of $\CONGEST$, allowing a clean interface for recursive edge-contraction based algorithms. The model compiles to a sequence of part-wise aggregation problems, which can then be implemented efficiently in $\CONGEST$ via the low-congestion shortcuts framework (see \cref{par: shortcuts_aggregations} for definitions).

\begin{definition}
[Minor-aggregation model \cite{GHSYZ22}]
\label{def: basic_model}
Given an undirected graph $G=(V,E)$, both vertices and edges are computational entities (i.e. have their own processor and their own private memory). Each vertex has a unique $\tilde{O}(1)$-bit ID. Communication occurs in synchronous rounds. All entities wake up at the same time, each vertex knows its ID, and each edge knows its endpoints' IDs.
Each round of communication consists of the following three steps (in that order):

\begin{enumerate}
    \item 
    {\em Contraction step.}
    Each edge $e$ chooses a value in $\{{0},{1}\}$. Contracting all edges that chose ${1}$ defines a minor $G'=(V',E')$ of $G$. 
    The nodes $V'$ of $G'$ are subsets of $V$ called  {\em super-nodes}.
    
    \item 
    {\em Consensus step.}
    Each node $v\in V$ chooses an $\tilde{O}(1)$-bit value $x_v$. For each super-node $s\in V'$, let $y_s = \bigoplus_{v\in s} x_v$, where $\oplus$ is a pre-defined aggregation operator. All $v\in s$ learn $y_s$.
    
    \item
    {\em Aggregation step.}
    Each edge $e=(a,b)\in E'$ learns $y_a$ and $y_b$, then chooses two $\tilde{O}(1)$-bit values, $z_{e,a}$ and $z_{e,b}$. For each super-node $s\in V'$, let $I(s)$ be the set of incident edges of $s$, then  $z_s=\bigotimes_{e\in I(s)} z_{e,s}$, where $\otimes$ is a pre-defined aggregation operator. All $v\in s$ learn the same $z_s$. 
    
\end{enumerate}
\end{definition}

We mention that~\cite{GZ22} have defined an extended version of the model that supports adding {\em virtual nodes} (nodes that do not exist in the graph and need to be simulated). We discuss this in \cref{sec: virtual_nodes_simulation}.
We need this for the SSSP algorithm of~\cite{RozhonGHZL22_shortestpaths} and to turn it into one that obeys the approximate triangle-inequality. See \cref{sec: SSSP} for more details.
Thus, we focus on simulating the standard model as defined above. It was shown by~\cite{GHSYZ22} that using low-congestion shortcuts, any minor-aggregation algorithm can be simulated in the standard $\CONGEST$ model. It is easy to check that the consensus and aggregation steps of the minor-aggregation model can indeed be phrased as PA problems, and it was shown in~\cite{GHSYZ22} that the contraction step can be broken to a series of $O(\log{n})$ PA problems. Hence:

\begin{lemma}[Minor-aggregation to $\CONGEST$ via PA \cite{GZ22, GHSYZ22}]
\label{lem: MA_via_PA}
If there exists an $N$-round $\CONGEST$ algorithm on $G$ for solving the PA problem on a graph $H$, then and $r$-round minor-aggregation algorithm on $H$ can be simulated in $\tilde{O}(r\cdot N)$ $\CONGEST$ rounds on $G$,  
\end{lemma}

For planar graphs, when $H=G$ we have that $\tilde{O}(\tau\cdot N)=\tilde{O}(\tau\cdot D)$ (from Corollary~\ref{cor: planar_PA}), and we have the same bound when $H=G^*$ (from~\cite{planardistributedmaxflow24}). We aim to show a similar result for simulating the minor-aggregation model on all annuli of the same level simultaneously:

\begin{restatable}{theorem}{theoremMinorAggregationSimulation}
\label{thm: MA_simulation}
         Any $r$-round minor-aggregation algorithm can be simulated simultaneously within $\tilde{O}(rD)$ $\CONGEST$ rounds in $G$ on all annuli of level $\ell$ in $\mathcal{T}$, where each annulus may contain additional $\tilde{O}(1)$ arbitrarily connected virtual nodes. The output of a node-part $f$ in such level-$\ell$ annulus is known to all vertices that lie on the corresponding $p(f)$ in $G$, and the output of an edge $e^*$ is known to the endpoints of $e$ in $G$. Finally, the output of a virtual node in annulus $A$ is known to all $v\in p(f)$ for $f\in A$.
\end{restatable}

\subsubsection{Partwise-aggregation on Annuli}
To prove \cref{thm: MA_simulation}, we first prove that the PA problem can be solved in $\tilde{O}(D)$ rounds on all annuli of the same level simultaneously. Then, we use that to prove \cref{thm: MA_simulation} similarly to~\cite{GHSYZ22,planardistributedmaxflow24}. 

\begin{lemma}
\label{lem: PA_level_T}
   Let $A_1,\ldots,A_k$ be the level-$\ell$  annuli.
   For every $i$, let $\mathcal{P}_i = A_{i,1}, A_{i,2},\ldots$ be a given partition  of $A_i$ into node-disjoint connected subgraphs, such that: (1) level $\ell$ is distributively stored, (2) each vertex $v\in G$, for each node-part $f\in A_i$ s.t. $v\in p(f)$, knows the ID of the part $A_{i,j}$ that contains $f$, (3) for each node-part $f\in A_i$, there is a leader vertex $v\in p(f)$ that knows the $\tilde{O}(1)$-bit input string $x_f$.
    Then, for any predefined aggregate operator $\bigoplus$, in $\tilde{O}(D)$ rounds, each vertex $v\in G$ learns the function $\bigoplus_{g\in A_{i,j}} x_g$ for each part $A_{i,j}$  where there exists $f\in A_{i,j}$ s.t. $v\in p(f)$.
\end{lemma}

\begin{proof}

 We first show that each $A_{i,j}$ corresponds to a unique subgraph of $\hat{G}_\ell$.
    Recall that we simulate the annuli $A_i$ using $\hat{G}_\ell$. In particular, $A_i$ corresponds to a connected component $\hat{A}_i$ of $\hat{G}_\ell[\ER\cup \EC]$ (\cref{lemma: annuli_components_hatG}).
    Moreover, the node-parts of $A_i$ correspond to vertex-disjoint subgraphs (paths) of $\hat{A}_i$ (\cref{lemma: faceparts_hat_G} and \cref{lemma: annuli_components_hatG}). Thus, a partition $\mathcal{P}_i$ of $A_i$ induces a partition $\hat{\mathcal{P}}_i$ of $\hat{A}_i$.
    In particular, for each part $A_{i,j}$ there is a part $\hat{A}_{i,j}$ in $\hat{\mathcal{P}}_i$ that contains all vertices that lie on a path $\hat{p}(f)$ of $\hat{A}_{i,j}$ for a node $f\in A_{i,j}$.
    Note:
    \begin{claim}
        $\hat{A}_{i,j}$ is connected.
    \end{claim}
    The claim follows from \cref{lemma: annuli_components_hatG}, and concretely: (1) each $\hat{p}(f)$ is connected (\cref{lemma: faceparts_hat_G}), (2) there is an $\EC$ edge connecting $\hat{p}(f),\hat{p}(g)$ iff the edge $(f,g)\in A_{i,j}$ (\cref{lemma: dual_edges_hat_G}), and  (3) $A_{i,j}$ is connected by assumption. That is, $\hat{\mathcal{P}}_i$ is a partition of $A_i$ into node-disjoint connected subgraphs.   
 The next claim follows from: (1) all $\hat{p}(f),\hat{p}(g)$ are disjoint for any $f \neq g$ in $A_i$ (\cref{lemma: faceparts_hat_G}), and (2) $\mathcal{P}_i$ is a partition so a node $f\in A_i$ is in exactly one part ${A}_{i,j}\in \mathcal{P}_i$.
    \begin{claim}
        Each two parts $\hat{A}_{i,j}\neq \hat{A}_{i,j'}$ in $\hat{\mathcal{P}}_i$ are vertex-disjoint.
    \end{claim}

   Let $\hat{\mathcal{P}}$ 
   be the union of all partitions  $\hat{\mathcal{P}}_1,\ldots, \hat{\mathcal{P}}_k$ together with the star-center vertices $V_S=V$ (the only $\hat{G}_\ell$ vertices that do not appear in any $\hat{\mathcal{P}}_i$). Then $\hat{\mathcal{P}}$ can be seen as a partition $\hat{\mathcal{P}}$ of $\hat{G}_\ell$ itself. 
    Thus, we have: 
    \begin{claim}
    $\hat{\mathcal{P}}$ is a partition of $\hat{G}_\ell$ into vertex-disjoint connected subgraphs.
    Computing PA on $\hat{\mathcal{P}}$ where the input of a leader vertex of a node-part $f$ in  $\hat{p}(f)$ is $x_f$ and otherwise an identity element, results in the same output of computing PA on the union of $\mathcal{P}_1,\ldots,\mathcal{P}_k$ in level $\ell$ of $\mathcal{T}$.
    \end{claim}

    Thus, we compute part-wise aggregations on $\hat{\mathcal{P}}$ in $\hat{G}_\ell$ in order to simulate part-wise aggregations on $\mathcal{P}$.
We explain the input and output format to the part-wise aggregation task in more detail.    
    First, note that the $V_S$ star-centers and the edges $\ES$ do not participate as  input because they do not appear on any $\hat{p}(f)$ or $\hat{A}_{i,j}$ of $\hat{G}_\ell$ (\cref{lemma: faceparts_hat_G} and \cref{lemma: annuli_components_hatG}). They are used exclusively for communication. 
    Initially, all vertices in a part $\hat{A}_{i,j}$ know the ID of $\hat{A}_{i,j}$. This is because $v\in G$ knows, for each $f\in A_i$ where $v\in p(f)$, the ID of $A_{i,j}$ that $f$ belongs to (assumption), and because $v$ knows the correspondence between its copies in $\hat{G}_\ell$ and $\hat{p}(f)$ (\cref{lem: distributed_knowledge}). In addition, 
    the partition $\hat{\mathcal{P}}$ into  vertex-disjoint subgraphs is known to the vertices of $\hat{G}_\ell$ (\cref{lemma: construct_hat_G}).

    Recall that in the PA task, each node $f$ in a part of the partition has an $\tilde{O}(1)$-bit input,
    and the aggregate operator is computed over all inputs of nodes that are in the same part (\cref{def: PA}).
    When computing a PA task over nodes of $A_i$, the input of $f\in A_i$ is required to be known to the leader vertex of $\hat{p}(f)$ in $\hat{G}_\ell$ (see \cref{lem: distributed_knowledge}). Then, aggregations inside all $\hat{A}_{i,j}$ simulate aggregations inside $A_{i,j}$, where only the input of the leader of $\hat{p}(f)$ is considered as an input (all other vertices on $\hat{p}(f)$ participate with an identity element).
    If the PA task is over inside-part edges or over outgoing edges of parts $A_{i,j}$, then each vertex $v_i$ of $\hat{G}_\ell$ that is incident to an edge $\hat{e}=\{v_i,v_{i+1}\}\in \EC$ (recall, edges of $\EC$ map 1-1 to edges of $A_i$s by \cref{lemma: dual_edges_hat_G}), knows if $\hat{e}$'s input should or should not be considered in the aggregation. As $v_i$ knows the ID of the part $\hat{A}_{i,j}$ that contains it, and knows the ID of $\hat{A}_{i,j'}$ that contains $v_{i+1}$ locally (because both $v_i,v_{i+1}$ are simulated by $v\in G$).
    Note that this is the only case, i.e.,  for $i\neq i'$ there are no edges between any part $A_{i,j}$ and $A_{i',j'}$ because $A_{i}$ and $A_{i'}$ are distinct annuli (disconnected graphs) in level $\ell$. 
    Hence, aggregations on all $\hat{A}_{i,j}$ simulate aggregations on the parts $A_{i,j}$, where any $\hat{G}_\ell$ vertex knows whether it should participate in the aggregation and if so with which input (as explained above).

    Finally, by \cref{cor: PA_g_hat}, the PA task on $\hat{G}_\ell$ for the partition $\hat{\mathcal{P}}$ (i.e., computing aggregations on all $\hat{A}_{i,j}$ simultaneously) is done in $\tilde{O}(D)$ rounds. Upon termination, all vertices in $\hat{G}_\ell\setminus V$ know the output value of their part. This directly translates from $\hat{G}_\ell$ to $G$, as each vertex of $G$ knows all information of its copies in $\hat{G}$ (\cref{lemma: congest_hat_G}).
    In particular, $v\in G$ knows the output of a node any $f$ that is contained in an annulus $A_i$ of level $\ell$ s.t. $v\in p(f)$.      
    This concludes the proof of \cref{lem: PA_level_T}.
    \qedhere
    
\end{proof}

\subsubsection{Standard Minor-aggregation on Annuli} 
\label{sec: standard_model_simulation}
Next, we use the above lemma in order to simulate minor-aggregation algorithms on level-$\ell$  annuli, concluding \cref{thm: MA_simulation}.

\begin{proof}[Proof of \cref{thm: MA_simulation}.]
Let $A_1,A_2,\ldots,A_k$ be the level-$\ell$ annuli.
We explain how to simulate the minor aggregation model using PA on an arbitrary annulus $A_i$ of level $\ell$. The parallelization to all annuli in the same level follows from \cref{lem: PA_level_T}. In particular, each of the $r$ rounds of the minor-aggregation algorithm is simulated by a fixed number of rounds in $\tilde{O}(D)$, s.t. the PA problem being solved in any given moment is the same for all annuli $A_i$. I.e., we would have a problem if different annuli of level $\ell$ are in different stages of the minor-aggregation algorithm. The issue is then that some annuli may want to change the partition on which we solve PA (according to the algorithm they run), meaning that in $\hat{G}_\ell$ we would need to construct low-congestion shortcuts again despite the fact that perhaps other annuli did not finish their PA task yet.
This issue is easily solved by scheduling the update of partitions  ( i.e., reconstruction of low-congestion shortcuts in $\hat{G}_\ell$) to start at specific rounds that are given by the simulation algorithm. E.g. divide the time into intervals of $t=D \log^c n$ (for some constant $c$) rounds, such that the partitions are updated at the beginning of each interval, and the aggregation task is computed in the remaining time until the beginning of the next interval.

We proceed to the simulation for the annulus $A_i$.
Let $s=\{f_1,f_2,\ldots\}$ be a super-node of $A_i$ ($f_i\in V(A_i)$). The vertices of $\hat{G}_\ell$ that map to $s$ are the vertices that lie on the paths $\hat{p}(f_i)$ of $\hat{G}_\ell$ that represent $p(f_i)$.
Then, the computational power of a super-node of $A_i$ is simulated by all vertices of $\hat{G}_\ell$ that map to it.
The computational power of an edge in (a minor of) $A_i$ is simulated by the endpoints of its corresponding $\EC$ edge in $\hat{G}_\ell$ (\cref{thm: hat_G}, \cref{lemma: dual_edges_hat_G}). We use \cref{lem: PA_level_T} 
as a black-box for PA on $A_i$ in order to implement the simulation.

There are three (ordered) steps in a minor-aggregate round: contraction, consensus,  and aggregation. There is an additional implicit initial step at the start of each round for electing a leader of each super-node. We follow~\cite{GHSYZ22} and show how all steps compile down to PA in $A_i$. Obviously, we differ from~\cite{GHSYZ22} for simulating the model on dual (rather than primal) subgraphs (the annuli). For that purpose, we consider our communication graph to be $\hat{G}_\ell$ and not $G$ (\cref{lemma: congest_hat_G}).

\medskip 
\noindent
{\bf  1. Electing a leader (and contraction).}
In the contraction step, we need to contract every edge with value one. First, we merge {\em clusters} of nodes of $A_i$ over their outgoing edges of value one (clusters eventually define super-nodes of $A_i$). Second, we assign each cluster an ID that is known to all vertices of $\hat{G}_\ell$ that simulate it.
Initially, a cluster is a single node of $A_i$. Then, clusters are grown to match the set of super-nodes of (the minor of) $A_i$ in the current round of the minor-aggregation algorithm being simulated. 
    To implement this, we first detect
connected components of $\hat{G}_\ell[\ER]$ (paths and cycles of $\hat{G}_\ell$ that map to nodes $f\in A_i$), where the ID of a node $f$ of $A_i$ is the maximal ID of a vertex in $\hat{G}_\ell$ that lies on $\hat{p}(f)$.  
This is done in $\tilde{O}(D)$ rounds as explained in the proof of \cref{lem: distributed_knowledge}.
Now, each vertex of $\hat{G}_\ell$ knows its connected component, and each node of $A_i$ is defined as a cluster.
We can apply Lemma \ref{lem: PA_level_T} (PA in $A_i$) due to the fact that all vertices of $\hat{G}_\ell$ that correspond to a node $f\in A_i$ in a cluster, know the same ID for that cluster.
    
 The merging process follows Boruvka's classic scheme, that does $O(\log n)$ star-shaped merges. The classic implementation is similar to that of~\cite{GhaffariH16_shortcuts, GP17, GZ22}.
Namely, the PA task on $A_i$ is when each cluster defines a part in the partition of $V(A_i)$ and the aggregate operator is computed over the clusters' outgoing edges that chose a value of 1 in the contraction step.
This determines the merges as follows: the operator computed is over labels of neighboring clusters incident to those edges, where each cluster is labeled as a {\em joiner} or a {\em receiver}. Joiners suggest merges and receivers accept them (i.e., receivers are the star centers in the star-shaped merges). Randomly, the labels are decided by a fair coin flip that the cluster's leader tosses and broadcasts to the entire cluster. Deterministically, the classic algorithm for 3-coloring by Cole and Vishkin \cite{cole-vishkin} is used (E.g. see~\cite{GZ22} for more details).
In any case, at most $\tilde{O}(1)$ merges are repeatedly done until clusters match super-nodes of $G^*$. 
In each phase of merges, all nodes of the same cluster maintain the cluster's ID to be the maximum ID of a node $f\in A_i$ that is in the  cluster (then, clusters' IDs define the partition of $A_i$ into clusters, which is the input of the next phase of the $O(\log n)$ merging phases). 
The vertex of $\hat{G}_\ell$ that has that maximum ID is considered to be the leader of vertices in $\hat{G}_\ell$ that map to the super-nodes represented by the cluster. 
At the end of each merging phase \cref{lem: PA_level_T} (PA on $A_i$) is applied. Overall, it is applied repeatedly $\tilde{O}(1)$ times, and after $\tilde{O}(D)$ rounds the clusters are the super-nodes and all vertices of $\hat{G}_\ell$ that simulate the same super-node have a shared leader and ID.

\medskip 
\noindent 
{\bf 2. Consensus.}
After assigning IDs to clusters, each vertex of $\hat{G}_\ell$ knows the ID of the cluster that it maps to, which is used as the part ID in the part-wise aggregation tasks on $A_i$. 
Regarding input, the input is perhaps unique per node of $A_i$, and  the elected leader in $\hat{G}_\ell$ knows it for the node $f\in A_i$ that it leads.
Hence, Lemma \cref{lem: PA_level_T} is applicable and is used to implement the consensus step in $\tilde{O}(D)$ rounds on $G$. 
    
\medskip 
\noindent 
{\bf 3. Aggregation.}
After performing the consensus step, each endpoint $v_i$ in $\hat{G}_\ell$ of an edge in $\EC$
informs the other endpoint $v_{i+1}$ with the consensus value of the cluster that $v_i$ maps to. After that, each endpoint $v_i\in \hat{G}_\ell$ of an edge $e\in \EC$ that maps to the edge $e^*=(a_1,a_2)$ (where $a_1,a_2$ are distinct super-nodes of the minor of $A_i$) thinks of the edge $e$ as if it was assigned the value/weight of $z_{e,a_i}$ if $v_i$ is in the cluster (that maps to) $a_i$. Otherwise (when the endpoints $a_1,a_2$ of $e$ map to the same cluster), $v_i$ considers the identity element as input for $e$.
Again, the aggregation is done by \cref{lem: PA_level_T} in $\tilde{O}(D)$ rounds on $G$.
Finally, the two endpoints of each edge know the output of the aggregate operator. Thus, they can choose zero or one for that edge (the input to the next contraction step). 

Note, after the simulation finishes, each vertex $v$ of $G$ knows the output value of the super-node that contains $f\in A_i$ where $v\in p(f)$, and the output of each dual edge $e^*$ that corresponds to an incident edge to $v$. This information is known to $v$ from its copies in $\hat{G}_\ell$ that simulate $f$ and $e^*$ (\cref{lemma: construct_hat_G}). 
\qedhere
    
\end{proof}

\subsubsection{Minor-aggregation on Annuli with Virtual-nodes} 
\label{sec: virtual_nodes_simulation}
We aim to simulate the {\em extended} minor-aggregation model, which is the model as defined in \cref{def: basic_model} after adding to the annulus $\tilde{O}(1)$ arbitrary connected {\em virtual-nodes}. See~\cite{GZ22} for a formal definition and discussions.
We need to use this (in a black-box manner) in order to have an SSSP algorithm that runs in $\tilde{O}(D)$ rounds and obeys the approximate-triangle inequality.

The proof that we can implement this model is also a black-box corollary of implementing the standard model (\cref{def: basic_model}). I.e., is a corollary of the above \cref{sec: standard_model_simulation}.
In particular,~\cite{GZ22} show that there is a standard minor-aggregation algorithm (compiler) that simulates the extended model with a multiplicative overhead proportional to the amount of virtual nodes. Thus, we use that result in a black-box manner (see Section 4.1 and Theorem 14 of~\cite{GZ22} for details), which immediately allows us to simulate the extended model, concluding the proof of the stimulation theorem on annuli (\cref{thm: MA_simulation}).

\begin{lemma}[Theorem 14 of \cite{GZ22}]
\label{lem: simulating-virtual-nodes}
  Let $Alg$ be a $r$-round minor-aggregation algorithm, and $H_{virt}$ be a connected virtual graph obtained from a connected input graph $H$, that contains $\beta$ virtual nodes. There is a $O(\beta\cdot r)$-round (standard) minor-aggregation algorithm $Alg'$ that runs in $H$ and simulates $Alg$ on $H_{virt}$. Upon termination, each non-virtual node $f \in H$ learns all information learned by $f$ and all virtual nodes.
\end{lemma}

We elaborate on how this translates to simulating extended minor-aggregation algorithms on all annuli of level $\ell$, where each annulus may contain up to $\tilde{O}(1)$ virtual nodes. At the beginning this may seem not possible, for in certain levels of recursion we might have up to $\Omega(n)$ annuli, and thus a total of $\Omega(n)$ virtual nodes across the level that need to be simulated. However, those virtual nodes are simulated the same way we simulate annuli. That is, each annulus corresponds to a distinct connected component of $\hat{G}_\ell$ (and thus $G$) that simulates it. The same component is also responsible for the virtual nodes of the annuli it simulates. 

\section{Approximate SSSP and Eulerian Orientation}
\label{sec: SSSP}
This section proves the following theorem on running SSSP simultaneously on all annuli of the same level.

\begin{theorem}
\label{th: SSSP}
    Let $\epsilon$ be a parameter known to all $v\in G$. Assume all level-$\ell$ incisions are distributively stored, and for each level-$\ell$ annulus $A_i$ the ID of a source node $f_i$ is known to all $v\in G$ where $v\in p(f_i)$. Then, in $\tilde{O}(D/\epsilon^2)$ rounds, a $(1+\epsilon)$-approximate SSSP tree $T_i$ rooted at $f_i$ is computed for all $i$ simultaneously. Moreover:
    \begin{enumerate}
        \item $T_i$ obeys the approximate triangle inequality (as defined in Appendix \ref{appendix: centralized_missing_proofs}).
        
        \item For each node-part $g\in A_i$, the length of the $f_i$-to-$g$ path in $T_i$ is known to every vertex $v\in p(g)$. 
        
        \item 
        Each vertex $v\in G$ knows for each incident edge $e$, in what trees $T_i$ (a copy of) $e$'s dual is contained (and the IDs of its endpoints in $T_i$). 
        
    \end{enumerate}
\end{theorem}

To prove \cref{th: SSSP}, we wish to use the approximate SSSP algorithm of~\cite{RozhonGHZL22_shortestpaths}. This raises three obstacles. First, an annulus can  be a multi-graph and to use the algorithm of~\cite{RozhonGHZL22_shortestpaths} we need to turn it into a simple graph. Second, we need to turn the algorithm of~\cite{RozhonGHZL22_shortestpaths} into one that obeys the approximate triangle inequality.   
Both these obstacles are easily resolved by simulating two minor-aggregation procedures from previous works~\cite{RozhonHMGZ23_triangle_inequality,planardistributedmaxflow24}, which is possible by using \cref{thm: MA_simulation}. See \cref{section: SSSP}  for more details.

The third obstacle is the most challenging: The original algorithm of~\cite{RozhonGHZL22_shortestpaths} is not (fully) a minor-aggregation algorithm.
Specifically, all its components do work in the minor-aggregation model except for one - an  {\em Eulerian orientation oracle}. The oracle was implemented in~\cite{RozhonGHZL22_shortestpaths} using $\tilde{O}(1)$ $\CONGEST$ rounds and $\tilde{O}(1)$ instances of the PA problem.
This is not good enough for our purposes as we do not know how to simulate $\CONGEST$ on annuli (in fact, it is not known whether $\CONGEST$ can be even simulated on $G^*$ itself without suffering the trivial $O(n)$ blow-up in round complexity). 
We discuss this in this section, where we define and refine the Eulerian orientation oracle to run in $\tilde{O}(D)$ rounds on all annuli of the same level simultaneously. After solving the problem in \cref{sec: oreintation_for_annuli}, we complete the proof of \cref{th: SSSP} in \cref{section: SSSP}.

\medskip
\noindent
{\bf The Eulerian orientation oracle \boldmath$\mathcal{O}_{Euler}$.}
In this problem one is given an {\em Eulerian graph} $H$ as an input (a graph is Eulerian iff all its degrees are even), and needs to orient its edges such that the in-degree equals the out-degree for all nodes of $H$. 
The problem that we need to solve in order to implement~\cite{RozhonGHZL22_shortestpaths} SSSP algorithm on annuli is the following.

\begin{definition}[$\mathcal{O}_{Euler}$, Definition 3.8 of \cite{RozhonGHZL22_shortestpaths}]
\label{def: eulerian_oracle}
Let $A$ be the input graph to the SSSP algorithm. The problem asks to find an Eulerian orientation to an Eulerian graph $H$, which is a subgraph of $A_{virt}$, where $A_{virt}$ is $A$ with up to $\tilde{O}(1)$ additional virtual nodes.
\end{definition}

\noindent {\bf A simple (centralized) solution.}
We describe a classic simple solution, the barriers of turning it into a distributed solution, and how to surpass them. See also \cref{fig: eulerian_orientation}.
\begin{lemma}[\cite{Euler_par_AtallahV84}]
\label{lem: centralized_euler_orientation}
    The following procedure solves the Eulerian Orientation problem:
\begin{enumerate}
    \item {\em Local step.}
    Decompose $H$ into edge-disjoint cycles by pairing the edges incident to each node $v\in H$.

    \item {\em Global step.} Identify and consistently direct each cycle of the decomposition.
\end{enumerate}
\end{lemma}

\begin{figure}[htb]
\centering
\begin{subfigure}[t]{0.3\textwidth}
  \centering
   \includegraphics[width=\linewidth]{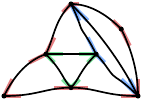}
\end{subfigure}
\hspace{1.7cm}
\begin{subfigure}[t]{0.29\textwidth}
\centering
\includegraphics[width=\linewidth]{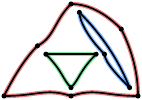}
\end{subfigure}
\caption{Left: an Eulerian graph where edges with the same color incident to a node are paired. Note that pairing of edges corresponds to edge-disjoint closed walks (cycles) of the graph. Right: The resulting two-regular graph from splitting each node $v$ into $\deg(v)/2$ nodes (a node per two paired edges). 
\label{fig: eulerian_orientation}
}
\end{figure} 

 To see how the local step gives a decomposition into edge-disjoint cycles, imagine splitting each node $v\in H$ into $\deg(v)/2$ nodes, each incident to two paired edges of $v$. Then, the resulting graph is 
 a union of node-disjoint cycles (see \cref{fig: eulerian_orientation} right). 
 Clearly, there is a 1-1 mapping between those cycles and cycles of $H$ defined by the edge pairing. Note, each two paired edges incident to a vertex $v\in H$ participate in exactly one cycle. Directing this cycle in a consistent manner, directs paired edges of each node that participates in the cycle in opposite directions.
 
 Note that when the degrees are not all even, then each node might have one unpaired edge, and the resulting  graph is of maximum degree two, which is a union of cycles and paths. We emphasize this because in certain cases of the distributed algorithm, we may also deal with paths. 
 
 \medskip
\noindent
{\bf Towards a solution on annuli.}
 Indeed, this approach of pairing edges is the main idea to solve the problem in various models of computation.
 E.g. the parallel algorithms of~\cite{Euler_par_AtallahV84, Euler_par_CaceresDSS93}, the distributed HYBRID algorithm of~\cite{Euler_hybrid_ChangHLS24}, and even the  $\CONGEST$ algorithm of~\cite{RozhonGHZL22_shortestpaths} that we cannot run on annuli.
 
 The main issue with this approach in the  $\CONGEST$ model is to run an algorithm on all cycles simultaneously in order to direct each of them. That is because the cycles are not node-disjoint and may be of large diameter.
 Of course, the cycles
 being edge disjoint intuitively should help but it is not clear how. The work of~\cite{RozhonGHZL22_shortestpaths} deals with those challenges and devises a fairly complicated algorithm that solves the problem for various graph families with different round complexity tradeoffs.  We face a similar challenge although we are in a different setting. Namely, we can run minor-aggregation algorithms on annuli but not on multiple edge-disjoint subgraphs of them.
 However, we show that the centralized approach can be adapted to direct a set of edge-disjoint paths (cycles) of each annulus.
 The solution uses tools developed in the previous section. In particular, we show that such a set of paths of annuli maps to a related set of edge-disjoint subgraphs in $\hat{G}_\ell$. Then, with a small change to $\hat{G}_\ell$, those subgraphs are made vertex-disjoint. Thus, we can work in all of them simultaneously using low-congestion shortcuts in order to 
 orient the paths of annuli.
    
We define a notion {\em edge paring} that utilizes planarity, so that it can be used to implement the centralized \cref{lem: centralized_euler_orientation} distributively in annuli.
    \begin{definition}[Annuli edge pairing] 
    \label{def: edge_pairing}
    Let $f$ be a node-part defined by edges $e,e'$ (see \cref{def: nodepart}).
    An edge pairing splits edges incident to $f$ into pairs of consecutive edges in $f$'s local embedding. I.e., into pairs $(e_i,e_{i+1})$ for $e\leq e_i< e_{i+1}\leq e'$. 
     We say that an edge paring is {\em distributively stored} if:
    For each node $f$, each incident edge $e$ knows whether it is the first or the second edge in $f$'s pair that contains it.
    If $f$'s degree is odd, then one edge of $f$ is allowed to be not paired.
    That edge knows itself.
    \end{definition}

\subsection{Eulerian Orientation on Annuli}
\label{sec: oreintation_for_annuli}
We now show how to solve the Eulerian Orientation problem, as required by the SSSP algorithm of~\cite{RozhonGHZL22_shortestpaths} (see \cref{def: eulerian_oracle}).

\begin{lemma}[$\mathcal{O}_{Euler}$ for annuli]
\label{lem: distributed_eulerian_orientation}
There is an $\tilde{O}(D)$-round algorithm that implements the Eulerian orientation oracle $\mathcal{O}_{Euler}$ (\cref{def: eulerian_oracle}) for all level-$\ell$ annuli simultaneously.
\end{lemma}

The solution of~\cite{RozhonGHZL22_shortestpaths} consists of two steps. First, virtual nodes of $H$ get eliminated, producing an Eulerian subgraph $H'$ of $A$ (Lemma 6.4 and Proposition 6.5 of~\cite{RozhonGHZL22_shortestpaths}). $H'$ is then handled using different approaches related to the graph's family (Section 6.2 of~\cite{RozhonGHZL22_shortestpaths}). To implement these steps in the minor-aggregation model, we follow a similar scheme to~\cite{RozhonGHZL22_shortestpaths}. Namely, we show that their procedure for handling virtual nodes can be implemented using our techniques for all level-$\ell$ annuli. 
Afterwards, we handle the remaining subgraphs differently, by implementing the centralized procedure (\cref{lem: centralized_euler_orientation}) for annuli.
Both procedures are not implementable in the minor-aggregation model straightforwardly. However, they can be implemented easily (not particularly in the minor-aggregation model) if the following two primitives are available: (1) pairing edges in order to define a set of edge-disjoint paths, and (2) directing each of those paths in a consistent manner.
Thus, we assume that those tasks can be solved near-optimally on annuli and provide a proof to \cref{lem: distributed_eulerian_orientation}.
After which, we prove that we can indeed solve these tasks  (which is the main contribution of this section). I.e., for now we assume the following two lemmas (proofs of which are provided in \cref{sec: edge_pairing} and \cref{sec: paths_orientation}).

\begin{lemma}[Pairing annuli edges]
\label{lem: edge_pairing}
    Assume a given subgraph $H$ for each annulus $A$ of level $\ell$. For all annuli of level $\ell$ (simultaneously), we compute and distributively store (as in \cref{def: edge_pairing}) an edge pairing of all nodes $f\in H$ within $\tilde{O}(D)$ rounds. 
\end{lemma}

\begin{lemma}[Orienting edge-pairing induced paths]
\label{lem: orienting_paired_edges}
Assume a distributively stored (as in \cref{def: edge_pairing}) edge pairing of a subgraph $H$ for each annulus of level $\ell$. 
Then, within $\tilde{O}(D)$ rounds on $G$, one can consistently orient the  edge-pairing induced paths
in all annuli simultaneously. Upon termination, each edge of $H$ knows its direction.
\end{lemma}

Given the above two lemmas, we now prove \cref{lem: distributed_eulerian_orientation}.

\begin{proof}[Proof of \cref{lem: distributed_eulerian_orientation}.]
Recall (\cref{def: eulerian_oracle}) that for each level-$\ell$ annulus $A$ the problem asks to find an Eulerian orientation to an Eulerian graph $H$, which is a subgraph of $A_{virt}$, where $A_{virt}$ is $A$ with up to $\tilde{O}(1)$ additional (arbitrarily-connected) virtual nodes. 
$H$ is a graph produced by the SSSP algorithm of~\cite{RozhonGHZL22_shortestpaths}. Concretely, by the $\tilde{O}(1)$ minor-aggregation rounds simulated on $A$ in the other steps of~\cite{RozhonGHZL22_shortestpaths}'s algorithm, before the Eulerian orientation step, which we solve here.
We follow~\cite{RozhonGHZL22_shortestpaths} and work in two phases. The first eliminates virtual nodes from $H$ (and some other edges) leaving a residual Eulerian graph $H'$ that contains no virtual nodes ($H'$ is a subgraph of $A$). The second phase solves the problem on $H'$.

\medskip
\noindent
\underline{\em Phase two.}
We discuss the second phase first. It follows directly from the above  \cref{lem: edge_pairing} and \cref{lem: orienting_paired_edges}. I.e., 
after the first phase terminates, we get the subgraph $H'$ of $A$ (for all $A$ of level $\ell$).
We then need to solve the Eulerian orientation problem on all graphs $H'$ for all annuli $A$ of level $\ell$. 
Namely, by the end of the first phase, each edge of $A$ knows whether it is in $H'$ or not. Then,
we apply the centralized approach of \cref{lem: centralized_euler_orientation} on $H'$, implemented as follows: via \cref{lem: edge_pairing}, for all annuli of level $\ell$ simultaneously, we compute an edge-pairing of $H'$ within $\tilde{O}(D)$ rounds. 
This results in a decomposition of each annulus into a set of edge-disjoint cycles, each of which needs to be consistently oriented.
Then, by \cref{lem: orienting_paired_edges}, we simulate the global step of orienting the resulting cycles. 
This is also done within $\tilde{O}(D)$ rounds, for all annuli of level $\ell$ simultaneously (by \cref{lem: orienting_paired_edges}).
Note, the input format of \cref{lem: orienting_paired_edges} matches the output format of \cref{lem: edge_pairing} (both follow \cref{def: edge_pairing}).

Next, we focus on phase one. First, we discuss the idea and components as in~\cite{RozhonGHZL22_shortestpaths}.  Then, we prove that we can implement it for all annuli of the level simultaneously within $\tilde{O}(D)$ rounds.

\medskip
\noindent
\underline{\em General idea of phase one.}
We explain the intuitive idea that~\cite{RozhonGHZL22_shortestpaths} apply to handle the case of a single virtual node $f$. We want to find edge-disjoint cycles that contain all edges incident to $f$. To do so, they find a set of edge-disjoint paths of $H\setminus \{f\}$, where each of them can be closed to a cycle by adding  exactly one pair of $f$'s edges.
Orienting each of these paths implies an orientation on the pair of $f$'s edges that participate in the same cycle.
This takes care of all edges of $f$ and some other edges of $H$. Hence, we can eliminate $f$. Moreover, since the edges of $H$ that get directed are edges of cycles, removing those cycles keeps a residual Eulerian graph $H'$ that does not contain any virtual node.

There are two key steps in extending this approach to the case of $\tilde{O}(1)$ virtual nodes. The first is to find a certain set of edge-disjoint paths in $H$ that does not contain any virtual edge, and then orient them. Those paths extend the above idea of one virtual node to many. In particular,  orienting those paths induces an orientation to all edges of $H$ whose (exactly) one endpoint is a virtual node.
This is done in Lemma 6.4 of~\cite{RozhonGHZL22_shortestpaths}
(also a variant of Lemma 4.3.2 of~\cite{peleg-book}), which we cannot use as black-box on annuli (it is not a minor-aggregation procedure) but we can implement it using our tools, as we explain shortly in the part of the proof that discusses the implementation of phase one.
The second step is to handle edges between different virtual nodes. Those edges were not oriented in the previous step.
Both steps are done by Proposition 6.5 of~\cite{RozhonGHZL22_shortestpaths}, which in the first step applies the aforementioned Lemma 6.4 of~\cite{RozhonGHZL22_shortestpaths} multiple times. Then, in the second, broadcasts a total of $\tilde{O}(1)$-bits. Concretely, the broadcast information is roughly the virtual-node to virtual-node edges in addition to some other $\tilde{O}(1)$ edges. Then, orienting those edges can be solved locally as they are known to the entire graph.
Since broadcast is an aggregate operation, then broadcasting $\tilde{O}(1)$-bits can be implemented in the minor-aggregation model within $\tilde{O}(1)$ rounds. Note, we have no problem simulating this on all annuli of level $\ell$ simultaneously in $\tilde{O}(D)$ rounds (by \cref{thm: MA_simulation}). Thus, all we need in order to implement phase one on all annuli of the same level is an implementation of Lemma 6.4 of~\cite{RozhonGHZL22_shortestpaths} 
for finding a certain set of edge-disjoint paths and orienting them (discussed next).

\begin{claim}[Proposition 6.5 of~\cite{RozhonGHZL22_shortestpaths}]
\label{claim: virtual_edges}
Given an implementation of Lemma 6.4 of~\cite{RozhonGHZL22_shortestpaths} that runs within $\tilde{O}(D)$ rounds, in additional $\tilde{O}(D)$ rounds all edges incident to virtual nodes of $H$ are oriented (possibly with some other edges) leaving a Eulerian subgraph $H'$ of $A$.
\end{claim}

\medskip
\noindent
\underline{\em Implementation of phase one.}
We are asked to find a set of edge-disjoint paths that obeys a certain structure.
Concretely, Lemma 6.4 of~\cite{RozhonGHZL22_shortestpaths} asks for the following:

 \begin{enumerate}[leftmargin=*]
     \item[] {\em Input:} A rooted spanning tree $T$ of $H$ minus the virtual nodes, and a set of marked nodes $S\subseteq V(T)$.\footnote{For more context, $S$ is the set of non-virtual neighbors of the virtual nodes of $H$. Our implementation is independent from this fact.}

     \item[]{\em Output:}
     A set of edge-disjoint paths that are subpaths of $T$, such that all nodes of $S$ (except for possibly one) are endpoints of exactly one path. In addition, each of those paths needs to be oriented.
     
 \end{enumerate}
To find those paths, first we find the tree $T$, then, in order to find such a set of paths, a simple $\tilde{O}(1)$ minor-aggregation round subtree sum procedure is performed on $T$. Second, a step of pairing edges completes computing the paths (using  \cref{lem: edge_pairing}). After which, the paths get oriented using \cref{lem: orienting_paired_edges}. We explain each step and its implementation.

 \begin{enumerate}[leftmargin=*]
\item  {\em Finding  $T$ and subtree sums.} In Proposition 6.5 of~\cite{RozhonGHZL22_shortestpaths} the spanning tree $T$ of $A$ is computed using $\tilde{O}(1)$ instances of the PA problem (\cref{def: PA}) and $\tilde{O}(1)$ $\CONGEST$ rounds. However, by Example 4.4 of~\cite{GHSYZ22}, $T$ is  computed by a $\tilde{O}(1)$ round minor-aggregation procedure that computes a (minimum) spanning tree.

After computing $T$, we need to apply the subtree sums procedure that defines the desired paths. In particular,  using this subtree sums computation,
nodes of $T$ mark some of their incident edges that need to be paired. Then,  pairing marked edges locally defines the required set of paths. 
For context, those edges indicate merging requests that define paths. 
I.e., the computation of~\cite{RozhonGHZL22_shortestpaths} is as follows.
For each node $g \in T$, let $x_g = 1$ if $g\in S$ and $x_g = 0$ otherwise. We want to compute the subtree sum $s_g$ for every $g\in T$. This defines paths as follows. 
If $s_g$ is odd, $g$ will be connected in a path with its parent (unless $g$ is the root). Each edge connecting such a node to its parent marks itself; indicating a request of $g$ to its parent $f$ to be paired with another child of $f$, or possibly the parent of $f$ if all other children of $f$ are paired to each other. I.e., $f$ simply needs to pair its marked edges. This subtree sum procedure can be easily implemented in $\tilde{O}(D)$ rounds on all annuli of level $\ell$ using a classic $\tilde{O}(1)$ rounds minor-aggregation procedure as a black-box (simulated by \cref{thm: MA_simulation}). E.g.  Lemma 16 of~\cite{GZ22}, or Lemma B.8 of~\cite{RozhonGHZL22_shortestpaths} (see also \cref{lem: subtreesums} later in the section).

\item {\em Pairing edges.}
As mentioned, a node $f$ needs to match its marked edges. However, in our case $f$ does not have a list of its children and did not learn the IDs of marked edges as this is not feasible in the minor-aggregation model (and requires a round complexity proportional to $f$'s degree).
However, we know how to deal with this.
I.e., this problem of pairing edges can be solved on all annuli $A$ of level $\ell$ simultaneously within $\tilde{O}(D)$ rounds. Simply, we apply \cref{lem: edge_pairing} on the graph $H$ after removing all its virtual nodes (i.e., on a subgraph of $A$).
The pairing we find is distributively stored as in \cref{def: edge_pairing} (roughly, each marked edge knows whether it is the first or the last in its pair, and knows an ID of the pair). This knowledge is enough to orient the formed paths by the edge-pairing as we explain next.

\item {\em Orienting the paths.} 
After pairing the edges and distributively store the edge pairing, we use 
\cref{lem: orienting_paired_edges} in order to consistently orient each of those paths, 
for all annuli of level $\ell$  simultaneously within $\tilde{O}(D)$ rounds.
 \end{enumerate}

 It follows from \cref{claim: virtual_edges} (Proposition 6.5 of~\cite{RozhonGHZL22_shortestpaths}) that, within additional $\tilde{O}(D)$ rounds, we take care of all virtual nodes of $H$ (and some other edges) leaving an Eulerian subgraph $H'$ of $A$ that needs to be handled. That graph is handled as explained at the beginning of the proof. This concludes \cref{lem: distributed_eulerian_orientation} for implementing the Eulerian orientation oracle $\mathcal{O}_{Euler}$ (\cref{def: eulerian_oracle}). \qedhere
\end{proof}

\subsubsection{Decomposing Annuli to Edge-disjoint Paths}
\label{sec: edge_pairing}
We now prove \cref{lem: edge_pairing}, stating that for all annuli $A$ of level $\ell$, simultaneously, in $\tilde{O}(D)$ rounds we can pair edges of any given subgraph $H$ of $A$ as in \cref{def: edge_pairing}. 

Since we can run only minor-aggregation algorithms on annuli,  we have a significant problem: A node cannot pair its incident edges because it does not know their IDs, and learning the IDs of neighbors in this model has round complexity that is proportional to the maximum degree of a node, which may be $\Omega(n)$.
However, this problem can be solved easily using $\hat{G}_\ell$:
Use aggregations on $\hat{p}(f)$ in order to enumerate edges of a node $f$. This defines a pairing of the edges incident to $f$. E.g. edges that are assigned numbers $2k-1,2k$ are paired, for $1 \leq k \leq \deg(f)/2$.

\begin{proof}[Proof of \cref{lem: edge_pairing}.]
    Let $A$ be an annulus of level $\ell$, and assume a given subgraph $H$ of $A$ s.t. each edge of $A$ knows whether it is in $H$.
    We show how to pair edges of $H$.  We focus on a single node $f\in A$, the parallelization to all nodes in $A$, and to all annuli of level $\ell$ is discussed at the end of the proof.

    Note that $f$ has two types of edges (in $H$ or not).
    In order to pair edges of $H$ as described, we assign numbers (IDs) to  edges of $H$ according to their order in $f$'s embedding, then each edge of an odd number is paired with the next edge (of even number). That is, the edges $2k-1,2k$ are together for $1 \leq k \leq \lfloor\deg_H(f)/2\rfloor$. If there is a residual edge then it is not paired.
    Hence, the problem is reduced to finding such an enumeration of edges. Such an enumeration of edges can be computed via $\hat{p}(f)$ in $\hat{G}_\ell$ as we explain next.

    We are able to compute this enumeration of $f$'s edges in $H$ by computing subtree sums on $\hat{p}(f)$: (1) Assign $\hat{p}(f)$ a root, which is an endpoint of $\hat{p}(f)$ in case it is a path, and otherwise (when it is a cycle) an arbitrary vertex. (2)
    If $p(f)$ is not an entire face of $G$, then $\hat{p}(f)$ is a path (tree). Otherwise, $p(f)$ is an entire face of $G$, $\hat{p}(f)$ is a cycle, and we want to compute a spanning path of it. To do so, the root eliminates one of its incident edges in $\hat{p}(f)$, defining a spanning path of $\hat{p}(f)$.
    Now we are ready to perform the subtree sum on a (spanning path of) $\hat{p}(f)$:
    Each vertex of $\hat{G}_\ell$ incident to an $\ER$ edge that corresponds to an $H$ edge, learns that the input to that edge is one, all other edges' input is zero.\footnote{We note that an edge of $H$ might have two corresponding $\ER$ edges in $\hat{p}(f)$. In that case, we want to assign only one them an input, and the other zero. Those edges are incident to two copies $u_j,u_k$ of the same node $u$, thus, $u$ and their other endpoints can locally coordinate and choose an arbitrary $\ER$ copy of the edge to assign it an input.}
    This subtree sums on edges is reduced easily to subtree sums on vertices (the input to a vertex is the input to the edge connecting it to its parent).  After computing the subtree sums, (3) The vertices of $\hat{p}(f)$ can locally learn the numbers of their incident edges (in case $\hat{p}(f)$ is a cycle, then the eliminated edge can be assigned a number by the root if it is in $H$).
    Thus, the pairing can be deduced locally and be distributively stored.
    I.e., the endpoints $u,v$ of an edge $e$ in $G$ know whether (the dual in $A$ of) $e$ is in $H$, and if so, $u,v$ know  for each $\ER$ copy of $e$ whether it is the first (assigned an odd number) in its pair or not (assigned an even number). I.e., for each endpoint of (the dual of) $e$ in the annulus $A$, $u,v$ know whether $e$ is the first in its pair. Since $u,v$ simulate the edge $e$ in the minor-aggregation simulation of $A$, then one can say that the edge in $A$ knows the information about itself that $u,v$ know. 
    
    Electing a leader and the subtree sums procedure can be performed via simple minor-aggregation procedures performed on $\hat{p}(f)$ as the input graph. E.g. electing a leader can be done in one round of computing the maximal ID of a vertex to be the leader, and subtree sums can be performed within $\tilde{O}(1)$ minor-aggregation rounds (for example, by using the algorithm from \cref{lem: subtreesums} in the next section).
    The reason we can simulate a round of the minor aggregation model on $\hat{p}(f)$ within near-optimal $\tilde{O}(D)$ rounds, is that planar graphs of diameter $O(D)$ allow such simulation (Theorem 17 of~\cite{GZ22};~\cite{RozhonGHZL22_shortestpaths, GHSYZ22}). By \cref{thm: hat_G}, $\hat{G}_\ell$ is such a graph where a $\CONGEST$ round on it is simulated by $O(1)$ rounds on $G$.
    Finally,  the parallelization to all nodes of level $\ell$ (all $\hat{p}(f)$ in $\hat{G}_\ell$) follows  for two reasons: (1) For all nodes $f\neq g$ of level $\ell$, the graphs $\hat{p}(f)$ and $\hat{p}(g)$ are connected, and are vertex-disjoint (\cref{thm: hat_G}). 
    (2) Running a minor-aggregation algorithm in multiple vertex-disjoint subgraphs ($\hat{p}(f)$, $\hat{p}(g)$) of a graph ($\hat{G}_\ell$) in parallel can be simulated within the same round complexity of running the algorithm on a single subgraph (Corollary 11 of~\cite{GZ22};~\cite{RozhonGHZL22_shortestpaths, GHSYZ22}).
\end{proof}

\subsubsection{Orienting Edge-disjoint Paths in Annuli}
\label{sec: paths_orientation}
    It is not hard to see that an edge pairing in an annulus corresponds to a collection of edge-disjoint paths (see the discussion following \cref{lem: centralized_euler_orientation}). 
    We aim to orient each of those paths consistently. I.e., to prove \cref{lem: orienting_paired_edges}.
    Trying to do that on annuli straightforwardly does not work, as each path might be very long and the round complexity would depend on the paths' diameter.  
    To tackle this problem, we turn to $\hat{G}_\ell$ and observe that the problem on annuli corresponds to a similar problem on $\hat{G}_\ell$. 
    In particular, we take a primal point of view on edge-pairings in annuli. Then, we derive several properties related to $\hat{G}_\ell$ and paths given by the edge-pairing. After which, we employ those structural properties to obtain a new graph $\hat{G}_\ell'$ whose structure is related to the paths induced by the edge-pairing.
    In particular, paths defined by the edge-pairing  correspond to vertex-disjoint subgraphs of this new graph. By performing computations on those subgraphs, we simulate an orientation procedure, producing the promised orientation of the paths.

\medskip
\noindent
{\bf Paths defined by edge pairings and $\hat{G}_\ell$.}
Note that when \cref{def: edge_pairing} is applied to a node $f$ in a subgraph $H$ of the annulus $A$ ($f\in A$), then the definition discusses $f$'s incident edges in $H$. See also \cref{fig: edge_pairing}.
\begin{figure}[htb]
\centering
\begin{subfigure}[t]{0.2\textwidth}
  \centering
   \includegraphics[width=\linewidth]{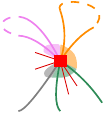}
\end{subfigure}
\hspace{2.7cm}
\begin{subfigure}[t]{0.2\textwidth}
\centering
\includegraphics[width=\linewidth]{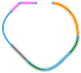}
\end{subfigure}
\caption{Left: A node $f$ and an edge pairing of it in a subgraph $H$ of $A$. Edges of $A\setminus H$ incident to $f$ are {\color{red}red}.
Paired edges of $H$ are demonstrated by colors: {\color{YellowOrange}orange}, {\color{VioletRed}pink}, {\color{OliveGreen}green}. The {\color{gray}gray} edge is the only non-paired edge, because $f$'s degree in $H$ is odd. 
Right: The corresponding face-part $p(f)$ and its partition into internally-disjoint subpaths $\hat{p}(f)_i$, which is demonstrated by shadow colors. Note the mapping between subpaths $\hat{p}(f)_i$ and paired edges of $f$ in $H$.
Note that a pair of edges in $H$ (for example the {\color{OliveGreen}green} pair) might map to a subpath of $\hat{p}(f)$ that contains edges whose dual is not in $H$. This keeps the paths $\hat{p}(f)_i$ connected. In addition, some edges of $\hat{p}(f)$ whose dual is not in $H$ do not participate in any path $\hat{p}(f)_i$; For example, the edges of $A\setminus H$ in between the {\color{gray}gray} and {\color{VioletRed}pink} edges (clockwise).
\label{fig: edge_pairing}
}
\end{figure} 
   Taking a primal point of view on an edge pairing, we note that each two paired edges of $f$ are on a unique subpath of $\hat{p}(f)$.
   This is because two consecutive edges of $f$ are also consecutive in $\hat{p}(f)$.
    Moreover, those subpaths overlap only with their endpoints.
    See \cref{fig: edge_pairing}. Therefore: 
    
    \begin{lemma}
    \label{lem: pairing_partition_pf}
    Let $H$ be a subgraph of an annulus $A$, and assume a given edge pairing of $H$. For a node $f\in H$, any two paired edges correspond to a subpath of $\hat{p}(f)$, denoted $\hat{p}(f)_i$ for the $i$'th pair. The paths $\hat{p}(f)_i,\hat{p}(f)_j$ are internally disjoint for any $i\neq j$.
    Since there are edges of $A$ not in $H$, $\hat{p}(f)$ minus the paths $\hat{p}(f)_i$ may contain some other paths corresponding to those edges. 
    \end{lemma}
    
    \begin{proof}
        A pair of consecutive edges of $f$ in $H$ corresponds to a sequence of edges of $f$ in $A$.  This sequence contains the two paired edges in $H$ and all $A$ edges between them in $f$'s embedding. In addition, any two disjoint sequences of edges of $f$ corresponds to edge-disjoint subpaths of $\hat{p}(f)$, by \cref{lemma: faceparts_hat_G}.
        Edges of $A\setminus H$ that are not in between two paired edges (in the local ordering of $f$) correspond to their own subpaths of $\hat{p}(f)$.
    \end{proof}
    
    Using this lemma, we can make the paths $\hat{p}(f)_i$ vertex disjoint by simply splitting their endpoints in $\hat{G}_\ell$. This allows us to perform computations on components of the resulting graph in order to simulate orienting the paths of $H$ defined by the edge pairing.

    Concretely, let $Q=q_1,q_2,\ldots$ be some path in an annulus, defined by the given pairing of edges. I.e., for each node $q_j$, its incident edges $(q_{j-1},q_j),(q_j,q_{j+1})$ are paired and participate in $Q$ (one of those edges might be empty if $q_j$ is an endpoint of $Q$ when $Q$ is a path and not a cycle). In addition, by \cref{thm: hat_G} each edge of an annulus (in particular of $Q$) has a corresponding $\EC$ edge in $\hat{G}_\ell$.  Moreover, by \cref{lem: pairing_partition_pf} this pair of edges of $q_j$ corresponds to a subpath $\hat{p}(q_j)_{i_j}$ of $\hat{p}(q_j)$ ($\ER$ edges). 
    Then (see also \cref{fig: edge_pairing_paths}):
    
    \begin{claim}
    \label{claim: edge_pairing_paths_hatG}
    The path $Q$ corresponds to a subgraph $\hat{Q}$ of $\hat{G}_\ell[\ER\cup \EC]$, defined as follows: the subpath $\hat{p}(q_1)_{i_1}$ of $\hat{p}(q_1)$,     
 the $\EC$ edge mapping to edge $(q_1,q_2)$, 
    the subpath $\hat{p}(q_2)_{i_2}$ of $\hat{p}(q_2)$, the $\EC$ edge mapping  to $(q_2,q_3)$, and so on.
    \end{claim}
    \begin{proof}
        
   We rephrase the claim:  For $q_j\in Q$, the $\EC$ edges in $\hat{G}_\ell$ corresponding to its (up to) two edges $(q_{j-1}, q_j),(q_j,q_{j+1})$ in $Q$ are in fact connected. In particular, the subpath $\hat{p}(q_j)_{i_j}$ (consisting of $\ER$ edges) of $\hat{p}(q_j)$ connects them. The vertices of this subpath are not incident to any $\EC$ edge that corresponds to (an edge of) $H$ except for $(q_{j-1}, q_j),(q_j,q_{j+1})$.
    Next, we prove that this is correct.
    The edges $(q_{j-1}, q_j)$, $(q_j,q_{j+1})$ are paired by $q_j$, and by \cref{lem: pairing_partition_pf} two paired edges correspond to a subpath $\hat{p}(q_j)_{i_j}$ of $\hat{p}(q_j)$ whose ($\ER$) edges correspond to $q_j$'s edges in $A$. The local ordering of those edges is in between the edges $(q_{j-1}, q_j),(q_j,q_{j+1})$. 
    Moreover, by definition of $\hat{G}_\ell$, an $\EC$ edge that corresponds to an edge $e$ is incident to a an endpoint of an $\ER$ copy of $e$.
    Hence, the only edges of $\hat{p}(q_j)_{i_j}$ that correspond to $H$ edges, are the edges that correspond to $(q_{j-1}, q_j),(q_j,q_{j+1})$. Finally, by definition of $\hat{G}_\ell$, the path $\hat{p}(q_j)_{i_j}$ consists only of $\ER$ edges.
    \end{proof}

\begin{figure}[htb]
\centering
\begin{subfigure}[t]{0.3\textwidth}
  \centering
   \includegraphics[width=\linewidth]{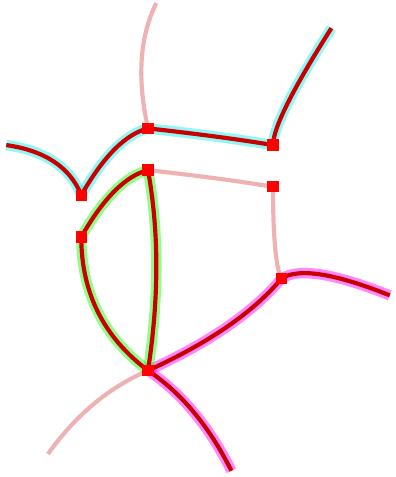}
\end{subfigure}
\hspace{1.7cm}
\begin{subfigure}[t]{0.33\textwidth}
\centering
\includegraphics[width=\linewidth]{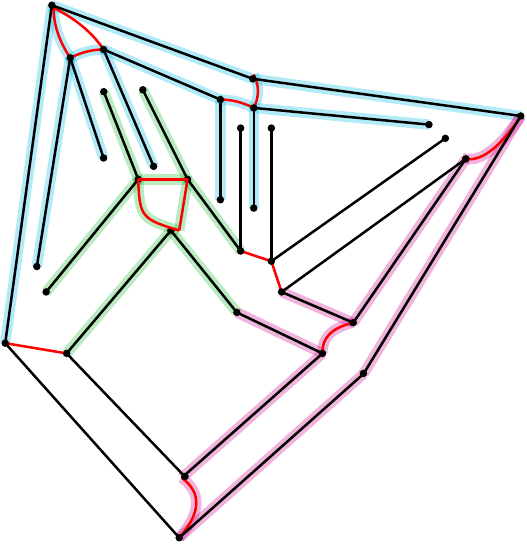}
\end{subfigure}
\caption{Left: an annulus and a subgraph whose edges are paired  into cycles $Q$ (highlighted). For simplicity, a node of the graph is omitted. Its incident edges are all edges missing an endpoint.
Right: the graph $\hat{G}_\ell[\ER\cup \EC]$. $\ER$ edges are {\bf black}, $\EC$ edges are {\color{red}red}, and the highlighted edges  constitute subgraphs $\hat{Q}$ that corresponds to cycles $Q$.
\label{fig: edge_pairing_paths}}
\end{figure}

    From the above claim, we can deduce:    \begin{claim}\label{claim: edge_pairing_paths_edgedisjoint}
   For any two paths $Q\neq Q'$ defined by the given edge pairing, $\hat{Q}$ and $\hat{Q'}$ are edge-disjoint. 
    \end{claim}
    \begin{proof}
    Since $Q,Q'$ are edge-disjoint, then $\hat{Q}$ and $\hat{Q'}$ do not share any $\EC$ edge (each $H$ edge has a corresponding unique $\EC$ edge, by \cref{lemma: dual_edges_hat_G}).
    Thus, the only way for $\hat{Q}$ and $\hat{Q'}$ to not be edge-disjoint is for $Q,Q'$ to intersect in a node in $H$. That is because, if all their nodes are different, their corresponding face-parts in $\hat{G}_\ell$ are vertex disjoint by \cref{lemma: faceparts_hat_G}. I.e., the subpaths of those face-parts with some $\EC$ edges constitute $\hat{Q},\hat{Q'}$, and since none of them are shared, then $\hat{Q},\hat{Q'}$ are vertex-disjoint. 
    Therefore, assume $Q,Q'$ intersect in nodes, and let $q$ be such an arbitrary node. Note, the pair of edges in $Q$ and the pair of edges in $Q'$ that are incident to $q$ are disjoint (by definition of edge pairing, \cref{def: edge_pairing}). Then, by \cref{lem: pairing_partition_pf}, the subpath of $\hat{p}(q)$ corresponding to the two edges of $q$ in $Q$ and the subpath of $\hat{p}(q)$ corresponding to the two edges of $q$ in $Q'$ are edge disjoint, and if they intersect it is only in their endpoints.
    \end{proof}

    Now  we to use the above claims in order to perform computations on the subgraphs $\hat{Q}$ of $\hat{G}_\ell$, which will simulate computations on the paths $Q$ defined by the given edge-pairing in (subgraphs $H$ of) annuli of level $\ell$.
    We are able to perform computations in parallel on the subgraphs $\hat{Q}$ despite them being (only) edge-disjoint in $\hat{G}_\ell$. 
    However, we show next that by a small change to $\hat{G}_\ell$, resulting in a graph denoted $\hat{G}_\ell'$, we make $\hat{Q},\hat{Q}'$ vertex-disjoint. Which is then used to consistently orient the paths $Q$.

    \medskip
    \noindent
    {\bf The graph  $\hat{G}_\ell'$.} 
    We focus on a single annulus $A$ and a single node $q\in H$, and explain the change we do on $\hat{G}_\ell$.
    We split $\hat{p}(q)$ so that all $\hat{p}(q)_i$ become vertex-disjoint.
    To do so, for all $i$, a shared endpoint to $\hat{p}(q)_i$ and $\hat{p}(q)_{i+1}$ in $\hat{G}_\ell$ gets expanded into two vertices, one that is defined to be in $\hat{p}(q)_i$ and the other to be in $\hat{p}(q)_{i+1}$.

\begin{figure}[htb]
\centering
\begin{subfigure}[t]{0.35\textwidth}
  \centering
   \includegraphics[width=\linewidth]{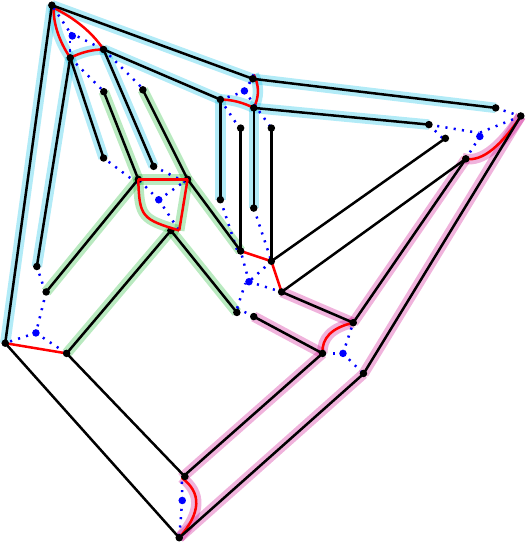}
   \caption{}
\end{subfigure}
\hspace{.4cm}
\begin{subfigure}[t]{0.35\textwidth}
\centering
\includegraphics[width=\linewidth]{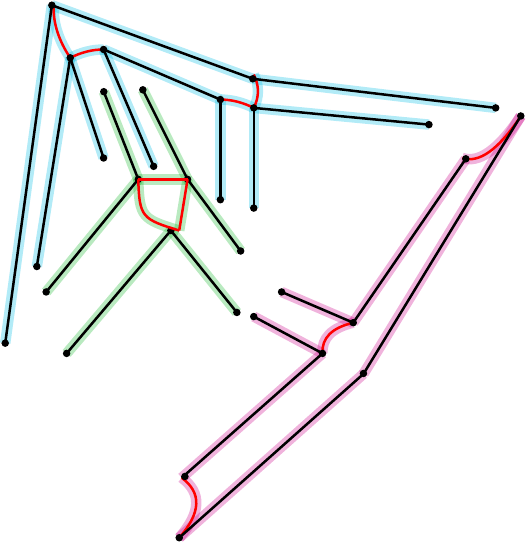}
   \caption{}
\end{subfigure}
\hspace{.4cm}
\begin{subfigure}[t]{0.2\textwidth}
\centering
\includegraphics[width=\linewidth]{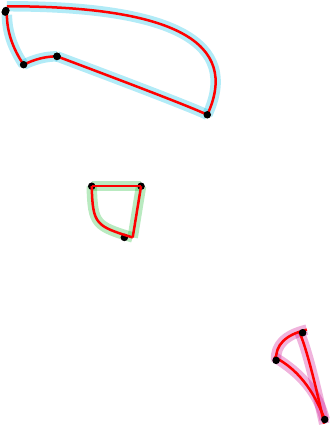}
   \caption{}
\end{subfigure}
\caption{(a) The graph $\hat{G}_\ell'$.
(b) The graph $\hat{G}_\ell'$ with only components $\hat{Q}$ that correspond to paths $Q$ obtained from the edge pairing. Note that they are vertex-disjoint.
(c) The cycles $\mathcal{Q}$
after contracting $\ER$ edges of subgraphs $\hat{Q}$. Note, each cycle $\mathcal{Q}$ corresponds to the $\EC$ edges of a certain component $\hat{Q}$, and to a cycle $Q$ (see \cref{fig: edge_pairing})
\label{fig: hatG_prime}
}
\end{figure}

    We describe the change on $\hat{G}_\ell$ formally (recall the definition of $\hat{G}_\ell$ in \cref{par: def_hat_G}). If $\hat{p}(q)_i, \hat{p}(q)_{i+1}$ do not share a vertex, then they already are vertex-disjoint and are taken care of. So in the rest of the proof we assume they share an endpoint (that is the only way they can intersect with vertices by \cref{lem: pairing_partition_pf}). Such a vertex connects the last edge of $\hat{p}(q)_i$ with the first edge of $\hat{p}(q)_{i+1}$. Those edges are $\ER$ edges ($\hat{p}(q)$ consists of $\ER$ edges by definition of $\hat{G}_\ell$).
    Concretely, the $\ER$ copy of $last(q)_i$ intersects the $\ER$ copy of $first(q)_{i+1}$ with a vertex;  where, $last(q)_i$ is the second edge in the $i$'th paired edges of $q$ in $H$. Analogously, $first(q)_{i+1}$ is the first edge of the $i+1$'th paired edges of $q$ in $H$.
    The vertex that is shared to $last(q)_i$ and $first(q)_{i+1}$ is a vertex of the type  $v_k\in D(v)$ for some $v\in G$, as those are the only vertices that are incident to two $\ER$ edges (by \cref{par: def_hat_G}). Then, expanding $v_k$ to two vertices, one for each of $\hat{p}(q)_i,\hat{p}(q)_{i+1}$,  makes them vertex disjoint. See \cref{fig: hatG_prime} for $\hat{G}_\ell'$, and \cref{fig: edge_pairing} for $\hat{G}_\ell$ and the paths $Q$.
    
    Concretely, aside from the $\ER$ edges corresponding to  $last(q)_i,first(q)_{i+1}$, $v_k$ is also incident to $\ES$ edges, and to up to two $\EC$ edges. One of the $\EC$ edges simulates $last(q)_i$ and the other simulates $first(q)_{i+1}$. Note, either $\EC$ edge may or may not be incident to $v_k$ (by \cref{par: def_hat_G}). In particular,  an $\EC$ edge is incident to a copy of an arbitrary endpoint of its corresponding primal edge.
    E.g. if $last(q)_i$ maps to the primal edge $(u,v)$, then the $\EC$ copy of $last(q)_i$ may be incident to some copy $u_j$ of $u$. 
    Formally, the simple change we perform to make $\hat{p}(q)_i, \hat{p}(q)_{i+1}$ vertex disjoint is on $v_k$: Expand the vertex $v_k$ into two vertices $v_k^0,v_k^1$, where $v_k^0$ is incident to the $\ER,\EC$ edges of $v_k$ that map to $last(q)_i$, and similarly, $v_k^1$ is incident to $\EC,\ER$ edges of $v_k$ that map to $first(q)_{i+1}$.
    $\ES$ edges that were incident to $v_k$ get incident to an arbitrary vertex of the two.
    Finally, we connect the newly created vertices with an $\ES$ edge $(v_k^0,v_k^1)$ to keep the graph connected with small diameter (contracting this edge reverses the change that we just did on $v_k$). 
    
    Note, changes related to $\EC,\ES$ edges incident to {\em all} copies $v_k\in D(v)$ of $v\in G$ can be simulated locally by $v$, because those types of edges connect only copies of $v$ itself in $\hat{G}_\ell$. The change on $\ER$ edges can be coordinated in one round of communication on $G$ with their other (primal) endpoints.
    That is because $v$ is a neighbor of the other endpoints in $G$ of (edges corresponding to) $\ER$ edges incident to $v$'s copies (by definition of $\hat{G}_\ell$, \cref{par: def_hat_G}).
    We do this change on $\hat{G}_\ell$ for all $\hat{p}(q)_i, \hat{p}(q)_{i+1}$, and for all nodes $q$ of level $\ell$ in one round of communication on $G$. Note, the only knowledge required is, for $v_k$ to know that it is a shared endpoint of $\hat{p}(q)_i, \hat{p}(q)_{i+1}$, which can be easily deduced by the numbers assigned to its incident $\ER$ edges that are given as input. 
    That is, $v_k$ checks if one of those $\ER$ edges is assigned an odd number and the other is assigned its successor (then those edges belong to different paths $\hat{p}(q)_i, \hat{p}(q)_{i+1}$).

    \medskip
    \noindent
    {\bf Properties of $\hat{G}_\ell'$.}
    \begin{claim}
        \label{claim: simulation_hatG_prime_ell}
            $\hat{G}_\ell'$ is a planar graph of diameter $O(D)$, and a $\CONGEST$ round can be simulated in it within $O(1)$ rounds on $G$.
    \end{claim}
    \begin{proof}     
    Note first that adding the new $\ES$ edges $(v_k^0,v_k^1)$ keeps $\hat{G}_\ell'$ connected of twice the diameter of $\hat{G}_\ell$ at most, which remains $O(D)$ (see \cref{lemma: diameter_hat_G}). In   addition, since the communication on $\ES,\EC$ edges can be simulated locally, then the only edges that need to simulated are $\ER$ edges whose number did not change, hence, like $\hat{G}_\ell$, an edge of $G$ simulates at most four $\ER$ edges (see \cref{lemma: congest_hat_G}), and $\CONGEST$ round on $\hat{G}_\ell'$ is simulated within $O(1)$ rounds on $G$. Obviously, expanding a vertex as we do does not violate planarity.
    \end{proof}

    \begin{remark}
    \label{remark: correspondence_hatg_prime}
     Note that there is a bijection between edges of $\hat{G}_\ell[\ER\cup\EC]$ and edges of  $\hat{G}_\ell'[\ER\cup\EC]$. Thus, we abuse notation and refer to a subgraph of $\hat{G}_\ell'[\ER\cup\EC]$ by its corresponding subgraph of $\hat{G}_\ell[\ER\cup\EC]$. E.g. we may refer to $\hat{Q}$, or $\hat{p}(q)_i$ in  $\hat{G}_\ell'[\ER\cup\EC]$. 
\end{remark}

    \begin{claim}
    \label{claim: edge_pairing_paths_vertexdisjoint}
    Let $Q$ be any path defined by the given edge pairing in any subgraph $H$ of any level $\ell$ annulus, then:
    \begin{enumerate}
    \item The paths $\hat{p}(q)_i, \hat{p}(q')_{k}$ (as defined by \cref{lem: pairing_partition_pf}) are vertex disjoint in $\hat{G}_\ell'[\ER\cup \EC]$ for any two nodes $q, q'\in H$, (except when $q=q'$ and $i=k$).
    In particular,  any two paths $Q'\neq Q$ defined by the given edge pairing,  the graphs $\hat{Q},\hat{Q'}$ in $\hat{G}_\ell'[\ER\cup \EC]$ are vertex-disjoint.
    
        \item 
    The subgraph $\hat{Q}$ of $\hat{G}_\ell'$ is connected, and has the same structure as in Claim~\ref{claim: edge_pairing_paths_edgedisjoint}.

        \end{enumerate}
    \end{claim}
    \begin{proof}
    The first property, follows 
    by the fact that $\hat{p}(q)_i, \hat{p}(q')_{k}$ are vertex-disjoint in $\hat{G}_\ell[\ER\cup \EC]$ by \cref{lemma: faceparts_hat_G} if $q\neq q'$.
    Otherwise, $q= q'$ and $i\neq k$, thus, by Claim~\ref{lem: pairing_partition_pf}, $\hat{p}(q)_i, \hat{p}(q)_{k}$ are edge disjoint.
    Then, the property follows by the change we do on $\hat{G}_\ell$. I.e., the fact that each vertex $v\in \hat{p}(q)_i\cap \hat{p}(q)_{k}$ in $\hat{G}_\ell$ was split into two vertices, one in $\hat{p}(q)_i$ and the other in $\hat{p}(q)_{k}$.
    The proof of  $\hat{Q}, \hat{Q'}$ for $Q\neq Q'$ being vertex-disjoint follows.
    As by Claim~\ref{claim: edge_pairing_paths_edgedisjoint}, $\hat{Q}, \hat{Q'}$ in $\hat{G}_\ell$ consist of distinct $\EC$ edges and distinct paths $\hat{p}(q)_i$. Then, since $\hat{Q}, \hat{Q'}$ are edge disjoint in $\hat{G}_\ell$, they cannot overlap in $\EC$ edges. I.e., for them to overlap in vertices they must overlap in $\hat{p}(q)_i\subseteq Q$ and $ \hat{p}(q')_{k}\subseteq Q'$, but we just proved that this is not possible for $\hat{p}(q)_i \neq  \hat{p}(q')_{k}$ in $\hat{G}_\ell'$. Note, $\hat{p}(q)_i \neq  \hat{p}(q')_{k}$ because each corresponds to a pair of edges in $Q,Q'$ and $Q,Q'$ are edge-disjoint (edge pairings produce edge-disjoint paths).

    The second property is correct in $\hat{G}_\ell[\ER\cup \EC]$ by Claim~\ref{claim: edge_pairing_paths_hatG} (where $\hat{Q}$ was defined). The change we perform on $\hat{G}_\ell$ to get $\hat{G}_\ell'$ only expands some vertices of $\hat{Q}$, but does not apply change the structure of $\hat{Q}$; in the sense that any edge in $\hat{G}_\ell$ has exactly one corresponding edge in $\hat{G}_\ell'$ (\cref{remark: correspondence_hatg_prime}), s.t., an edge in $\hat{G}_\ell$ belongs to $\hat{Q}$ iff its corresponding edge in $\hat{G}_\ell'$ also belongs to $\hat{Q}$. In particular, this is correct for $\EC,\ER$ edges by definition of $\hat{G}_\ell'$. I.e., when we expand a vertex in $\hat{G}_\ell$ that belongs to $\hat{Q}$, its copy in $\hat{G}_\ell'$ that belongs to$\hat{Q}$ gets incident to all edges that belonged to $\hat{Q}$ in $\hat{G}_\ell$.
   \end{proof}

   \medskip
    \noindent
    {\bf Orienting the paths.} 
    The above claim implies that orienting the  path $Q$ in $H$ can be easily simulated by computations on $\hat{Q}$ in $\hat{G}_\ell'$. That is, the first property (by definition of $\hat{Q}$ as in Claim~\ref{claim: edge_pairing_paths_edgedisjoint}) states that there is a 1-1 mapping between $\EC$ edges in $\hat{Q}$ and edges of $H$ in $Q$. Thus, orienting those edges consistently in $\hat{Q}$ is sufficient.
    In addition, the second property allows to compute such an orientation on all subgraphs $\hat{Q}$ simultaneously.
    Concretely, we prove a structural property related to $\hat{Q}$, which then would be used to compute the desired orientation.
    \begin{claim}
    \label{claim: mathcal_Q}
    Contracting $\ER$ edges of $\hat{Q}$ in $\hat{G}_\ell'$ results in a minor of it, denoted
            $\mathcal{Q}$. Such that,  $\mathcal{Q}$ contains all $\EC$ edges of $\hat{Q}$, and only them. Moreover, $\mathcal{Q}$ is a path (cycle) and has the exact same topology as $Q$ (see \cref{fig: hatG_prime}).
    \end{claim}
    \begin{proof}
    the proof follows from  Claim~\ref{claim: edge_pairing_paths_vertexdisjoint} above.
    I.e., by the second item of the claim, the path $Q=q_1,q_2,\ldots$ corresponds to a connected subgraph $\hat{Q}$ of $\hat{G}_\ell'[\ER\cup \EC]$, where $\hat{Q}$ has a structure as in Claim~\ref{claim: edge_pairing_paths_edgedisjoint}. In particular $\hat{Q}$ is the union of: the ($\ER$) subpath $\hat{p}(q_1)_{i_1}$ of $\hat{p}(q_1)$,     
     the $\EC$ edge mapping to edge $(q_1,q_2)$, 
        the ($\ER$) subpath $\hat{p}(q_2)_{i_2}$ of $\hat{p}(q_2)$, the $\EC$ edge mapping  to $(q_2,q_3)$, and so on.
    Then, by the first item of Claim~\ref{claim: edge_pairing_paths_vertexdisjoint}, we have that $\hat{p}(q_j)_{i_j}$s are vertex disjoint. Note, these are exactly the components whose edges are contracted in $\hat{Q}$ to obtain $\mathcal{Q}$. Henceforth, $\mathcal{Q}$ is exactly the path consisting of: the $\EC$ edge mapping to edge $(q_1,q_2)$, the $\EC$ edge mapping  to $(q_2,q_3)$, and so on.
    \end{proof}

    Henceforth, to orient edges of paths $Q$ in $H$, all we have to show is that this contraction can be performed (distributively) to obtain the paths $\mathcal{Q}$. Then, to orient all paths $\mathcal{Q}$ consistently and show that this indeed can be translated to an orientation of $Q$ in $H$ (i.e., edges of $H$ learn their orientation).

\medskip
\noindent
{\em Orienting the edges.}    
    We focus on orienting the $\EC$ edges of  $\hat{Q}$ in $\hat{G}_\ell'$. Afterwards we explain why we can parallelize this to all paths and all annuli of level $\ell$.
    We will work in the minor-aggregation model in $\hat{G}_\ell'$ this time.
    
    First,  all $\ER$ edges inside subgraphs $\hat{Q}$ contract themselves 
    (which is allowed by the model), producing the paths $\mathcal{Q}$. 
    Note those $\ER$ edges know themselves, as their endpoints in $G$ that simulate them know this information about them because they simulate their vertex copies on paths $\hat{p}(q)_j$.
    We then orient all edges of $\mathcal{Q}$, using a minor-aggregation procedure that runs in $\tilde{O}(1)$ rounds and orients edges towards an elected root of $\mathcal{Q}$ (e.g. procedure of orienting a tree \cref{lem: subtreesums} from the next section). 
    
    We elaborate.
    If  $\mathcal{Q}$ is a path then it is also a tree and this can be achieved if we elect one of its endpoints as a root. Otherwise, $\mathcal{Q}$ is a cycle and the elected node to be a root can  omit one of its incident edges, defining a spanning path of $\hat{Q}$ rooted with an endpoint, on which the lemma can be used. 
    Finally, leader election for $\mathcal{Q}$ can be done easily: If $\mathcal{Q}$ is a path then one of its endpoints informs all nodes of $\mathcal{Q}$ with this information and the maximal ID endpoint is elected. If $\mathcal{Q}$ is a cycle, then a root is elected by computing the maximal ID of a node in $\mathcal{Q}$. Each of both tasks can be performed in one minor-aggregation round by contracting $\mathcal{Q}$ to a single node, which then performs a consensus step to detect whether $\mathcal{Q}$ is a path (has endpoints), and to compute the maximum ID of nodes (or endpoints) of $\mathcal{Q}$.

\medskip
\noindent
{\em Correspondence to $H$.}    
        Note, by Claim~\ref{claim: mathcal_Q}, there is a bijection between edges of $\mathcal{Q}$   to edges of $Q$ and both paths has the exact topology. Obviously, the consistent orientation computed for $\mathcal{Q}$ of implies an consistent orientation for $Q$.  
        Thus, it remains to explain how the computed orientation to $\mathcal{Q}$ is leaned by the edges of $Q$ in  $H$. 
        Concretely,  by Claim~\ref{claim: mathcal_Q} there is a bijection  between edges of $\mathcal{Q}$   to $\EC$ edges of $\hat{Q}$, which is known to endpoints of $\EC$ edges in $\hat{G}_\ell'$.
        In addition, recall (\cref{remark: correspondence_hatg_prime}), there is a bijection between
        the $\EC$ edges of $\hat{Q}$ in $\hat{G}_\ell'$ and $\EC$ edges of $\hat{Q}$ in $\hat{G}_\ell$.
        Henceforth, since an $H$ edge is simulated by its corresponding $\EC$ edge in $\hat{G}_\ell$, then the orientation of an $\EC$ edge $(v_i,v_{i+1})$ from $v_i$ to $v_{i+1}$ in $\hat{G}_\ell$ is considered (by $v_i,v_{i+1}$) as orienting the edge's counterpart in $H$ from $g$ to $f$ where $v_i\in \hat{p}(g)$ and $v_{i+1}\in \hat{p}(f)$. This is enough as these are the vertices that simulate the edge $(g,f)$ in $H$. I.e., edges in $H$ know their direction.
 
\medskip
\noindent
{\em Simulation.}    
    The reason we can simulate a round of the minor aggregation model on $\hat{Q}$ in $\hat{G}_\ell'$ within near-optimal $\tilde{O}(D)$ rounds, is that connected subgraphs (in this case, $\hat{Q}$) of planar graphs whose diameter is $O(D)$ allow such simulation (Theorem 17 and Corollary 11  of~\cite{GZ22};~\cite{RozhonGHZL22_shortestpaths, GHSYZ22}). By Claim~\ref{claim: simulation_hatG_prime_ell}, $\hat{G}_\ell'$ is such a planar graph of diameter $O(D)$, and a $\CONGEST$ round on it is simulated by $O(1)$ rounds on $G$. I.e., this simulation is possible.
    Finally,  the parallelization to all $\hat{Q}$ for all paths $Q$ of all graphs $H$ of level $\ell$, is possible by Claim~\ref{claim: edge_pairing_paths_vertexdisjoint}. I.e., $\hat{Q}$ and $\hat{Q'}$ are connected (second item of the claim) and vertex disjoint for all $Q\neq Q'$ (first item of the claim) in level $\ell$, and running a minor-aggregation algorithm in multiple connected vertex-disjoint subgraphs in parallel can be simulated within the same round complexity of running the algorithm on the input graph (Corollary 11 of~\cite{GZ22};~\cite{RozhonGHZL22_shortestpaths, GHSYZ22}).

This concludes the proof of \cref{lem: orienting_paired_edges}.

\subsection{SSSP on Annuli}
\label{section: SSSP}
We explain how to put everything together in order to prove \cref{th: SSSP}.
That is, we need to explain how to combine: (1) Deactivating parallel edges, (2) Simulateing the approximate SSSP algorithm of~\cite{RozhonGHZL22_shortestpaths} on annuli (this includes the Eulerian orientation problem), and (3) Turning the approximate SSSP algorithm of~\cite{RozhonGHZL22_shortestpaths} into one that obeys the triangle inequality (as in \cref{def: approx_triangle_inequality}). 

\medskip
\noindent
{\bf 1. Dealing with multi-graphs.}
The graph $G^*$, as well as level-$\ell$ annuli, may be multi-graphs (i.e., have self-loops and parallel edges) and the algorithm of~\cite{RozhonGHZL22_shortestpaths} assumes the graph is simple.
To make the graph simple, we delete self-loops and parallel edges while preserving distances. I.e., we delete all edges connecting nodes $f,g$ except for the one with minimal weight. 
A simple minor-aggregation procedure that solves the problem was shown in~\cite{planardistributedmaxflow24}:

 \begin{restatable}[Deleting parallel edges, Lemma 4.15 of~\cite{planardistributedmaxflow24}]{lemma}{deactivateParallelEdges}
    \label{lem: deactivating_parallel_edges}
      There is a minor-aggregation algorithm that runs in $\tilde{O}(1)$-rounds on planar multi-graphs and deletes parallel edges. Upon termination, each edge is assigned a label of either "active" or "not-active". Each active edge $(f,g)$ is the minimal weight edge with endpoints $f,g$.
\end{restatable}

We can therefore delete parallel edges and self-loops in all annuli of any given level $\ell$ simultaneously within $\tilde{O}(D)$ rounds, by simulating the algorithm from \cref{lem: deactivating_parallel_edges} using \cref{thm: MA_simulation}. 
From now on, we assume that annuli are simple graphs. 

\medskip
\noindent
{\bf 2. Simulating the (approximate) SSSP algorithm of~\cite{RozhonGHZL22_shortestpaths}.}
Next, we prove that our solution to the Eulerian orientation problem (\cref{def: eulerian_oracle}) can be used in~\cite{RozhonGHZL22_shortestpaths}'s algorithm in order to compute (approximate) SSSP on annuli - without the approximate triangle inequality guarantee for now.
In particular, one way to phrase the result of~\cite{RozhonGHZL22_shortestpaths} is:

\begin{lemma}[\cite{RozhonGHZL22_shortestpaths}]
\label{lem: sssp_blackbox}
There is a algorithm that works in $\tilde{O}(1/\epsilon^2)$ rounds and computes an $(1+\epsilon)$ approximate SSSP tree from a given source $s$ in any undirected graph $G_{virt}$ with weight from the range $[1,n^{O(1)}]$.  Where $G_{virt}$ is obtained from an undirected graph $G$ by adding $\tilde{O}(1)$ arbitrarily-connected virtual nodes.
Upon termination, each node knows its (approximate) distance from $s$, and each edge knows whether it is in the (approximate) SSSP tree.
Each round of the algorithm is either a minor-aggregation round on {\boldmath$G$}, or a call to the Eulerian orientation oracle $\mathcal{O}_{Euler}$ on {\boldmath$G_{virt}$} (as in \cref{def: eulerian_oracle}). 
Moreover, the algorithm is deterministic, if the minor-aggregation model and the oracle $\mathcal{O}_{Euler}$ can be simulated deterministically.
\end{lemma}

Assume all level-$\ell$ incisions are distributively stored, and for each level-$\ell$ annulus $A_i$ the ID of a source node $f_i$ is known to all $v\in G$ where $v\in p(f_i)$ (i.e., is known to all node of the annulus $A_i$).
Then it is almost straightforward to compute approximate SSSP trees in all annuli of the level simultaneously.

Concretely, we can split time into phases of $D\log^c n$ for some constant $c$. Such that an even phase $2i$ simulates one round of the algorithm if that round is a minor-aggregation round. This simulation terminates within $D\log^{c_{2i}} n$ rounds on all annuli of level $\ell$ simultaneously by \cref{thm: MA_simulation}. Where $c_{2i}$ is some constant, and $c>\max_i \{c_{2i}\}$.
Similarly, an  odd phase $2i+1$ simulates one round of the algorithm if that round is a call to the Eulerian orientation oracle (\cref{def: eulerian_oracle}). Note that the input of the oracle in this case is given by the algorithm and is computed by the previously simulated rounds, and its output is given in the format required by~\cite{RozhonGHZL22_shortestpaths}. This is simply because we follow~\cite{RozhonGHZL22_shortestpaths}'s definition of the oracle.
This simulation terminates within $D\log^{c_{2i+1}} n$ on all annuli of level $\ell$ simultaneously by \cref{lem: distributed_eulerian_orientation}. Where $c_{2i+1}$ is some constant, and $c$ is chosen to be larger than $\max_i \{c_{2i+1}\}$.
Hence, the total runtime of simulating~\cite{RozhonGHZL22_shortestpaths}'s algorithm on all annuli of a given level is $\tilde{O}(D)$.
Note that this gives \cref{th: SSSP} except for the triangle inequality condition that we discuss next.

\medskip
\noindent
{\bf 3. Approximate triangle inequality.}
As explained in~\cref{sec: Refis_extensions} and \cref{appendix: centralized_missing_proofs}, we need the approximate SSSP algorithm to obey the approximate triangle inequality. 
This property requires that any subpath $f,\ldots,g$ of a $(1+\epsilon)$-approximate shortest path returned by the SSSP algorithm to be also a $(1+\epsilon)$-approximate shortest $f$-to-$g$ path (see \cref{def: approx_triangle_inequality}).
Let $\mathcal{A}$ be any $(1+\epsilon)$-approximate SSSP algorithm that supports adding a total of $\tilde{O}(1)$ arbitrarily connected virtual nodes to the input graph.
Recently,~\cite{RozhonHMGZ23_triangle_inequality} have shown a reduction that after $\tilde{O}(1)$ executions of $\mathcal{A}$ (each time on a different graph that is obtained from the input-graph by adding virtual nodes and changing the edge weights) gives a $(1+\epsilon\log n)$ approximate algorithm that obeys the triangle inequality.
The reduction is implementable in the minor-aggregation model as shown in~\cite{planardistributedmaxflow24}. 

\begin{lemma}[Theorem 5.1 of~\cite{RozhonHMGZ23_triangle_inequality} and   Theorem 1.3 of~\cite{planardistributedmaxflow24}]
\label{lem: triangle_inequality}
    Let $\mathcal{A}$ be an $r\cdot f(\epsilon)$ round $(1+\epsilon)$-approximate SSSP algorithm that supports adding $\tilde{O}(1)$ arbitrarily connected virtual nodes to the input-graph. There is a $(1+\epsilon')$ approximate SSSP algorithm $\mathcal{A}'$ that obeys the approximate triangle inequality, and runs in $\tilde{O}(r \cdot f(\epsilon'))$ rounds, where $\epsilon'=\epsilon/\log n$.
\end{lemma}

Note that adding virtual nodes to the input of the SSSP algorithm of~\cite{RozhonGHZL22_shortestpaths} is supported (see \cref{lem: sssp_blackbox}), in addition, our  simulation of the minor-aggregation model supports adding $\tilde{O}(1)$ virtual nodes to each annulus  (see the end of \cref{sec:Minor-aggregation}).
Hence,
we use the above \cref{lem: triangle_inequality} to convert the approximate SSSP algorithm of~\cite{RozhonGHZL22_shortestpaths} into one that, obeys the approximate triangle inequality and can be run on all annuli of a given level simultaneously.
I.e., in $\tilde{O}(D/\epsilon^2)$ rounds, we compute $(1+\epsilon)$-approximate shortest paths tree for all annuli of any given level $\ell$ simultaneously, concluding \cref{th: SSSP}.

\section{Distributed Reif Implementation}
\label{sec: implemetnation_details}
In this section we wrap up the entire (distributed) implementation of Reif's algorithm as summarized by the following \cref{alg: distributed_reif}. 
 
\begin{algorithm}[htb]
\caption{Distributed Reif}
\KwIn{A $D$-diameter positively weighted undirected planar network $G$, two vertices $s,t\in G$, and a precision parameter $\epsilon$}
\KwResult{Within $\tilde{O}(D)$-rounds, each vertex $v\in G$ learns its incident edges in a $(1+\epsilon)$-approximate minimum $st$-cut $P_i$}
\label{alg: distributed_reif}
\SetKwFunction{FindSPAnnulus}{ShortestCycle} 
\SetKwProg{Proc}{Procedure}{}{}

Construct $G^*$, and detect nodes $f\in s^*, g\in t^*$ \label{line: construct_G*}

Compute an $f$-to-$g$ $(1+\epsilon)$-approximate shortest path $P$ \label{line: find_P}

Incise along $P$ (results in annulus $A$) \label{line: incise_P}

\FindSPAnnulus{$A$, $0$, $|P|-1$}

\Return the $st$-cut edges, bisection, and length of the shortest $P_i$ found by \FindSPAnnulus \label{line: find_cut}

\medskip 

\Proc{\FindSPAnnulus{$A$, $j$, $k$}:}{

{\bf if} $k\leq j+1$  {\bf then} compute $P_j,P_k$ \label{line: leaf_annulus}
\tcc*[f]{leaf annuli}

Compute a $(1+\epsilon)$-approximate shortest path $P_i$  for $i=\lfloor (j+k)/2\rfloor$\label{line: find_Pi}

\tcc*[f]{overlapping annuli}

\For{$(P_a,P_b)\in \{(P_j,P_i),(P_i,P_k)\}$}{\label{line: start_overlapping_annuli}
\If{$P_a\cap P_b\neq \emptyset$
}{
Compute a $(1+\epsilon)$-approximate SSSP tree $T$ from $x\in P_a\cap P_b$ 

Compute the $p_t^0$-to-$p_t^1$ paths $P_t$ in $T$ for all ${a\leq t\leq b}$)

}}\label{line: end_overlapping_annuli}

\tcc*[f]{define next level's annuli}

Find first and last\label{line: intersection_P_Pi_line1} intersections $p_{i_1},p_{i_2}\in P_i\cap P$

Find the maximal\label{line: intersection_P_Pi_line2} subpath $P_i'$ of $P_i$, which is internally-disjoint from (copies) of $P$

Incise $P_i'$ giving the annuli $A_l,A_r$ of the next level \label{line: start_next_level} and discard non-annuli graphs 

\tcc*[f]{recursive calls (if any)}

{\bf if} $A_l$ exists  {\bf then} \FindSPAnnulus{$A_l$, $j$, $i_1$}

{\bf if} $A_r$ exists  {\bf then} \FindSPAnnulus{$A_r$, $i_2$, $k$} \label{line: end_next_level}

}

\end{algorithm}

In the first three lines of \cref{alg: distributed_reif}, we need to construct $G^*$, detect nodes $f,g$, find  an approximate shortest $f$-to-$g$ path $P$, and incise along $P$. 
In Line~\ref{line: construct_G*} we construct $G^*$ (represented by $\hat{G}_0$)  in $\tilde{O}(D)$ rounds as in \cref{lem: distributed_knowledge}. In particular, nodes $h$ of $G^*$ are identified and assigned IDs that are known to all primal vertices on the face $p(h)$ in $G$, and for each edge $e\in G$ its endpoints know the IDs of the endpoints of $e^*\in G^*$.  
Since $s$ (resp. $t$) knows the faces it participates in, then $s$ (resp. $t$) chooses locally an arbitrary such face to be $f$ (resp. $g$), and  broadcasts the ID of $f$ (resp. $g$) to the entire graph $G$ in $O(D)$ rounds using the following textbook lemma:
\begin{lemma}[Basic aggregation~\cite{peleg-book}]
\label{lem: global_broadcast}
    For every $v\in G$ let $x_v$ be an arbitrary $O(\log n)$-bit input. Within $O(D)$ rounds, all $v\in G$ learn the result $\bigoplus_{v\in G} x_v$ of any predefined aggregate operator $\oplus$.
\end{lemma}

Once all vertices of $G$ know the IDs of $f,g$, we have the required input format for \cref{th: SSSP}, and we can compute a $(1+\epsilon)$-approximate SSSP tree rooted at $f$ (Line~\ref{line: find_P}) in $\tilde{O}(D)$ rounds. 
After computing the tree, using  \cref{thm: MA_simulation} we can in $\tilde{O}(D)$ rounds
simulate the following classic $\tilde{O}(1)$ minor-aggregation rounds procedure
in order to mark the $f$-to-$g$ path in the tree as $P$. 

\begin{lemma}[Tree primitives]
\label{lem: subtreesums}
    There are $\tilde{O}(1)$ minor-aggregation rounds procedures $\mathcal{P}_1,\mathcal{P}_2$ s.t given any tree $T$ and nodes $f,g\in T$:
  $\mathcal{P}_1$
     computes subtree sums on $T$ w.r.t. $f$ as the root over arbitrary $\tilde{O}(1)$-bit strings as an input.
   $\mathcal{P}_2$ marks the edges and nodes of the $f$-to-$g$ path in $T$.
\end{lemma}
\begin{proof}
    The procedure $\mathcal{P}_1$ is Lemma 16 of~~\cite{GZ22}, also appears in~\cite{GhaffariH16_shortcuts, GP17, deterministic_sep, OurSIROCCO, RozhonGHZL22_shortestpaths}.
    The procedure $\mathcal{P}_2$ exists in many works (e.g.~\cite{deterministic_sep, OurSIROCCO, GP17}), however, not exactly as we state it. For completeness, we restate the proof:
    Simply, $g,f$ are assigned an input value of one. All other nodes are assigned an input value of zero. Then, compute subtree sums with $\mathcal{P}_1$. Then, an edges has one endpoint of sum  exactly one, and the other endpoint of sum at least one, iff the edge is on the $f$-to-$g$ path. I.e., all internal nodes of the path has sum exactly one, except for the lowest common ancestor of $f,g$, has a sum of two.
\end{proof}

After using the above lemma, edges of $G$ whose duals are in $P$ are known to their endpoints in $G$, and all vertices $v\in G$ know for each $h\in G^*$ s.t $v\in p(h)$ whether $h\in P$. Then, with additional $O(D)$ rounds, all vertices of $G$ learn the (weighted) length of $P$: By \cref{th: SSSP}, vertices $v\in p(g)$ know the (approximate) $f$-to-$g$ distance. Thus, they broadcast it in $O(D)$ rounds to all vertices of $G$.

Next, we aim to enumerate $P$'s endpoints from zero to $|P|-1$. 
We use \cref{lem: subtreesums} to compute subtree sums on $P$ as the tree of input, where the root is $g$ and the input to all $P$'s nodes is one except for $f$ whose input is zero.
The procedure terminates in $\tilde{O}(D)$-rounds (by the minor-aggregation simulation, \cref{thm: MA_simulation}). 
After which, each vertex $v\in G$ that lies on a face $p(p_i)$ learns that this face corresponds to $p_i\in  P$.

In order to incise along $P$ (Line~\ref{line: incise_P}), we need to handle the endpoint-case edge (see \cref{def: incision_P}, \cref{remark: incision_P}). For $p_0$ (resp. $p_{|P|-1}$), we identify a non-$P$ edge incident to it and set it to be the endpoint-case edge denoted (resp. $e_{|P|-1}$). 
To find $e_0$, the vertex $s$ knows the ID of $p_0$ (since $s\in p(p_0)$) 
and its incident edges in $p(p_0)$. Moreover, $s$ knows its incident $P$-edges (if any). Thus, $s$ chooses an incident edge $e_0$ of $p(p_0)$ whose dual is not on $P$.
Note, such an edge always exists because $P$ cannot contain all edge of the face $s^*$ in $G^*$ (this would contradict $P$ being a simple path). 
In primal terms, the vertex $s$ is incident to at least one edge on $p(p_0)$ whose dual is not on $P$. Then, $e_0$ is set to be such an edge. A similar argument holds for $t$ and $e_{|P|-1}$.
Note that this is done locally in $s,t$. 

Overall, in $\tilde{O}(D)$ rounds we have computed $P$ and so the incision of  level $\ell=1$ of $\mathcal{T}$ is distributively stored . 
Thus, using \cref{lem: distributed_knowledge} $P$ is incised (represented with $\hat{G}_1$) and this level is learnt. In addition, by \cref{thm: MA_simulation} we can run minor-aggregation algorithms on it.

\subsection{The Procedure ShortestCycle}
We assume that level $\ell-1$ is constructed. That is: (1) paths that define annuli of level $\ell-1$ are distributively stored (\cref{def: distributed_storage_incisions}) and incised (meaning $\hat{G}_{\ell-1}$ was already constructed), (2) annuli of level $\ell-1$ are known to the vertices that simulate them in $\hat{G}_{\ell-1}$ (\cref{lem: distributed_knowledge}), (3) for every level $\ell-1$ annulus $A$ bounded by paths $P_j,P_k$, the edges of $P_j,P_k$ knows themselves. Moreover, let $j\leq j'\leq k'\leq k$ where we aim to compute the paths $P_t$ for $j'\leq t\leq k'$ in $A$ (and its subtree in $\mathcal{T}$), then both $j',k'$ are known to all $G$ vertices that {\em simulate} $A$. I.e., vertices that lie on $p(h)$ for some node-part $h\in A$. 
Then (Line \ref{line: find_Pi}) we compute the $p_i^0$-to-$p_i^1$ path $P_i$ for $i=\lfloor (j'+k')/2\rfloor$ inside the annulus $A$. Afterwards, we prepare for the next recursion by incising $P_i$ as in \cref{def: incision_Pi} (computing the child annuli of $A$ as in \cref{def: incision_decomposition}).
We describe this on a single annulus $A$. However, these steps can be applied on all annuli level $\ell-1$ simultaneously in the same round complexity (by \cref{thm: MA_simulation} and \cref{th: SSSP}).

\begin{remark}
In the following we may treat nodes and edges of annuli as computational entities, when that is done then it is either to compute a part-wise aggregation over the annuli (\cref{lem: PA_level_T}) or to run a minor-aggregation algorithm (\cref{sec:Minor-aggregation}, \cref{thm: MA_simulation}). Annuli however are not computational entities in $\CONGEST$, and we may only discuss $\CONGEST$ on vertices of $G$ and $\hat{G}_{\ell-1}$ (since $\CONGEST$ on $\hat{G}_{\ell-1}$ can be simulated without any asymptotic overhead by \cref{lemma: congest_hat_G}).
\end{remark}

\medskip
\noindent
{\bf Leaf annuli (Line~\ref{line: leaf_annulus}).}
If $k'\leq j'+1$, then we compute $P_{k'},P_{j'}$ and terminate the procedure for this annulus, as it is a leaf annulus  (see \cref{def: incision_decomposition}).
Both computations are done in the same manner of computing any $P_i$ (explained next). 
Finally, a special treatment is required for an annulus where  $P_i$ and $P_k$ or $P_j$ (or both) intersect (such an annulus might have only one child or also be a leaf annulus). 
This is discussed later in this section.

\medskip
\noindent
{\bf Computing \boldmath$P_i$ (Line~\ref{line: find_Pi}).}
Since all nodes of the annulus $A$ know $j'$ and $k'$, they locally compute $i:= \lfloor (j'+k')/2\rfloor$.
After which, the two nodes $p_i^0$ and $p_i^1$ know that the SSSP computation must originate in one of them. Note, the nodes do not know which of them is $p_i^0$ and which is $p_i^1$. However, they know that they are node-parts of $p_i\in P$, because the correspondence between node-parts of ancestor annuli is known to them by \cref{lem: distributed_knowledge}. This information is sufficient: the nodes $p_i^0,p_i^1$ aggregate their IDs on $A$ in one round of minor-aggregations to elect the one with the maximum ID to be the source of the SSSP computation. Then, using \cref{th: SSSP} we compute an approximate SSSP tree from $p_i^0$ (w.l.o.g). Then, using \cref{lem: subtreesums}: (1) we mark the tree $p_i^0$-to-$p_i^1$ path $P_i$  (its nodes and edges in $A$ know themselves), (2) the length of $P_i$ (sum of its edge weights) is known to all vertices $v\in G$ that simulate $A$.
Note, the SSSP algorithm terminates in $\tilde{O}(D)$ rounds on $G$. The subtree sums and aggregation procedures run in $\tilde{O}(1)$ minor aggregation rounds, so by \cref{thm: MA_simulation} they are simulated in $\tilde{O}(D)$ $\CONGEST$ rounds on $G$. 
To conclude, 
\begin{lemma}
\label{lem: leaf_paths}
Within $\tilde{O}(D)$ rounds, for each non-leaf annulus $A$ of level $\ell-1$ that aims to compute $P_t$ for $j'\leq t\leq k'$, the edges of $P_i$ know themselves, and its length is known to its nodes. If $A$ is a leaf annulus edges of both $P_{j'},P_{k'}$ know themselves, and their lengths are known to nodes that participate in them.
\end{lemma}

\noindent
{\bf Intersection of \boldmath$P_i$ and \boldmath$P_j,P_k$ (Lines~\ref{line: start_overlapping_annuli}-\ref{line: end_overlapping_annuli}).}
As explained in \cref{sec: Refis_extensions} (Claim~\ref{claim: Pi_Pk_Pj_intresection}), if $P_i$ intersects $P_j$ at some vertex $x$ (analogously, we handle the case that $P_i,P_k$ intersect) then we compute all the paths $P_m$ for $j\leq m \leq i$ in the current recursive call. All such paths $P_m$ are the concatenation of the $p_m^0$-to-$x$ and $x$-to-$p_m^1$ paths. 
Note that each node of the annulus knows whether it is on any of the paths $P_j,P_i$ and $P_k$. Thus, a node that is in $P_i\cap P_k$ or $P_i\cap P_j$ knows itself. Then, by aggregating IDs of those nodes in one minor-aggregation round ($\tilde{O}(D)$ round on $G$ by \cref{thm: MA_simulation}), we choose the node $x$ and an SSSP computation (within $\tilde{O}(D)$ rounds, by \cref{th: SSSP}) is performed from $x$. After that, all nodes $p_m^0,p_m^1$ for $j\leq m \leq i$ know their distance from $x$.
I.e., in $G$, then vertices on $p(p_m^0)$ and $p(p_m^1)$ know the $x$-to-$p_m^0$ and $x$-to-$p_m^1$ distances (respectively). Thus, copies of those vertices broadcast these two distances on $\hat{p}(p_m)$ in $\hat{G}_0$, to compute and learn length of $P_m$. 
This is done in $\tilde{O}(D)$ rounds by \cref{cor: PA_g_hat}.
Finally, the length of the minimal $P_m$ is computed by one minor-aggregation round on $A$ ($\tilde{O}(D)$ round on $G$, by \cref{thm: MA_simulation}). That path is then marked in $x$'s SSSP tree  in $\tilde{O}(D)$ rounds (by \cref{lem: subtreesums}, \cref{thm: MA_simulation}).

\begin{lemma}
\label{lem: overlapping_annuli_paths}
within $\tilde{O}(D)$ rounds, for each annulus $A$ defined by $P_j,P_k$: If $P_i$ intersects $P_j$ (resp. $P_k$), then all nodes on the minimal-weight $P_m$ for $j\leq m\leq i$ (resp. $i\leq m\leq k$), learn that and learn $P_m$'s weight. 
\end{lemma}

\noindent
{\bf Intersection of \boldmath$P_i$ and \boldmath$P$ (Lines~\ref{line: intersection_P_Pi_line1}-\ref{line: intersection_P_Pi_line2}).}
    We have two tasks.
First, we find the first node $p_{i_1}$ and the last node $p_{i_2}$ of $P$ whose node-parts participate in $P_i$. Then, in the recursion (by Claim~\ref{claim: Pi_P_intresection}) we can discard $P_t$ for all $i_1 < t< i_2$. Second, we need to find  the maximal subpath $P_i'$ of $P_i$ that is internally disjoint from $P$ because that is the path that needs to be incised in the next step of the algorithm (by \cref{def: incision_Pi} of incising $P_i$).

\medskip
\noindent
\underline{\em Finding $i_1,i_2$.}
We prove the following lemma.
\begin{lemma}
\label{lem: P_intersecting_paths}
Within $\tilde{O}(D)$ rounds, for all annuli $A$ of level $\ell-1$:
All nodes of $A$ that correspond to $p_t$ (and the vertices of $G$ simulating them) learn that the length of $P_t$ is not the minimal (i.e., is to an identity element) for $i_1< t < i_2$ where $i_1$ (resp. $i_2$) is the first (last) node of $P\cap P_i$ in $P$. Moreover, all nodes of $A$ know $i_1,i_2$.
\end{lemma}

    The indices ${i_1}$, ${i_2}$ 
    are found as follows. Note, the $p_j$-to-$p_k$ subpath $P'$ of $P$ has two copies in $A$ (as a result of incising $P$), s.t. an endpoint (say $p_i^0$) of $P_i$ lies on the first copy of $P'$ (denoted $P'^0$) and the second endpoint $p_i^1$ of $P_i$ lies on the second copy of $P'$ (denoted $P'^1$). Moreover, those copies of $P'$ are node and edge disjoint (as a result of incising $P$ at first). Thus, we identify nodes on $P'^0$ and $P'^1$. Then, we set $i_1$ as the smallest index s.t. $p_{i_1}^0\in P_i$ or $p_{i_1}^1\in P_i$, and $i_2$ is the largest index s.t. $p_{i_2}^0\in P_i$ or $p_{i_2}^1\in P_i$.    
    To do so, we compute a minimum and a maximum aggregate operator on each of the paths $P'^0,P'^1$, where the input of a node $p_t^0\in P'0$ (resp. $p_t^1\in P'1$) is $t$ if it participates in $P_i$ and an identity element otherwise. This gives the first and last intersection between $P'^0,P'0^1$ and $P_i$. Then, via aggregating the results on $A$, we find $i_1$ and $i_2$ (all nodes in $A$ learn those indices). Those nodes know themselves because they know to what nodes of $P$ do they correspond. 
    The aggregations are made using \cref{lem: PA_level_T} in $\tilde{O}(D)$ rounds.

  It remains to show how to identify the paths $P'^0$ and $P'^1$ (the copies of $P$ in $A$). We perform the following $O(1)$ minor-aggregation round procedure (by \cref{thm: MA_simulation}, this is $\tilde{O}(D)$ rounds on $G$). Simply, each edge of $A$ that maps to an edge of $P$ contracts itself (which is allowed by the model). This results in two distinct (super-)nodes, one corresponding to $P'^0$, and the other to $P'^1$. Since $P'^0,P'^1$ are node-disjoint, the resulting super-nodes are indeed distinct. 
    Each of those super nodes chooses an ID, say Id of the maximal node $p_m^0$ for $P'^0$ (analogously for $P'^1$).
    Simply $p_i^0$ informs all nodes on $P'^0$ that they are on $P'_0$ and similarly for $p_i^1$ and $P'^1$.
    Then, we rewind the contraction (return to annulus $A$), and we get that all vertices on $P'^0$ know a shared ID for $P'^0$ (similarly for $P'^1$), and that $P'^0,P'^1$ have different IDs (as nodes have distinct IDs, \cref{lem: distributed_knowledge}).

\medskip
\noindent
\underline{\em Finding $P_i'$.}
We prove the following lemma
\begin{lemma}
\label{lem: find_disjoint_incision_path}
For all annuli of level $\ell-1$ simultaneously, the maximal subpath $P_i'$ of $P_i$ that is internally disjoint from (copies of) $P$ can be found within $\tilde{O}(D)$ rounds.
\end{lemma}

For detecting $P_i'$, we simply detect the last node $p_x^0$ of $P^0$ in $P_i$, and the first $p_y^1$ node of $P'^1$ in $P_i$. Then, node on the $p_x^0$-to-$p_y^1$ subpath of $P_i$ learn that this is in fact the path $P_i'$.

First we root $P_i$ with its endpoint that belongs to $P'^0$, then $p_x^0$ is in fact the deepest node of $P'^0$ in (the rooted) $P_i$. To find this node we perform a subtree sum computation on $P_i$, where all nodes of $P'^0$ in $P_i$ have an input of one. Hence, $p_x^0$ is simply the node of $P'^0$ whose subtree sum is one.
Next, we omit the prefix of $P_i$ that contains higher nodes than $p_x^0$. This can be done locally. I.e., each edge whose both endpoints got a subtree sum of at least one, omit themselves.
Now, we are kept with a subpath $Q$ of $P_i$ whose one endpoint is $p_x^0$ (and is the only node in $Q\cap P'^0$). However, $Q$ may contain many nodes of $P'^1$, and it remains to find $p_y^1$.
We do that the same way we found $p_x^0$. That is, root $Q$ with its $P'^1$ endpoint, compute subtree sums, the deepest node of $P'^1$ in $Q$ is  $p_y^1$. After that, we omit the prefix of nodes higher than $p_y^1$, which results in $P_i'$.

\medskip
\noindent
{\bf Defining next level's annuli (Line~\ref{line: start_next_level}).}  
Before the recursive calls.  we need to incise on $P_i'$ and define the next level's annuli. We abuse notation and denote $P_i'$ by $P_i$. We incise edges of $P_i$ that were not incised previously (in case it intersects $P_k$ or $P_j$). After the incision, some resulting subgraphs may not be annuli (by \cref{cor: residual_nodeparts}, see also a discussion at the end of \cref{sec: IDs_to_face_parts}). In particular, each node may have many node-parts in a level, the complementing parts are not in any annulus). Thus, we need to detect those subgraphs (if any). Once a node-part leanrs it is in such a graph, it stops its algorithm. I.e., vertices of $\hat{G}_\ell$ that simulate such nodes stop simulating them, and are used exclusively for communication.

\begin{claim}
A subgraph resulting from incising $P_i$ is not an annulus of level $\ell$ iff it contains (at least) a non-$P_i$ edge in $P_k$ (resp. $P_j$) that is incident to some node $f\in P_i\cap P_k$ (resp. $P_i\cap P_j$). We refer to such edges as {\em bad}. \label{claim:bad edge} 
\end{claim}
\begin{proof}
We prove for $P_k$ (an analogous proof holds for $P_j$).
This case is actually implicit in \cref{sec: Refis_extensions}, where we define the graphs that are discarded from being annuli of the next level.
    We explain. 
    Let $A'$ be the subgraph of $A$ bounded by $P_i,P_k$. $A'$ is an annulus of the next level, iff $P_i\cap P_k$ is empty. Otherwise, after incising $P_i$, $A'$ can be (and might) broken into subgraphs, each associated with one maximal subpath of $P_k$ that is internally-disjoint from $P_i$ (such a subpath has at least one and possibly both it endpoints in $P_i$). 
    I.e., each subgraph that is not an annulus of the next level must have a node (an endpoint of its associated $P_k$ subpath) that has an incident edge of $P_k$ and another incident edge of $P_i$.
    In the other direction, if such a (bad) edge exists, then $P_i$ and $P_k$ intersect and  $A'$ results from the incision shall be discarded.
\end{proof}
Note that because we chose to not incise an edge twice, then edges in $P_i\cap P_k$ are not duplicated in $A'$ after the incision of $P_i$ (that is why $A'$ might consist of many disconnected graphs, we deal with all of them). The annulus defined by $P_i,P_j$ remains connected by $P_i$ (if $P_i,P_j$ do not intersect, otherwise, it gets discarded).

By the above Claim \ref{claim:bad edge}, to identify subgraphs that are not annuli in the next level, we can identify bad edges instead (and discard a subgraph if it contains a bad edge). 
We now show how to identify bad edges in $P_k$ (the case of $P_j$ is analogous).
First, nodes of $A$ mark themselves if they are in $P_i\cap P_k$. Then, edges of $P_k$ mark themselves if one of their endpoints is marked. 
We implement this in two minor-aggregation rounds. In the first, contraction and consensus steps are skipped. In the aggregation step, edges choose one of two values $\{P_i\}, \{P_k\}$ indicating if they are in $P_i,P_k$ or an identity element 
(if the edge is in neither or in both). Then, the aggregate operator that a node computes over its incident edges checks if there are at least two incident edges s.t. one is in $P_k$ and one is in $P_i$. If so, the node marks itself. After the aggregation terminates, edges learn whether one of their endpoints is marked in an additional minor-aggregation round and possibly marks itself as a bad edge accordingly (if it is in $P_k$ or $P_j$ and has a marked endpoint).

Finally, we incise $P_i$ and construct the resulting subgraphs (represented by $\hat{G}_\ell$) as in \cref{sec: distributed_reif}. Then, each edge in each resulting subgraph chooses a value of one if its corresponding edge in $A$ bad (and zero otherwise). This is known to the edges of the resulting subgraphs because the correspondence between ancestor annuli and $A$ is known to $A$'s edges and nodes by the vertices  that simulate them. I.e., endpoints $u,v$ of edge $e\in G$ simulate all of $e$'s duplicates in all levels, thus, any information about any duplicate of $e$ in any previous level is known to $u,v$, and thus, is known to the edge's duplicates in $A$. 
Then, an aggregation (OR operator) is computed on each subgraph, detecting whether there exists an edge that chose one as an input.
Finally, each (node in a) graph in which a bad edge was detected, terminates its algorithm.

\subsection{Finding the Minimum $st$-cut and Flow Value}
Recall (\cref{sec:Centralized Reif}) that the minimum $st$-cut consists of the edges of $G$ whose duals participate in the cycle  of $G^*$ that corresponds to the $P_m$ ($0\le m \le |P|-1$) of minimal length.
We first discuss the approximation guarantees regarding the value of the maximum $st$-flow and minimum $st$-cut. Then discuss finding the desired path $P_m$ that admits this. Finally, we discuss a minor problem and its solution, concerning $P_m$ is an approximate path.

    \medskip
    \noindent
 {\bf Approximation guarantee and flow value.}
 Since, the dependency on $\epsilon$ is quadratic in the SSSP algorithm that we use \cref{th: SSSP} to find $P_m$, then if $\epsilon$ is set to $\tilde{O}(1)^{-1}$, we get  a $1+o(1)$ approximation to the minimum $st$-cut in $\tilde{O}(D)$ rounds (shortly we explain how to find $P_m$).
    In addition, a $1-o(1)$ approximation to the maximum $st$-flow value 
    is obtained simply by dividing the value of the minimum $st$-cut by $(1+\epsilon)$. This gives the flow value within a multiplicative error in the range $[1/(1+\epsilon),1]$. Note however, another form of $1/(1+\epsilon)$ is $1/(1-(-\epsilon))$, which (when $|\epsilon|<1$) is $1-\epsilon+\epsilon^2-\epsilon^3\ldots$ by the geometric series expansion. Obviously, this is at least $1-\epsilon$. Hence the multiplicative error is in the range $[1-\epsilon,1]$ as promised.

    \medskip
    \noindent
 {\bf Finding \boldmath$P_m$ (cut edges, and bisection).}
 While in \cref{alg: Reif} and \cref{alg: extended_Reif} the procedure \FindSPAnnulus only returns this length, in \cref{alg: distributed_reif} we wish (Line \ref{line: find_cut}) to also compute the $st$-cut edges and the corresponding bisection. 
 To find the minimal length $P_{m}$, note first that leaf annuli divide the interval $P$ into intervals. Each  interval either maps to a specific leaf annulus (defined by paths $P_i,P_{i+1}$ and handled in \cref{lem: leaf_paths}) or to an interval $[j,k]$ for which the minimal $P_m$ for all $m\in [j,k]$ was computed in a non-recursive manner (and was handled in Lemmas~\ref{lem: overlapping_annuli_paths} and~\ref{lem: P_intersecting_paths}).
Henceforth, the nodes of every $P_i$ and its length are known to vertices of $G$ that simulate them.
To find the minimal length $P_m$, 
each vertex in $G$ aggregates the minimal length of a path $P_i$ that it knows along with the index $i$. The aggregation terminates in $O(D)$ rounds by \cref{lem: global_broadcast}. 
To mark the $st$-cut edges in $G$ and learn the corresponding bisection we do the following. 
First, each vertex of $G$ knows for each incident edge whether it (its dual) participates in $P_m$ or not. 
Thus, each vertex marks such incident edges, and the approximate minimum $st$-cut is distributively stored.
To compute the bisection, we simply omit the $st$-cut edges from $G$ and run a connectivity algorithm in $\tilde{O}(D)$ rounds~\cite{GP17, GhaffariH16_shortcuts} (the omitted edges still exist in the communication graph, but not in the input graph). Then, two connected components are detected. The one that contains $s$ is set to $S$ and the other to $V\setminus S$. Vertices can learn their side after $s$ and $t$ broadcast their component's ID in $O(D)$ rounds to all vertices of $G$ (using \cref{lem: global_broadcast}).

\medskip
\noindent
 {\bf Dealing with non-simple \boldmath$P_m$.}
Finally, there is a slight issue with this description. In particular $P_m$ might not correspond to a simple cycle in $G^*$, because it is an approximate path. I.e., it is simple in the annulus it was found in, but it may use copies of the same edge of $P$ in that annulus. To deal with this, we force $P_m$ to correspond to a simple cycle by removing such overlaps. That is, we compose a new cycle in $G^*$, whose weight is larger than that of $P_m$ only by a $(1+\epsilon)$ (by \cref{lem: approx_reif}, \cref{appendix: centralized_missing_proofs}) as follows: Let $P_m'$ be the maximal subpath of $P_m$ that is internally disjoint from (copies of) $P$ (as in \cref{lem: find_disjoint_incision_path}). $P_m'$ then corresponds to a simple path of $G^*$ whose endpoints $p_x,p_y\in P$. We take the union of $P_m'$ and the $p_x$-to-$p_y$ subpath of $P$ to be our cycle. The path $P_m'$ can be found by \cref{lem: find_disjoint_incision_path} within $\tilde{O}(D)$ rounds, while the  $p_x$-to-$p_y$ subpath of $P$ can be marked by \cref{lem: subtreesums} (in $\tilde{O}(D)$ rounds, by \cref{thm: MA_simulation}).

\bibliographystyle{plain}
\bibliography{main}
\appendix

\appendix
\crefalias{section}{appendix} 

\section{Deferred Proof from \cref{sec: Refis_extensions}}
\label{appendix: centralized_missing_proofs}

\subsection{Proof of \cref{lem: incise_edge_once}}
\InciseEdgeOnce*

\begin{proof}
We prove by induction. 
Obviously, the claim holds for level zero (the graph $G^*$) where no edge is incised. 
In addition, from the definition of an incision on $P$ (\cref{def: incision_P}), the claim is immediate for level one ($G^*$ after incising $P$). I.e., only edges of $G^*$ in $P$ are incised (once).

Let $e\in G^*$ s.t. it was incised at most once until level $\ell-1$, we show that the claim holds for $e$ in level $\ell$.
Let $A$ be defined by paths $P_k,P_j$, and assume $k+1<j$; otherwise $k+1=j$ and $A$ is a leaf annulus (and the claim holds by the hypothesis). I.e, we 
incise on $P_i$ ($i=\lfloor (j+k)/2\rfloor$) to define $A$'s child annuli in level $\ell$. We examine two cases related to $e$ and $A$.

Namely, if $e$ is in $A$ and was not incised before, then: It is either in $P_i$ and gets incised, or it is not and does not get incised. In either case $e$ is incised at most once until level $\ell$ (inclusive).
Otherwise, $e$ was incised (only) once before, and a duplicate of it is in $A$. If $e$'s duplicate is not in $P_i$, then nothing changes and the claim still holds. Otherwise, $e$'s duplicate is in $P_i$. Since $e$ was incised before, then its duplicate in $A$ must be in $P$, $P_i$ or $P_j$. In this case, (the duplicate of) $e$ is not incised, because common subpaths of $P_i$ and $P$, $P_k$ or $P_j$ do not get incised (\cref{def: incision_Pi}).
\end{proof}

\subsection{Proof of Lemma~\ref{lem: approx_reif}}
\lemApproximateReif*

The proof of \cref{lem: approx_reif} relies on the fact that any subpath $P_i[u,v]$ of an approximate path $P_i$ returned by the SSSP algorithm is also an approximation (with the same approximation ratio) to the $u$-to-$v$ shortest path, if the SSSP algorithm obeys the approximate triangle inequality: 
\begin{definition}[Approximate triangle inequality]
    \label{def: approx_triangle_inequality}
Let $s$ be any source vertex. Any $\alpha$-multiplicative approximate distances $d_s^\alpha(\cdot)$ from $s$ obey the approximate triangle inequality if for any two vertices $u,v$ we are guaranteed that: 
  $$ |d^{\alpha}_s(v)-d^{\alpha}_s(u)|\leq \alpha\cdot d(u,v)  $$  
\end{definition}

We begin the proof of \cref{lem: approx_reif} with the following claim. 
\begin{claim}
\label{claim: approx_reif}
    Let $P_i$ be a $(1+\epsilon)$-approximate shortest $p_i^0$-to-$p_i^1$ path that obeys the approximate triangle inequality.
    For $j \neq i$, there is a $(1+\delta)(1+\epsilon)$ approximate shortest $p_j^0$-to-$p_j^1$ path $P_j$ that does not cross $P_i$, where $\delta, \epsilon>0$ are arbitrary constants.
\end{claim}

\begin{proof}
Let $P_j'$ be a $(1+\delta)$-approximate shortest $p^0_j$-to-$p^1_j$ path. Assume w.l.o.g that $p^0_j$ and $p^1_j$ are in the interior of $P_i$. 
If $P_j'$ crosses $P_i$, then it must cross it an even number of times, say $2\ell$.
If $\ell=0$ then $P_j'$ does not cross $P_i$ and we are done.
Otherwise, there is an ordered (by distance to $p^0_i$) set of vertices $(v_0,v_1,\ldots,v_{2\ell-1})$ in which $P_j'$ crosses $P_i$, such that, the $P_j'[v_{2k-1},v_{2k}]$ subpaths are  enclosed by $P_i$, and the $P_j'[v_{2k},v_{2k+1}]$ subpaths are not enclosed by $P_i$. 
Meanwhile, since $P_i$ obeys the approximate triangle inequality, the $P_i[v_{2k},v_{2k+1}]$ subpaths are $(1+\epsilon)$ approximation to the shortest $v_{2k}$-to-$v_{2k+1}$ paths. I.e., denote by $w(Q)$ the length of a path $Q$, then, by the approximate triangle inequality:
\[
w(P_i[v_{2k},v_{2k+1}])=
d_{p_i}^{1+\epsilon}(v_{2k+1}) - d_{p_i}^{1+\epsilon}(v_{2k+1}) \leq (1+\epsilon) d(v_{2k},v_{2k+1}) 
\]
Thus, we obtain $P_j$ by replacing the $P_j'[v_{2k},v_{2k+1}]$ subpaths with $P_i[v_{2k},v_{2k+1}]$, as follows:

\begin{equation*}
\begin{split}
w(P_j) & = \sum_{k=0}^{\ell -1} w(P_i[v_{2k},v_{2k+1}]) + \sum_{k=1}^{\ell -1} w(P'_j[v_{2k-1},v_{2k}]) \\ 
& \leq \sum_{k=0}^{\ell -1} (1+\epsilon)w(P'_j[v_{2k},v_{2k+1}]) + \sum_{k=1}^{\ell -1} w(P'_j[v_{2k-1},v_{2k}]) \\
& \leq \sum_{k=0}^{2\ell -2} (1+\epsilon)w(P'_j[v_{k},v_{k+1}]) = (1+\epsilon)w(P'_j)\\
& \leq (1+\epsilon)(1+\delta)d(p^0_j,p^1_j). 
\end{split} 
\qedhere  
\end{equation*}

\end{proof}

The proof of \cref{lem: approx_reif} follows easily from the above claim, by induction on the recursion level $i$ and using $\delta=\epsilon$. In the base case, $i=1$. 
   For the same reasoning of Claim \ref{claim: approx_reif}, in the first recursion level, we can compute $P$ to be a $1+\epsilon$ approximate shortest path, and then ensure that $P_{|P|/2}$ is a $(1+\epsilon)^2$ approximation. In level $1< i\leq O(\log n)$, applying Claim \ref{claim: approx_reif} to compute all paths of level $i$ guarantees an approximation equivalent to the approximation guarantee of level $i-1$ with a multiplicative factor of $(1+\delta)$. By the inductive assumption, the level $i$  shortest paths constitute an   $(1+\epsilon)^{(i+1)}$ approximation.
    Since there are $O(\log n)$ recursion levels, all we need to show is that for each constant $c_2$, there exists a constant $c_1$, s.t., $(1+\epsilon)^{O(\log n)}=(1+c_1\epsilon'/\log n)^{(c_2 \log n)}$ is at most $1+\epsilon'$. Choosing $c_1<0.5/c_2$, we get that for all $0<\epsilon'<1$:
    \[
    (1+c_1\epsilon'/\log n)^{(c_2 \log n)} \leq e^{c_1 c_2 \epsilon'} < e^{0.5 \epsilon'} \leq 1+\epsilon'.
    \]
    The last inequality is true since: (1) $1+\epsilon'$ is linear in $\epsilon'$, while $e^{0.5\epsilon'}$ is convex. (2) For $\epsilon'=0$ both functions are equal, and for $\epsilon'=1$ we have that $e^{0.5\epsilon'}=\sqrt{e}<2=1+\epsilon'$. This concludes the proof of \cref{lem: approx_reif}.\qedhere  

\end{document}